\documentclass[12pt,letterpaper]{article}

\usepackage{mathpazo}
\usepackage{times}
\usepackage{bbm}
\usepackage{latexsym}
\usepackage{mathtools}
\usepackage{amssymb}
\usepackage{float}
\usepackage{amsmath}
\usepackage{amssymb}
\usepackage{amsthm}
\usepackage{indentfirst}
\usepackage{graphicx}
\usepackage[utf8]{inputenc}
\usepackage{afterpage}
\usepackage{subcaption}
\usepackage[unicode]{hyperref}
\usepackage{epstopdf}
\usepackage{underlin}
\usepackage{fancyhdr}
\usepackage[toc,page]{appendix}
\usepackage{tikz}
\usepackage{tikz-cd} 
\usepackage{enumerate}
\usepackage{pdfpages}
\usepackage{pifont}
\usepackage{algorithm}
\usepackage{algpseudocode}
\usepackage{color}
\usepackage{pgfplots}
\usepackage{cite} 
\usepackage{fullpage}
\usepackage{xfrac}
\usepackage{csquotes}

\usepackage[margin=1in]{geometry}
\usepackage{setspace}

\newtheorem{theorem}{Theorem}
\newtheorem{lemma}{Lemma} 
\newtheorem{corollary}{Corollary}

\theoremstyle{definition}
\newtheorem{definition}{Definition}

\theoremstyle{remark}
\newtheorem{remark}{Remark} 

\newcommand{\abs}[1]{\ensuremath{\left\vert #1\right\vert}}
\newcommand{\norm}[1]{\ensuremath{\left\| #1\right\|}}

\newcommand{\set}[1]{\ensuremath{\left\{ #1\right\}}}

\newcommand*\boxast{\,
	\begin{tikzpicture}
	\node [rectangle,scale=.5,draw] {$\ast$};
	\end{tikzpicture}\,
}

\newcommand{\beq}{\begin{equation}}
\newcommand{\eeq}{\end{equation}}
\newcommand{\bea}{\begin{eqnarray}}
\newcommand{\eea}{\end{eqnarray}}
\newcommand{\bean}{\begin{eqnarray*}}
\newcommand{\eean}{\end{eqnarray*}}
\newcommand{\bcen}{\begin{center}}
\newcommand{\ecen}{\end{center}}
\newcommand{\bitm}{\begin{itemize}}
\newcommand{\eitm}{\end{itemize}}

\newtheorem{example}{\bf Example}
\def\sgcnv{\hbox{$\, \bigcirc \,$\kern-0.9em\hbox{\mgop}$\,$}} 
\def\supgeno{\hbox{$\, \bigcirc \,$\kern-1.0em\hbox{$\wedge$}$\,$}} 
\def\infgeno{\hbox{$\, \bigcirc \,$\kern-1.0em\hbox{$\vee$}$\,$}} 
\newcommand{\mgop}{\ensuremath{\star}}  

\hypersetup{
    citecolor=black,
    filecolor=black,
    linkcolor=green,
    urlcolor=black,
    linkbordercolor={0 1 0},
    citebordercolor={1 0 0}
}

\newcounter{enumrom}
\renewcommand{\theenumrom}{(\roman{enumrom})}

\makeatletter
\renewcommand{\@endtheorem}{\endtrivlist}
\makeatother

\makeatletter
\renewcommand{\thefigure}{{\bf \@arabic\c@figure}}
\renewcommand{\fnum@figure}{{\bf Figure}\,\thefigure}
\makeatother

\title{\Huge $\,$\\
Channel Coding at Low Capacity\footnote{A brief version of this paper was presented in 2019 IEEE Information Theory Workshop (ITW) \cite{ConfVersion}. }\\[1ex]}

\author{
Mohammad Fereydounian \\
   \small University of Pennsylvania \vspace*{-1ex}\\
   \small 3401 Walnut St, Philadelphia, PA\,19104 \vspace*{-1ex}\\
   \ttfamily\bfseries\small mferey@seas.upenn.edu\vspace*{3ex}\\
\and
{Hamed Hassani}\\
   \small University of Pennsylvania \vspace*{-1ex}\\
   \small 3401 Walnut St, Philadelphia, PA\,19104 \vspace*{-1ex}\\
   \ttfamily\bfseries\small hassani@seas.upenn.edu \vspace*{3ex}\\
\and
{Mohammad Vahid Jamali}\\
\small University of Michigan \vspace*{-1ex}\\
\small 1301 Beal Avenue, Ann Arbor, MI \,48109 \vspace*{-1ex}\\
\ttfamily\bfseries\small mvjamali@umich.edu \vspace*{3ex}\\
\and
{Hessam Mahdavifar}\\
   \small University of Michigan \vspace*{-1ex}\\
   \small 1301 Beal Avenue, Ann Arbor, MI \,48109 \vspace*{-1ex}\\
   \ttfamily\bfseries\small hessam@umich.edu \vspace*{3ex}\\
 
}

\date{}  

\begin{document}
\maketitle
\vspace*{-5ex}
\begin{abstract}
Low-capacity scenarios have become increasingly important in the technology of the Internet of Things (IoT) and the next generation of wireless networks. Such scenarios require efficient and reliable transmission over channels with an extremely small capacity. Within these constraints, the state-of-the-art coding techniques may not be directly applicable. Moreover, the prior work on the finite-length analysis of optimal channel coding provides inaccurate predictions of the limits in the low-capacity regime. In this paper, we study channel coding at low capacity from two perspectives: fundamental limits at finite length and code constructions. We first specify what a low-capacity regime means. We then characterize finite-length fundamental limits of channel coding in the low-capacity regime for various types of channels, including binary erasure channels (BECs), binary symmetric channels (BSCs), and additive white Gaussian noise (AWGN) channels. From the code construction perspective, we characterize the optimal number of repetitions for transmission over binary memoryless symmetric (BMS) channels, in terms of the code blocklength and the underlying channel capacity, such that the capacity loss due to the repetition is negligible. Furthermore, it is shown that capacity-achieving polar codes naturally adopt the aforementioned optimal number of repetitions. 

\end{abstract}

\section{Introduction} \label{sec:intro}

Error-correcting codes are often designed assuming an underlying channel with a certain capacity $C>0$. In order to understand how optimal the designed codes are, studying the finite-length fundamental limits becomes relevant, i.e., given a fixed block error probability $p_e$, what is the maximum achievable rate $R$ in terms of the blocklength $n$? There has been a large body of work in the past decade to study the fundamental limits of finite-length channel coding relating $p_e$, $R$, and $n$ together. This has been of interest to information theorists since the early years of information theory \cite{dobrushin1961mathematical,Strassen}, and a precise characterization is provided in \cite{polypaper} as $R = C - \sqrt{{V}/{n}} Q^{-1}(p_e) + \mathcal{O}\left({\log n}/{n}\right),$ where $C$ is the channel capacity, $Q(\cdot)$ is the tail probability of the standard normal distribution, and $V$ is a characteristic of the channel referred to as channel dispersion. In recent years, this finite-length analysis is further enhanced to include up to the third and later to the fourth order for particular channels including BEC, BSC, and AWGN (see \cite[Theorems~41,44]{polyphd}, \cite{TanAWGN,Erseghe2,moulin,ScarlettPaper,BetaBeta}).

In general, the fundamental question of what is achievable in the finite-length regime has been answered for various types of channels and up to several orders of approximation in the \textit{moderate-capacity} regime, where the higher order terms of approximating $R$ are significantly smaller than the first few terms. In this paper, we consider cases where the capacity $C$ is extremely small where the first-order (i.e., $C$) and/or the second-order terms are as small as the higher-order terms. In general, as we will see throughout this paper, designing optimal channel codes in such a \textit{low-capacity} regime (which we explicitly specify in Section~\ref{sec:limits}) and understanding how far they are from what is fundamentally achievable require addressing various theoretical and practical challenges.

From the code construction perspective, some of the state-of-the-art codes may not be directly applicable in the extremely low-rate regime. A notable instance is the class of iterative codes, e.g., turbo \cite{berrou1996near} or low-density parity-check (LDPC) codes \cite{GallagerLDPC,MacKay1999}. It is well known that decreasing the design rate of iterative codes results in denser decoding graphs which further leads to highly complex iterative decoders with poor performance. To circumvent this issue, the current practical designs use repetition coding. In particular, a low-rate repetition code is concatenated with a powerful moderate-rate code. Although repetition leads to efficient implementations, the rate loss through many repetitions may become significant. This implies that a comprehensive analysis is necessary to understand the optimality of coded repetition schemes in the low-capacity regime.

\subsection{Problem Motivation}

Low-capacity scenarios have become increasingly important in the technology of the Internet of Things (IoT) and the next generation of wireless networks. The Third Generation Partnership Project (3GPP) has introduced new features into the standard in order to integrate IoT into the cellular network. These new features, called Narrow-Band IoT (NB-IoT) and enhanced Machine-Type Communications (eMTC), were introduced in the release 13 of 3GPP standard and have been evoloving since then. The aim of these features is to enable deploying IoT in cellular networks where a massive number of users need to be served \cite{ratasuk2016overview}. From the channel modeling perspective, it turns out that users operating in these modes typically experience very low signal-to-noise ratios (SNRs). In particular, to ensure high coverage, the standard has to support coupling losses as large as $170$ dB for these applications, which is approximately $20$ dB higher than that of the legacy standard. Tolerating such coupling losses requires reliable detection for a typical $-13$ dB of effective SNR \cite{ratasuk2016overview,ratasuk2016nb}, translated to capacity $\approx 0.03$ bits/transmission. To enable reliable communications in such low-SNR regimes, the standard has adopted a legacy turbo code of moderate rate, i.e., rate $1/3$, in eMTC and NB-IoT (uplink) as the mother code together with many repetitions. The standard allows up to 2048 repetitions to enable the maximum coverage requirements, thereby supporting effective code rates as low as $1.6 \times 10^{-4}$ \cite{ratasuk2016overview}. However, as mentioned earlier, such repetition schemes may result in a significant rate loss. In general, studying finite-length fundamental limits as well as designing practical code constructions are necessitated to address the challenges of wireless system design for such emerging applications.

Communication in low-capacity regimes is also relevant in deep-space communication. In addition to the limited capacity, deep-space communication also suffers from catastrophic link loss and
severe signal attenuation. Hence, sophisticated code concatenation designs are often required in order to combat these design challenges. An overview of code designs adopted for various historical deep-space missions can be found in \cite{andrews2007development}. Designing efficient coding techniques to enhance the performance of deep-space communication is still an ongoing and open area of research that necessitates further attention given the importance of the targeted applications \cite{liang2020raptor}.

\subsection{Related Work}\label{RelW}

Following the earlier work of Polyanskiy et al. \cite{polypaper}, fundamental limits at finite length were later studied for various other types of channels beyond BECs, BSCs, and AWGN channels, including block-fading channels \cite{yang2013block,yang2015fading}, multiple-antenna channels \cite{yang2014quasi}, and multiple access channels \cite{huang2012finite, tan2013dispersions}. This has also motivated studying finite-length analysis in other related settings including lossy compression \cite{kostina2012fixed}, Slepian-Wolf coding \cite{tan2013dispersions,nomura2014second}, covert communications \cite{wang2016fundamental,tahmasbi2018first}, and coding with side-information \cite{watanabe2015nonasymptotic}, among others.


Another line of work in the literature is concerned with the application of saddlepoint approximations to efficiently compute rather complicated expressions such as random-coding union bound \cite{martinez2011saddlepoint,scarlett2014saddlepoint,ScarlettPaper,font2018saddle,altuug2020exact,honda2018exact}.
To this end, \cite{martinez2011saddlepoint} derived saddlepoint approximations of random-coding bounds to the decoding error probability with maximum-metric mismatched decoders allowing for accurate and simple numerical evaluations.
In \cite{scarlett2014saddlepoint}, a single-letter saddlepoint approximation, that is shown to be asymptotically tight for both fixed and varying rates, is presented for random-coding union bound of Polyanskiy \textit{et al.} \cite{polypaper} for i.i.d. random coding over discrete memoryless channels. Moreover, saddlepoint approximations of the meta-converse (hypothesis-testing) lower bound and random-coding union upper bound of channel coding minimum error probability are derived in \cite{font2018saddle} for symmetric memoryless channels in a wide range of system parameters.


In a related line of work, \textit{very noisy channels} (VNCs) are defined and studied. The notion of VNCs was first defined by Reiffen in \cite{reiffen1963note} by specifying certain conditions on the channel transition probability. Later, Gallager computed exponent-rate functions for random coding and convolutional codes in \cite{gallager1965simple}, and Majani \cite{majani1988model} carried out a comprehensive study of VNCs. VNCs are also relevant in Poisson photon channels, modeling direct detection optical communication channels when they are approximated by binary-input binary-output
discrete memoryless channels \cite{wyner1988capacity}. Recently, Sakai \textit{et al.} \cite{sakai2020second} derived finite-length laws for channel coding over continuous-time Poisson channels. Also, very recently, Wagner \textit{et al.} established that feedback neither improves the second-order coding rate for very noisy discrete memoryless channels \cite{wagner2020new} nor their high-rate
error exponent or moderate deviations performance \cite{shende2017very}. Although the low-capacity setting in this paper shares similar motivations to that of VNCs, the characterization of the low-capacity regime for a channel is fundamentally different. To clarify this difference, note that there is no notion of blocklength in the formulation of VNCs. However, our definition of low-capacity channels directly relates the low-capacity regime to the blocklength. More precisely, according to what will be discussed in Section~\ref{sec:lowcap}, a channel with a fixed capacity $C$ may not be at low capacity for a given blocklength, but may fall in the low-capacity regime for a shorter blocklength. Therefore, a VNC may or may not be a low-capacity channel necessarily.

For code construction, the focus of this paper is on the  class of binary memoryless symmetric (BMS) channels. Asymptotically, state-of-the-art polar codes, introduced by Ar{\i}kan \cite{arikan2009channel}, are the first class of provably capacity-achieving codes with explicit constructions as well as low-complexity encoding and decoding. Furthermore, their construction method is rate-adaptive, allowing constructing codes of rate ${k}/{n}$ for $k=0,1,2\dots,n$, where $n$ is the block length. While this makes them a natural choice for low-capacity regimes, they have not been particularly studied in very low-rate regimes when the number of information bits $k$ is much smaller than $n$.

\subsection{Our Contributions}

In this paper, we provide a specific formulation of low-capacity regimes from a  finite-length analysis perspective. We then provide fundamental non-asymptotic laws of channel coding in the low-capacity regime for a diverse set of channels with practical significance, namely, BEC, BSC, and AWGN channels. We observe that channel variations in the low-capacity regime can be better approximated by different probabilistic laws rather than the ones used for channels with moderate capacities. In particular, for BEC channels, we show that the behavior of channel variations in the low-capacity regime can be better approximated through the Poisson convergence theorem that studies the occurrence of rare events. This is basically intuitive noticing that "non-erasure" in a BEC with a very small capacity is a rare event. This phenomenon in the low-capacity BEC changes the relative significance of order terms in the classical expansions of the best achievable rate \cite{polypaper} in such a way that the higher order terms are then comparable to the lower ones and hence leading to an imprecise approximation. Incorporating Poisson laws in this paper rather than the Gaussian laws used in the classical expansions, however, results in a more accurate approximation. 
The inaccuracy of the classical expansion in the low-capacity BSC is as well due to the  aforementioned change in the relative significance of the order terms but unlike BEC, the Poisson laws do not govern the behavior of the BSC in the low-capacity regime and are not applicable here. As our analysis shows, employing the Gaussian laws with a sharper analysis and more precise tail bounds (e.g., the bound obtained by Talagrand \cite{Talag}) can circumvent this issue and lead to a different and more accurate expansion. For an AWGN channel, it turns out that the low-capacity expansion can be seen as a term-by-term limit of the existing expansion in the moderate-capacity regime. Proving this observation in Section~\ref{sec:awgn} is our contribution to the AWGN channel case. This leaves no necessity for a numerical evaluation of the AWGN case.

From the code construction perspective, assuming transmission over a BMS channel, repetition is often considered as a straightforward method to design practical low-rate binary codes that utilize the power of state-of-the-art binary code designs at a moderate rate while keeping the complexity low. This is mainly due to the fact that the encoding/decoding complexity of a concatenation scheme with inner repetition is effectively reduced to that of the outer code with a significantly shorter length. Thus, a major question is how large the number of repetitions can be such that the capacity loss due to the repetition is negligible? To answer this fundamental question, we characterize the optimal number of repetitions, in terms of the code blocklength and the underlying channel capacity. As mentioned earlier, polar codes are very appealing for low-capacity regimes due to their rate-adaptive structure. In this regard, we prove that the polarization transform implicitly induces the aforementioned optimal number of repetitions that we characterize in the low-capacity regime. This means the resulting low-rate polar codes naturally adopt the optimal number of repetitions. 

Our approximations of the fundamental limits in the low-capacity regime for the BEC and BSC cases are numerically evaluated and compared with the most well-known moderate-capacity estimation \cite{polypaper} and an all-rate estimation known as the saddlepoint approximation \cite{scarlett2014saddlepoint}. 
\vspace{-4mm}
\subsection{Content Organization}
The rest of this paper is organized as follows. In Section~\ref{sec:background}, we provide the necessary background. In Section~\ref{sec:limits}, we formally define the low-capacity regime and provide non-asymptotic laws of channel coding. Section~\ref{sec:design} is devoted to the code design. The numerical results are discussed in Section~\ref{sec:simulations}.  Section~\ref{sec:con} presents the conclusion and some future directions. Finally, all proofs of the theorems and their intermediate lemmas are provided in the Appendix which is presented in the supplementary material.


\section{Preliminaries} \label{sec:background}
\subsection{Finite-Length Analysis}
In this section, we will review the main concepts of channel coding in the finite-length regime, sometimes referred to as the non-asymptotic regime in the literature, along with a brief review of previous works.\footnote{For more details, we refer the reader to \cite{Tan} for an excellent review on this topic.}
For an input alphabet $\mathcal{X}$ and an output alphabet $\mathcal{Y}$, a channel $W$ can be defined as a conditional distribution on $\mathcal{Y}$ given $\mathcal{X}$. An $(M,p_e)$-code for the channel $W$ is characterized by a message set $\mathcal{M} = \{1,2,\cdots,M\}$, an encoding function $f_{enc}:\mathcal{M}\rightarrow\mathcal{X}$, and a decoding function $f_{dec}:\mathcal{Y}\rightarrow\mathcal{M}$ such that the \emph{average} probability of error does not exceed $p_e$, that is\footnote{In this paper, we only consider the average probability of error.  Similar results can be obtained for the maximum probability of error.}
\vspace{-1mm}
\begin{equation}
    \frac 1M \sum_{m\in\mathcal{M}}W\left(\mathcal{Y}\setminus f_{dec}^{-1}(m)\,\big| \,f_{enc}(m)\right)\le p_e.
    \vspace{-1mm}
\end{equation}
We consider $p_e$ to be a fixed given constant in $(0,1)$. Accordingly, an $(M,p_e)$-code for the channel $W$ over $n$ \emph{independent channel uses} can be defined by replacing $W$ with $W^n$ in the definition. The blocklength of the code is defined as the number of channel uses and is similarly denoted by $n$.  For the channel $W$, the maximum code size achievable with a given error probability $p_e$ and blocklength $n$ is denoted by
\vspace{-1mm}
\begin{equation}
    M^*(n,p_e)=\max\left\{M \mid\exists(M,p_e)\text{-code for } W^n  \right\}.
\end{equation}

\vspace{-1mm}
In this paper, we consider three classes of channels that vary in nature: 
\begin{itemize}
	\item BEC($ \epsilon $): binary erasure channel with erasure probability $\epsilon\in(0,1)$.
	\item BSC($\delta$): binary symmetric channel with crossover probability $\delta\in(0,1)$.
	\item AWGN($\eta$): additive white Gaussian noise channel with signal-to-noise ratio $\eta\in(0,\infty)$.  
\end{itemize}

Next, we mention the well-known finite-length expansion for a discrete memoryless channel. Due to \cite{Shannon,WolfSC}, we know that $\lim_{n\rightarrow\infty}  1/n \log_2 M^*(n,p_e)=C$.
Thus, the first order term in the finite-length expansion of $M^*(n,p_e)$ is $nC$.  The higher order terms can be written as (see \cite{polypaper,hayashi2009information})
\begin{equation} \label{poly-formula}
\log_2M^*(n,p_e)=nC-\sqrt{nV}Q^{-1}(p_e)+\mathcal{O}(\log n), 
\end{equation}
where $V$ is the channel dispersion and $Q^{-1}(\cdot)$ is the inverse of the so called Q-function that is $Q(\alpha)=\frac 1{\sqrt{2\pi}}\int_{\alpha}^\infty e^{-\frac{x^2}{2}}dx$.
The third and higher order terms, however, depend on the particular channel under discussion (see \cite[Theorem~41, 44, 73]{polyphd},  \cite{Erseghe2,moulin,ScarlettPaper,BetaBeta,TanAWGN}). More specifically, for BEC($\epsilon$), we have $C = 1-\epsilon$, $V = \epsilon (1-\epsilon)$ and the third order term is  $\mathcal{O}(1)$. For BSC($\delta)$, we have $C = 1 -h_2(\delta)$, $V = \delta(1-\delta) \log_2^2({(1-\delta)}/{\delta})$, and the third order term is $1/2 \log_2 n + \mathcal{O}(1)$. For AWGN($\eta$), we have $C = 1/2 \log_2(1+\eta)$, $V =  \eta(\eta+2)/(2(\eta+1)^2\ln^2 2)$, and the third order term is $1/2 \log_2 n + \mathcal{O}(1)$, as shown in \cite{TanAWGN}.

The formula \eqref{poly-formula} is basically obtained by approximating the Random Coding Union (RCU) and converse bounds using Gaussian laws. This approach is best known for moderate-rate (or equivalently moderate-capacity) scenarios. There are, however, other approaches in approximating the RCU and converse bounds which remain effective for a wider range of capacities such as the saddlepoint approximation \cite{martinez2011saddlepoint}. Here, we briefly mention the formulation\footnote{
we refer to \cite{ScarlettPaper,scarlett2014saddlepoint} for a thorough analysis of saddlepoint approximation.}. For a channel $W$ with input distribution $Q$ and a tuning factor $s>0$, the information density is defined as \vspace{-1mm}
\begin{equation}
i_{s}(x, y) = \log \frac{W(y \mid x)^{s}}{\sum_{\bar{x}} Q(\bar{x}) W(y \mid \bar{x})^{s}}\,\,,
\vspace{-1mm}
\end{equation}
where $x$ and $y$ represent the input and output of $W$. Then the saddlepoint approximation is\vspace{-1mm}
\begin{equation}\label{saddle}
\widehat{\operatorname{rcu}}_{s}^{*}(n, M) = \beta_{n}(Q, R, s)\, e^{-n E_{r}(Q, R)},
\vspace{-1mm}
\end{equation}
where $R$ is the rate, $n$ is the blocklength. The error exponent $E_{r}(Q, R)$ is defined as \vspace{-1mm}
\begin{equation}\label{errex}
E_{r}(Q, R) = \sup _{s>0,\, \rho \in[0,1]} -\log \mathbb{E}\left[e^{-\rho\, i_{s}(X, Y)}\right]-\rho R.
\vspace{-1mm}
\end{equation}
Moreover, the coefficient $\beta_{n}(Q, R, s)$ is called the sub-exponential prefactor. The computation of $\beta_{n}(Q, R, s)$ is quite complicated and needs further steps. See \cite{ScarlettPaper} for the introduction and approximation of  $\beta_{n}(Q, R, s)$ as well as relevant analysis for $E_{r}(Q, R)$.

\subsection{Polar Coding}

In Section\,\ref{sec:polar}, we study state-of-the-art polar codes at low capacity. Polar codes were introduced by Ar{\i}kan in  \cite{arikan2009channel}. They are the first family of codes for the class of binary-input symmetric discrete memoryless channels that are provable to be capacity-achieving with low encoding and decoding complexity~\cite{arikan2009channel}. Polar codes and polarization phenomenon have been successfully applied to a wide range of problems including data compression~\cite{Arikan2,abbe2011polarization}, broadcast channels~\cite{mondelli2015achieving,goela2015polar}, multiple access channels~\cite{STY,MELK}, physical layer security~\cite{MV,andersson2010nested}, and coded modulations \cite{mahdavifar2015polar}. 

The basis of channel polarization consists of mapping two identical copies of the channel $W: \mathcal{X}\to \mathcal{Y}$ into the pair of channels $W^0: \mathcal{X}\to \mathcal{Y}^2$ and $W^1:\mathcal{X}\to \mathcal{X}\times\mathcal{Y}^2$, defined as
\vspace{-5mm}
\begin{align}
\vspace{-10mm}
W^0(y_1, y_2\mid x_1) & = \sum_{x_2\in \mathcal X} \frac{1}{2}W(y_1\mid x_1 \oplus x_2) W(y_2\mid x_2),\label{eq:minus}\\
W^1(y_1, y_2, x_1\mid x_2) & = \frac{1}{2}W(y_1\mid x_1 \oplus x_2) W(y_2\mid x_2).\label{eq:plus}
\vspace{-10mm}
\end{align}
Then, $W^0$ is a worse channel in the sense that it is degraded with respect to $W$, hence it is less reliable than $W$; and $W^1$ is a better channel in the sense that it is upgraded with respect to $W$, hence it is more reliable than $W$. The operation in \eqref{eq:minus} is also known as the \emph{check} or \emph{minus} operation and the operation in \eqref{eq:plus} is also known as the \emph{variable} or \emph{plus} operation. 

By iterating this operation $n$ times, we map $n=2^m$ identical copies of the transmission channel $W$ into the synthetic channels $\{W_m^{(i)}\}_{i\in \{0, \ldots, n-1\}}$. More specifically, given $i\in \{0, \ldots, n-1\}$, let $(b_1, b_2, \ldots, b_m)$ be its binary expansion over $m$ bits, where $b_1$ is the most significant bit and $b_m$ is the least significant bit, i.e.,
\vspace{-1mm}
\begin{equation}\label{eq:defbinexp}
i = \sum_{k=1}^m b_k 2^{m-k}.
\vspace{-1mm}
\end{equation} 
Then, we define the synthetic channels $\{W_m^{(i)}\}_{i\in \{0, \ldots, n-1\}}$ as
\vspace{-1mm}
\begin{equation}\label{eq:syntchan}
W_m^{(i)} = (((W^{b_1})^{b_2})^{\cdots})^{b_m}.
\vspace{-1mm}
\end{equation}

\begin{example}[Synthetic Channel]
	Take $m=4$ and $i=10$. Then, the synthetic channel $W_{4}^{(10)} = (((W^{1})^{0})^{1})^{0}$ is obtained by applying first \eqref{eq:plus}, then \eqref{eq:minus}, then \eqref{eq:plus}, and finally \eqref{eq:minus}.
\end{example}

The polar construction is polarizing in the sense that the synthetic channels tend to become either completely noiseless or completely noisy. Thus, in the encoding procedure, the $k$ information bits are assigned to the positions (indices) corresponding to the best $k$ synthetic channels. Here, the quality of a channel is measured by some reliability metric such as the Bhattacharyya parameter of the channel. The remaining positions are ``frozen" to predefined values that are known at the decoder. As a result, the generator matrix of polar codes is based on choosing the $k$ rows of the matrix $G_n = [1\,\, 0; 1\,\, 1]^{\otimes m}$ (with ``;'' separating the rows)
which correspond to the best $k$ synthetic channels. 

\section{Fundamental Limits} \label{sec:limits}
\subsection{The Low-Capacity Regime}
\label{sec:lowcap}
We first provide an informal description of the low-capacity regime and then proceed with a more formal specification. The low-capacity regime consists of two main components:
(i) A channel $W$ with capacity $C$. We think of $C$ to be a very small number but \emph{fixed}; (ii) The blocklength $n$ which is defined as the number of times the channel $W$ is used for transmission. Here, $n$ should be thought of as a  finite value, i.e., non-asymptotic. 

We are interested in characterizing (optimal) ranges of $n$ for which reliable transmission of a certain number of information bits is possible.   Let $k$ denote the number of information bits to be sent.
To reliably communicate $k$ bits, we clearly must have $n \geq k/C$ and thus $n$ becomes fairly large when $C$ is small. We are interested in finite-length values of $n$ and their dependency on $C$ and $k$. For example, assume that we aim to send a constant number of information bits $k$ through the channel. We ask: What is the optimal (smallest) value of $n$ to send $k$ bits over the low-capacity channel with a given (fixed) error probability $p_e$? More precisely, we are searching over the set of all the possible values of $n$, such that $n \geq k/C$. In this search, we look for the smallest value of $n$ for which reliable transmission of $k$ bits with the desired error probability $p_e$ is achievable.

One may ask why this question is practically relevant and/or worth a deeper theoretical investigation. We argue from two perspectives: (i) \emph{Practical relevance:} There are many practical scenarios where the goal is to send a few bits over a low-capacity channel. For instance, in narrowband applications discussed in Section\,\ref{sec:intro}, the number of information bits $k$ is around a few tens, and the channel capacity $C$ is typically below $0.05$. This makes $n$ to vary in the range of a few thousand. For instance, if $k = 20$ and $C = 0.02$, then the blocklength $n$ is at least $1000$; (ii) \emph{Fundamental limits:} As we will see, the low-capacity regime allows for simple and precise trade-offs between the length $n$ and the number of information bits $k$, which depend on the channel and error probability. 
This stands in contrast to the case where $C$ is not very small. For example, when $C=\frac 12$, sending $20$ bits of information requires a fairly small blocklength $n$. In such a regime of $n$, it is intractable to provide precise closed-from estimates of the optimal blochlength and one typically resorts to (approximately) computing the well-known information-theoretic bounds, such as the random coding bound, among others. However, when $C$ is small, the blocklength $n$ becomes sufficiently large that allows us to provide simple, precise, and closed-form estimates of the optimal channel coding blocklength. Indeed, as we will see, such estimates require new techniques beyond the current methods used to analyze the finite-length limits of channel coding.  

The low-capacity regime is not necessarily limited to transmitting a constant number  of information bits, and one can consider other regimes of $k$. For example, assume that we would like to find the smallest value of $n$ such that we can send $k = \alpha \log n$ bits of information reliably using $n$ transmissions with an error probability not larger than $p_e$. How does the optimal (smallest) value of $n$ scale with $C$? In the same manner, we may consider $k = \alpha \sqrt{n}$ and ask for the smallest value of $n$ such that reliable transmission with $k$ bits is possible. In this case, we clearly have $n \geq \alpha \sqrt{n}/C $, or equivalently $n \geq \alpha^2/C^2$. As a result,  we are searching over the set of all the possible values of $n$, such that $n \geq \alpha^2/C^2$, and in this search, we look for the smallest value of $n$ for which reliable transmission of $\alpha \sqrt{n}$ bits with a desired error probability $p_e$ is achievable.


In order to proceed with a \emph{formal} and \emph{general} definition of the low-capacity regime, we need to provide a formal characterization of the term ``low'' in the finite-length regime where all the parameters such as $C$ and $n$ are assumed to be fixed and finite quantities. The low-capacity regime is formally defined using a channel $W$ with capacity $C$ and a function $f:\mathbb{R}_+\rightarrow \mathbb{R}_+$ such that $\displaystyle\lim_{n\to \infty} f(n)/n \to 0$. The main question that we ask is:\vspace{-1mm}

\begin{displayquote}
\textbf{ What is the smallest value of $n$ such that $f(n)$ bits can be transmitted reliably in $n$ transmissions over the channel $W$?}
\end{displayquote}
 \vspace{-1mm}
A reliable communication necessitates $f(n)/n\leq C$ and $f(n)=nC$ indicates the maximum rate that is hypothetically achievable under any coding scheme. Define $\kappa=nC$. For the sake of analysis, one can treat $n$ and $C$ as variables and can hypothetically take them to the limits.  As $C$ gets smaller and smaller, we require  $n$ to get larger and larger for a reliable communication. In the limit, $C\to 0$ leads to $n\to \infty$ but the behavior of $\kappa=nC$ is a different story. In terms of $n$, $\kappa$ can remain a constant or behave as a function, e.g., $\kappa=\log n$. This functionality is determined by the equality $f(n)=\kappa$. Hence, one can characterize the low capacity regime by determining the function $f$. As we will discuss later in this section, as $n$ and $C$ go to extremes, the expansion of $\log_2 M^*(n,p_e)$ in terms of $n$ (i.e., formula~\eqref{poly-formula}) becomes less accurate. To tackle this, our approach is obtaining the expansion of $\log_2 M^*(n,p_e)$ in terms of $\kappa$ rather than $n$. Suppose $\kappa=nC$ remains constant while $C\to 0$ and $n\to \infty$, then the expansion of $\log_2 M^*(n,p_e)$ in terms of $\kappa$ remains stable despite $n$ and $C$ going to extremes. Every such expansion for $\log_2 M^*(n,p_e)$ in terms of $\kappa$ remains valid when $\kappa$ or equivalently $f(n)$ lies in a specific range of functions. As it can be seen in the following sections, the constraint over $f$ only depends on its asymptotic behavior and thus, the valid range for $f(n)$ (or equivalently $\kappa$) can be represented in terms of $\mathcal{O}(\cdot)$ notation.

\noindent  {\bf Why the laws should be different in the low-capacity regime?} 
In this paper, we investigate code design over channels with a very low capacity.  Even though the formula \eqref{poly-formula} can still be used in the low-capacity regime, it provides a very loose approximation as  (i) the channel variations in the low-capacity regime are governed by different probabilistic laws than the ones used to derive \eqref{poly-formula}, and (ii) some of the terms hidden in $\mathcal{O}(\log n)$ will have significantly higher values and are comparable to the first and second terms. Similar arguments lead to the fact that the saddlepoint approximation \eqref{saddle} needs to be replaced with a more precise derivation in very low-capacity scenarios. Results provided in Section~\ref{sec:limits} will address these challenges.

Let us now explain why the current non-asymptotic laws of channel coding provided in \eqref{poly-formula} are not applicable in the low-capacity regime. 
Consider transmission over BEC($ \epsilon $) with blocklength $n$. When the erasure probability $ \epsilon $ is not too large (e.g., $ \epsilon = 0.5$), the number of channel non-erasures will be governed by the central limit theorem and behaves as $nC + \sqrt{n  \epsilon (1- \epsilon )}Z$, where $Z$ is the standard normal random variable.   
However, in the low-capacity regime, where the capacity $C=1- \epsilon $ is very small, the number of channel non-erasures will not be large, as the probability of non-erasure is very small. In other words, the expected number of non-erasures is $\kappa = n(1-\epsilon)$ which is much smaller than $n$. In this case, the number of  non-erasures is best approximated by the Poisson convergence theorem (i.e., the law of rare events) rather than the central limit theorem. Such behavioral differences in the channel variations will lead to totally different non-asymptotic laws, as we will see below. 

Another reason for \eqref{poly-formula} being loose is that some of the terms that are considered as $\mathcal{O}(1)$  become significant in the low-capacity regime. 
E.g., we have  $1/(\sqrt{n}C)\!=\!\sqrt{n}/(nC)\! =\!\!\sqrt{n}/\kappa$ which cannot be considered as $o(1)$ as  $\kappa$ is usually much smaller than $n$. As we will see, such terms can be captured by using sharper tail bounds.

{\bf Our approach.} Note that extremely tight converse and achievability bounds for BEC and BSC have existed prior to \cite{polyphd,polypaper} and stated as \cite[Corollary~42, Theorem~43]{polyphd} for BEC and \cite[Corollary~39, Theorem~40]{polyphd} for BSC. These bounds are in a raw implicit form. The novel contribution of \cite{polyphd,polypaper} is in using normal approximations and probability tail bounds to convert these implicit forms into explicit ones directly relating $\log_2M^*(n,p_e)$ to $n$, $p_e$. This procedure works well for moderate values of $C$ with respect to $n$ but fails to provide accurate estimates in the low-capacity setting considered in this paper. In order to provide an accurate estimate, we need novel probabilistic laws which are, in some cases such as the BEC, totally different than what has been used before. Our approach can be summarized as follows: our starting points are the same as \cite{polyphd,polypaper}, i.e., we start with \cite[Corollary~42, Theorem~43]{polyphd} for BEC and \cite[Corollary~39, Theorem~40]{polyphd} for BSC, but our analysis is based on Poisson approximations (for BEC) and much tighter probability tail bounds (for BSC) which are specifically suitable for the low-capacity regime but not necessary for moderate values of $C$. These novel approaches in our analysis lead to the low-capacity coding bounds for BEC and BSC stated in the following subsections. For the case of an AWGN channel, it turns out that the bounds for the low-capacity regime are just a limiting case of the state-of-the-art bounds in the moderate-capacity regime. All Proofs are provided in the Appendix which is presented in the supplementary material.

\subsection{The Binary Erasure Channel}
As discussed earlier, the behavior of channel variations for the BEC in the low-capacity regime can be best approximated through the Poisson convergence theorem for rare events. 
This will lead to different (i.e., more accurate) non-asymptotic laws. Theorem~\ref{Bounds for BEC} provides lower and upper bounds for the best achievable rate in terms of $n$, $p_e$, $\epsilon$, and $\kappa := n(1-\epsilon)$.  We use $\mathcal{P}_{\lambda}(x)$ to denote the Poisson cumulative distribution function, i.e., 
\begin{equation}\label{Poisson cdf}
\mathcal{P}_{\lambda}(x)=\text{Pr}\set{X<x}, \qquad \text{where   } X\sim \text{Poisson}(\lambda).
\end{equation}
\begin{theorem}[Finite-Length Coding Bounds for Low-Capacity BEC]\label{Bounds for BEC}

	Consider transmission over ${\rm{BEC}}( \epsilon )$ and let $\kappa = n(1- \epsilon)$. Then, $M_1 \leq M^*(n, p_e) \leq M_2$,
	where $M_1$ is any (or the largest) value that satisfies 
	\begin{equation}\label{M1}
	\mathfrak{P}_1(M_1)+2\alpha_2 \sqrt{ \mathfrak{P}_1(M_1)} +\alpha_1\sqrt{\mathcal{P}_{\kappa}(\log_2M_1)}-p_e\leq 0,
	\end{equation}
	and $M_2$ is any (or the smallest) value that satisfies 
	\begin{equation}\label{M2}
		 \mathfrak{P}_2(M_2)-\alpha_2\sqrt{\mathfrak{P}_2(M_2)}-\alpha_1\sqrt{ \mathcal{P}_{\kappa}(\log_2M_2)}-p_e\geq 0,
	\end{equation}
	and
	\begin{align}
	\mathfrak{P}_1(M_1)&=\mathcal{P}_{\kappa}( \log_{2}M_1)+M_1e^{-\kappa/2}\left(1-\mathcal{P}_{\kappa/2}(\log_2 M_1) \right),\\
	\mathfrak{P}_2(M_2)&=\mathcal{P}_{\kappa}( \log_{2}M_2)-\frac{e^{\kappa}}{M_2}\,\mathcal{P}_{2\kappa}\left(\log_2M_2\right),
	\end{align}
	\begin{equation}
	\alpha_1 = \frac{\sqrt{2}}{\epsilon^{3/2}}\left(\sqrt{e}-1\right)(1-\epsilon),\quad \alpha_2 = \frac{\sqrt{6}}{\epsilon^2\sqrt{\kappa}}\left(\sqrt{e}-1\right)(1-\epsilon).
	\end{equation}
\end{theorem}


\begin{proof}
	See Section~\ref{BEC proofs} in the Appendix.
\end{proof}

Following the discussion of Section~\ref{sec:lowcap}, to specify the domain that is considered as the low-capacity regime, one would provide the collection of functions $\kappa=f(n)$ for which Theorem~\ref{Bounds for BEC} holds. It is important to note that the analysis of Theorem~\ref{Bounds for BEC} does not depend on the values of $n$ and $C$ and thus the results hold for all choices of the function $f$, i.e., the results of Theorem~\ref{Bounds for BEC} mathematically hold for all values of $n$, $C$, and $p_e$, i.e., for a moderate-capacity regime as well. However, they provide a sharp estimate when $\kappa=f(n)=\mathcal{O}(1)$, e.g., when $\kappa=nC$ remains a moderate value despite a large $n$ and a small $C$.
 Moreover, note that the bounds in Theorem~\ref{Bounds for BEC} are expressed merely in terms of $\kappa := n(1-\epsilon)$ rather than $n$. This agrees with the intuition that the rate should depend on the amount of ``information'' passed through $n$ usages of the channel rather than the number of channel uses $n$. 
Typically, the value of $\kappa$ in low-capacity applications varies between a few tens to a few hundred. In such a range, no simple closed-form approximation of the Poisson distribution with mean $\kappa$ exists. As a result, the lower and upper bounds in Theorem~\ref{Bounds for BEC} cannot be simplified further.   Also, one can turn these bounds into bounds on the shortest (optimal) lengths $n^*$ needed for transmitting $k$ information bits with error probability $p_e$ over a low-capacity BEC. In Section\,\ref{sec:simulations}, we numerically evaluate the lower and upper bounds predicted by Theorem~\ref{Bounds for BEC} (see also Section~\ref{ramanujan}) and compare them with the prediction obtained from Formula~\eqref{poly-formula} \cite{polypaper}. It is observed that our predictions are significantly more precise compared to the prediction obtained from Formula~\eqref{poly-formula} and they become even more precise as the capacity approaches zero. 

\subsection{The Binary Symmetric Channel}
Unlike BEC, the non-asymptotic behavior of coding over BSC can be well approximated in the low-capacity regime by the central limit theorem (e.g., Berry-Essen theorem). To briefly explain the reason, consider transmission over BSC($\delta$) where the value of $\delta$ is close to $0.5$. The capacity of this channel is $1-h_2(\delta)$, where $h_2(x) := -x \log_2(x) - (1-x)\log_2(1-x)$ and we denote $\kappa = n(1-h_2(\delta))$. Note that when $\delta \to 0.5$ one can write $\delta \approx 0.5 - \sqrt{{\kappa}/{n}} $ by using the Taylor expansion of the function $h_2(x)$ around $x = 0.5$.  Transmission over BSC($\delta$) can be equivalently modeled as follows: (i) With probability $2\delta$, we let the output of the channel be chosen according to Bernoulli($0.5$), i.e., the output is completely random and independent of the input, and (ii) with probability $1-2\delta$, we let the output be exactly equal to the input.  In other words, the output is completely noisy with probability $2\delta$ (call it the noisy event) and completely noiseless with probability $1-2\delta$ (call it the noiseless event). As $\delta \to 0.5$, then the noiseless event is a \emph{rare event}.  Now assuming $n$ transmissions over the channel, the expected number of noiseless events is $n(1-2\delta) \approx \sqrt{n\kappa}$. Similar to BEC, the number of rare noiseless events follows a Poisson distribution with mean $n(1-2\delta)$ due to the Poisson convergence theorem. However, as the value of $n(1-2\delta) \approx \sqrt{n\kappa}$ is large, the resulting Poisson distribution can also be well approximated by the Gaussian distribution due to the central limit theorem (note that Poisson$(m)$ can be written as the sum of $m$ independent Poisson$(1)$ random variables).     

As mentioned earlier, central limit laws are the basis for deriving the laws of the form \eqref{poly-formula} which are applied to the settings where the capacity is not small. 
 However, for the low-capacity regime, considerable extra effort is required in terms of sharper arguments and tail bounds to work out the constants correctly.

\begin{theorem}[Finite-Length Coding Bounds for Low-Capacity BSC] \label{bounds for BSC} 

	Consider transmission over BSC($\delta$) in low-capacity regime in the sense that the function $f(n)$ mentioned in Section~\ref{sec:lowcap} belongs to the class $(o(n))^{2/3}$ and let $\kappa = n(1-h_2(\delta))$. Then,
	\vspace{-2mm}
	\begin{equation}\label{BSC b}
\log_2M^*(n,p_e)=\kappa-2\sqrt{\frac{2\kappa\delta(1-\delta)}{\ln 2}}\,Q^{-1}\left(p_e\right)+\frac 12\log_2 \kappa-\log_2p_e+\mathcal{O}\left(\log\log\kappa\right).
	\end{equation}
\end{theorem}
\begin{proof}
	See Section~\ref{BSC proofs1} in the Appendix.
\end{proof}
Following the discussion of Section~\ref{sec:lowcap}, note that the low-capacity regime considered in Theorem~\ref{bounds for BSC} is specified by the function $f$ belonging to the class $(o(n))^{2/3}$. This means for the estimate \eqref{BSC b} to hold, we must have $\kappa =f(n)=(o(n))^{2/3}$. This constraint comes from using the condition $\kappa\sqrt{\kappa}=o(n)$ in the proof of  Theorem~\ref{bounds for BSC}. Moreover, we remark that the $\mathcal{O}(\log \log \kappa)$ term contains some other terms such as $\mathcal{O}(\sqrt{-\log p_e}/\log \kappa)$. For practical scenarios, the term $\mathcal{O}(\log \log \kappa)$ will be dominant.\footnote{We include only the dominant term inside $\mathcal{O}(\cdot)$.} We also note that similar to the BEC case, all terms in (\ref{BSC b}) are expressed in terms of $\kappa$ rather than $n$. This agrees with the intuition that the rate should depend on the amount of ``information" passed through $n$ usages of the channel rather than the number of channel uses $n$.  
\begin{corollary}\label{optimal n for BSC} 
	Consider transmission of $k$ information bits over a low-capacity BSC($\delta$) that is specified in Theorem~\ref{bounds for BSC}. Then, the optimal blocklength $n^*$ for such transmission is
	\vspace{-2mm}
	\begin{equation}
	    n^* =\frac 1{1-h_2(\delta)} \left(k+2\sqrt{\frac{2\kappa\delta(1-\delta)}{\ln 2}}Q^{-1}( p_e )+\frac{4\delta(1-\delta)}{\ln 2}Q^{-1}(p_e)^2+\log_2 p_e+\mathcal{O}(\log k)\right).
	\end{equation}
\end{corollary}
\begin{proof}
\vspace{-2mm}
	See Section~\ref{BSC proofs2} in the Appendix.
\end{proof}

\subsection{The Additive White Gaussian Noise Channel}
\label{sec:awgn}

First, let us further clarify our description of coding over the AWGN channel.  We consider $n$ uses of a real AWGN channel in which the input $X_i$ and the output $Y_i$ at each $i=1,\dots,n$ are related as $Y_i=X_i+Z_i$. Here, the noise term $\set{Z_i}_{i=1}^n$ is a memoryless, stationary Gaussian process with zero mean and unit variance. 
Given an $(M,p_e)$-code for  $W^n$, where $W$ is the AWGN channel, a cost constraint on the codewords must be applied. The most commonly used cost is 
\vspace{-2mm}
\begin{equation}
    \forall m\in \mathcal{M}:\quad \norm{f_{enc}(m)}_2^2=\sum_{i=1}^n\left(f_{enc}(m)\right)_i^2\leq \eta\, n,
    \vspace{-2mm}
\end{equation}
where $\eta$, with a slight abuse of notation, refers to SNR. Since characterization of the code depends on the SNR $\eta$, we denote an $(M,p_e)$-code and $M^*(n,p_e)$ by $(M,p_e,\eta)$-code and $M^*(n,p_e,\eta)$, respectively.

Similar to BSC, the channel variations in low-capacity AWGN channels are best approximated by the central limit theorem. The following theorem is obtained by using the ideas in \cite[Theorem~73]{polyphd} with slight modifications. It turns out that coding bounds for AWGN in the low-capacity regime can be obtained as a limiting case of the state-of-the-art bounds in the moderate-capacity regime. The following theorem and corollary are resulted simply by repeating the same argument in a manner that remains valid for these limiting cases. This needs a tiny refinement of the analysis which is done in the Appendix.
\begin{theorem}[Finite-Length Coding Bounds for Low-Capacity AWGN] \label{bounds for AWGN} 
Consider transmission over AWGN($\eta$) in low-capacity regime and let $\kappa = \frac n2 \log_2(1+\eta)$. Then,
\vspace{-2mm}
	\begin{equation}\label{AWGN b}
		\log_2 M^*(n, p_e,\eta ) = \kappa -\frac{\sqrt{\eta+2}}{(\eta+1)\sqrt{\ln2}}\cdot \sqrt{\kappa}\,Q^{-1}( p_e )+\mathcal{E},
		\vspace{-2mm}
	\end{equation}
	\vspace{-2mm}
	where
	\vspace{-2mm}
	\begin{equation}
	\mathcal{O}(1)\leq \mathcal{E}\leq \frac 12\log_2 \kappa-\log_2p_e+\mathcal{O}\left(\frac{1}{\sqrt{-\log p_e}}\right).
	\vspace{-2mm}
	\end{equation}
\end{theorem}
\begin{proof}
	See Section~\ref{AWGN proofs1} in the Appendix.
\end{proof}
The same considerations about $\mathcal{O}(\cdot)$ notation, as discussed earlier, should be taken into account here. Also note that as for BEC and BSC, the optimal blocklength for the AWGN channel can be expressed in terms of other parameters in the low-capacity regime which is stated in the following corollary.
\begin{corollary}\label{optimal n for AWGN} 
	Consider transmission of $k$ information bits over a low-capacity AWGN($\eta$). Then, the optimal blocklength $n^*$ for such transmission is
	\begin{align}
		n^* =\frac 2{\log_2(1+\eta)} \left(k+\frac{\sqrt{\eta+2}}{(\eta+1)\sqrt{\ln2}}Q^{-1}( p_e )\cdot \sqrt{k}+\mathcal{O}\left(\log \frac{1}{p_e}\right)\right).
	\end{align}
\end{corollary}
\begin{proof}
	See Section~\ref{AWGN proofs2} in the Appendix. 
\end{proof}
As we mentioned earlier, Theorem~\ref{bounds for AWGN} is obtained directly by taking the limit of the AWGN results in \cite{polypaper}. The Corollary~\ref{optimal n for AWGN} is consequently a limiting case of the moderate-regime analysis in \cite{polypaper}. Therefore, the contribution of these results is merely the emphasis on the fact that when the underlying channel is AWGN, the moderate-capacity regime analysis in \cite{polypaper} may extend well to the low-capacity regime only by taking the limit. Having this, we do not numerically evaluate Corollary~\ref{optimal n for AWGN} on the AWGN channel.

\vspace*{-2ex}
\section{The Analysis of Practical Code Designs} \label{sec:design}

\vspace*{-2ex}
As we need to design codes with extremely low rates, some of the state-of-the-art codes may not be directly applicable. A notable instance is the class of iterative codes, e.g., turbo or low-density parity-check (LDPC) codes. It is well known that decreasing the design rate of iterative codes results in denser decoding graphs which further leads to highly complex iterative decoders with poor performance. E.g., an $(l,r)$-regular LDPC code with design rate $R = 0.01$ requires  $r,l \geq 99$. Hence, the Tanner graph will have a minimum degree of at least $99$ and even for code block lengths on the order of  tens of thousands, the Tanner graph will have many short cycles. To circumvent this issue, the current practical designs, e.g., the NB-IoT standard, use repetition coding.  More specifically, a low-rate repetition code is concatenated with a powerful moderate-rate code. For example, an iterative code of rate $R$ and length $n/r$ can be repeated $r$ times to construct a code of length $n$ with rate $R/r$. In Section~\ref{sec:repeat}, we will discuss the advantages and drawbacks of using repetition schemes along with trade-offs between the number of repetitions and the performance of the code.

Unlike iterative codes, polar codes and most algebraic codes (e.g., BCH or Reed-Muller codes) can be used without any modification for low-rate applications.  In Section~\ref{sec:polar}, we look into polar coding for low-capacity channels. In particular, we show that polar coding is advantageous in terms of inherently adopting an optimal number of repetitions. Theorems~\ref{r} and \ref{r_bms} provide tight bounds on the optimal number of repetitions in terms of the capacity, that are not specific to polar codes, and Theorem~\ref{polar_repetition} shows that the construction of polar codes naturally adopts a certain number of repetition blocks in the low-capacity regime that match the optimal number of repetitions (up to a constant multiplicative factor). 

Throughout this section, we will consider code design for the class of binary memoryless symmetric (BMS) channels. A BMS channel $W$ has binary input and, letting $W(y\mid x)$ denotes the transition matrix, there exists a permutation $\pi$ on the output alphabet such that  $W(y \mid 0) = W(\pi(y) \mid 1)$. Notable exemplars of this class are BEC, BSC, and BAWGN channels.\footnote{For a relevant study of code design at a low-SNR scenario, see \cite{Posner67,ChaoLowSNR}. }

\subsection{How Much Repetition is Needed?} \label{sec:repeat}
As mentioned earlier, repetition is a straightforward way to design practical low-rate codes while utilizing the power of state-of-the-art code designs. Let $r$ be a divisor of $n$, where $n$ denotes the length of the code. 
Repetition coding consists of designing first a smaller outer code of length $n/r$ and then repeating each of the coded bits $r$ times (i.e., the inner code is a repetition code of rate $1/r$). The length of the overall code is then $n/r\cdot r = n$.  This is equivalent to transmitting the outer code over the \emph{$r$-repetition channel}, $W^r$, which takes a bit as the input and outputs an $r$-tuple which is the result of passing $r$ copies of the input bit independently through the original channel $W$.   E.g., if $W$ is BEC($\epsilon$) then its corresponding $r$-repetition channel is  $W^r = \text{BEC}(\epsilon^r)$. 

The main advantage of repetition coding is the reduction in computational complexity, especially when $r$ is large. This is because the encoding/decoding complexity is effectively reduced to that of  the outer code, i.e., once the outer code is constructed, at the encoding side, we just need to repeat each of its coded bits $r$ times, and at the decoding side the log-likelihood of an $r$-tuple consisting of $r$ independent transmissions of a bit is equal to the sum of the log-likelihoods of the individual channel outcomes. The computational latency of the encoding and decoding algorithms is reduced to that of the outer code in a similar fashion.    

The outer code has to be designed for reliable communication over the channel $W^r$. If $r$ is sufficiently large, then the capacity of $W^r$ will not be low anymore. In this case, the outer code can be picked from off-the-shelf practical codes designed for channels with moderate capacity values (e.g., iterative or polar codes). While this looks promising, one should note that the main drawback of repetition coding is the loss in capacity.  In general, we have $C(W^r) \leq rC(W)$ and the ratio vanishes by growing $r$.  As a result, if $r$ is very large then repetition coding might suffer from an unacceptable rate loss.  Thus, the main question that we need to answer is: how large $r$ can be made such that the rate loss is still \textit{negligible}?   

We note that the overall capacity corresponding to $n$ channel transmissions is $nC(W)$. With repetition coding, the capacity will be reduced to $n/r \cdot C(W^r)$ since we transmit $n/r$ times over the channel $W^r$.  For any $\beta \in [0,1]$, we ask what is the largest repetition size $r_\beta$ such that 
\vspace{-2mm}
\begin{equation} \label{r_constraint}
\frac{n}{r_\beta} C(W^{r_\beta}) \geq \beta n C(W).  
\vspace{-2mm}
\end{equation}  
Let us first assume that transmission takes place over BEC($\epsilon$). We thus have $W^r = \text{BEC}(\epsilon^r)$. If $\epsilon$ is not close to $1$, then even $r=2$ would result in a considerable rate loss, e.g., if $\epsilon = 0.5$, then $ C(W^2) = 0.75$ whereas $2C(W) = 1$.  However, when $\epsilon$ is close to $1$, then at least for small values of $r$ the rate loss can be negligible, e.g., for $r = 2$, we have $C(W^2) =    1-\epsilon^2 \approx 2(1-\epsilon) = 2C(W)$. The following theorem provides lower and upper bounds for the largest repetition size $r_\beta$ that satisfies \eqref{r_constraint}. 

\begin{theorem}[Maximum Repetition Length for BEC]
	\label{r}
	If $W = \text{BEC}(\epsilon)$, then for the largest repetition size $r_\beta$ that satisfies \eqref{r_constraint}, we have
	\begin{equation}
	\label{rbound}
	\frac{n(1-\epsilon)\ell}{2\left(1-\frac{\beta}{\ell}\right)}\cdot\left(\frac{\beta}{\ell}\right)^2 \le \frac{n}{r_{\beta}} \le \frac{n(1-\epsilon)\ell}{2\left(1-\frac{\beta}{\ell}\right)},
	\end{equation}
	where $\ell = -{(\ln \epsilon)}/{(1-\epsilon)} $. Equivalently, assuming $\kappa = n(1-\epsilon)$, \eqref{rbound} becomes
	\begin{equation}
	\frac{\kappa}{2\left(1-\beta\right)}\cdot\beta^2 (1 + \mathcal{O}(1-\epsilon)) \le \frac{n}{r_{\beta}} \le     \frac{\kappa}{2\left(1-\beta\right)} (1 + \mathcal{O}(1-\epsilon)).
	\end{equation}
\end{theorem}
\begin{proof}
	See Section~\ref{rpp1} in the Appendix. 
\end{proof}

\begin{remark} \label{r_remark}
	Going back to the results of Theorem~\ref{Bounds for BEC}, in order to obtain similar finite-length guarantees with repetition-coding, a necessary condition is that the total rate loss due to repetition is $\mathcal{O}(1/n)$, i.e.,
	\vspace{-2mm}
	\begin{equation}
	    \frac{n}{r_\beta} C(W^{r_\beta}) =  n C(W) + \mathcal{O}(1).
	    \vspace{-2mm}
	\end{equation}
	If $W = \text{BEC}(\epsilon)$ and $\kappa = n(1 -\epsilon)$,  then the necessary condition implies plugging $\beta = 1 - \mathcal{O}(1/\kappa)$ into \eqref{r_constraint}. Moreover, from Theorem~\ref{r} we can conclude that, when $\epsilon$ is close to $1$, the maximum allowed repetition size is $\mathcal{O}\left(n / \kappa^2\right)$. Equivalently, the size of the outer code can be chosen as $\mathcal{O}(\kappa^2)$.
\end{remark}

A noteworthy conclusion from the above remark is that as having negligible rate loss implies the repetition size to be at most $\mathcal{O}(n/\kappa^2)$, then the outer code has to be designed for a BEC with erasure probability at least $\epsilon^{\mathcal{O}(n/\kappa^2)} = 1 - \mathcal{O}(1/\kappa)$.  This means that the outer code should still have a low rate even if $\kappa$ is as small as a few tens. Thus, the idea of using codes such as iterative codes as the outer code and repetition codes as the inner code will lead to an efficient low-rate design only if we are willing to tolerate non-negligible rate loss. 
In contrast, the polar coding construction has implicitly a repetition block of optimal size $\mathcal{O}(n/\kappa^2)$ as we will see in the next section.   

In Appendix~\ref{rpp2}, we show that the binary erasure channel has the smallest rate loss due to repetition among all the BMS channels. This property has been used in the following theorem to provide an upper bound on $r_\beta$ for any BMS channel. 

\begin{theorem}[Upper Bound on Repetition Length for any BMS] \label{r_bms}
	Among all BMS channels with the same capacity, BEC has the largest repetition length $r_\beta$ that satisfies \eqref{r_constraint}.   Hence, for any BMS channel with capacity $C$ and $\kappa = nC$, we have 
	\begin{equation}
	    \frac{n}{r_\beta} \geq  \frac{\kappa}{2(1-\beta)} \beta^2 (1 + \mathcal{O}(1 - C)).
	\end{equation}
\end{theorem}

\begin{proof}
	See Section~\ref{rpp2} in the Appendix. 
\end{proof}
\begin{remark} \label{r_remark_bms}
	Similar to Remark~\ref{r_remark}, we can conclude that for any BMS channel with low capacity, in order to have the total rate loss of order $\mathcal{O}(1)$, the repetition size should be at most $\mathcal{O}(n/\kappa^2)$.
\end{remark} 

\subsection{Polar Coding and Repetition at Low Capacity} \label{sec:polar}

We have shown in Section~\ref{sec:repeat} that the maximum allowed repetition size to have negligible capacity loss is $\mathcal{O}(n/\kappa^2)$. We show in this section that at low-capacity regime, the polar construction is enforced to have $\mathcal{O}(n/\kappa^2)$ repetitions. In other words, the resulting polar code is equivalent to a smaller polar code of size $\mathcal{O}(\kappa^2)$ followed by repetitions. Consequently, the encoder and decoder of the polar code could be implemented with much lower complexity taking into account the naturally adopted repetitions. That is, the encoding complexity can be reduced to $n + \mathcal{O}( \kappa^2 \log \kappa)$ and the decoding complexity using the successive cancellation list (SCL) decoder with list size $L$, proposed by Tal and Vardy \cite{tal2015list}, is reduced to $n + \mathcal{O}(L \kappa^2 \log \kappa)$. Recall that the original implementation of polar codes requires $\mathcal{O}(n \log n)$ encoding complexity and $\mathcal{O}(L n \log n)$ decoding complexity. Moreover, as the operations involving repeated blocks can all be done in parallel, the computational \emph{latency} of the encoding and decoding operations can be reduced to  $\mathcal{O}( \kappa^2 \log \kappa)$ and $\mathcal{O}( L\kappa^2 \log \kappa)$, respectively. To further reduce the complexity, the simplified SC decoder \cite{alamdar2011simplified} or relaxed polar codes \cite{el2017relaxed} can be invoked. Such complexity reductions are important for making polar codes a suitable candidate in practice. In a related work, designing low-rate codes for BSCs by concatenating high rate polar codes together with repetitions is considered \cite{dumer2017polar}. Furthermore, following the publication of the initial version of this paper, several works have looked into improving low-rate polar codes by either tweaking the repetition or concatenating them with other codes \cite{abbasi2022polar,dumer2021codes, abbasi2022hybrid}.

\begin{theorem} \label{polar_repetition}
	Consider using a polar code of length $n = 2^m$ for transmission over a BMS channel $W $. Let $m_0 =  \log_2 (4 \kappa^2) $ where $\kappa  = n C(W) $. 
	Then any synthetic channel $W_{n}^{(i)}$ whose Bhattacharyya value is less than $\frac 12$ has at least $m-m_0$ plus operations in the beginning.  As a result, the polar code constructed for $W$  is equivalent to the concatenation of a polar code of length (at most)
	$2^{m_0}$ followed by $2^{m - m_0}$ repetitions. 
\end{theorem}

\begin{proof}
	See Section~\ref{proofpolar} in the Appendix.
\end{proof}
\begin{remark}
	Note that from Theorem~\ref{polar_repetition}, polar codes automatically perform repetition coding with $\mathcal{O}(n/\kappa^2)$ repetitions, where $\kappa = nC$. This matches the necessary (optimal) number of repetitions given in Remark~\ref{r_remark}~and~\ref{r_remark_bms}.
\end{remark}

\vspace*{-3ex}
\section{Numerical Analysis of Fundamental Limits} \label{sec:simulations}

In this section, we numerically evaluate our channel coding bounds from Section~\ref{sec:limits}. We report the numerical results on the BEC and the BSC cases. As mentioned earlier, we do not numerically implement our results on the AWGN channel  (Theorem~\ref{bounds for AWGN} and Corollary~\ref{optimal n for AWGN}) since their contribution is merely showing that the limit of the existing moderate-capacity regime analysis will lead to a low-capacity regime bound.

For the BEC, we have compared in Figure~\ref{fig:BEC}, the lower and upper bounds obtained from Theorem~\ref{Bounds for BEC} with the predictions of Formula~\eqref{poly-formula}. We have also plotted the performance of polar codes. The setting considered in Figure~\ref{fig:BEC} is as follows: We intend to send $k=40$ information bits over the BEC$(\epsilon$). The desired error probability is $p_e = 0.01$.   For erasure values between $0.96$ and $1$, Figure~\ref{fig:BEC} plots bounds on the smallest (optimal) blocklength $n$ needed for this scenario as well as the smallest length required by polar codes.  Note that in order to compute a lower bound on the shortest length from Theorem~\ref{Bounds for BEC}, we should fix $M^*(n,p_e)$ to $k=40$ and search for the smallest $n$ that satisfies equation~\eqref{M2} with $\kappa = n(1-\epsilon)$ and $p_e = 0.01$.  
\begin{figure}[h!]
	\centering 
	\begin{subfigure}[b]{0.49\textwidth}
		\centering
%
\begin{NoHyper}
\begin{tikzpicture}

\definecolor{mycolor1}{rgb}{0.15,0.1,0.1}
\definecolor{mycolor2}{rgb}{0,0.447,0.741}
\definecolor{mycolor3}{rgb}{0.85,0.325,0.098}

\pgfplotsset{every axis grid/.style={style=solid}}

\begin{axis}[
scale only axis,
every outer x axis line/.append style={mycolor1},
every x tick label/.append style={font=\color{mycolor1}},
every outer y axis line/.append style={mycolor1},
every y tick label/.append style={font=\color{mycolor1}},
width=2.5in,
height=1.8in,
xlabel={Erasure probability $\epsilon$},
ylabel={$n$},
xmin=0.965, xmax=0.995,
ymin=1000, ymax=9000,
xmajorgrids,
xminorgrids,
ymajorgrids,
yminorgrids,
legend entries={\scriptsize{upper bound from Theorem~\ref{Bounds for BEC}}, \scriptsize{lower bound from Theorem~\ref{Bounds for BEC}}, \scriptsize{upper bound from Theorem~\ref{BEC achiv}},\scriptsize{lower bound from Theorem~\ref{BEC conv}}, \scriptsize{prediction from \eqref{poly-formula} \cite[Theorem~44]{polyphd}}, \scriptsize{polar-CRC (L=16)}},
legend style={nodes={scale=0.85, transform shape}},
legend style={nodes={right}},
legend cell align=right,
legend pos=north west,
xtick = {0.965,0.97,0.975,0.98,0.985,0.99,0.995},
xticklabels={.965,.97,.975,.98,.985,.99,.995}
]

\addplot [
color=blue,
solid,
line width=2.0pt,
mark=triangle,
mark options={solid}
]
coordinates{
 (0.991,6619)
 (0.99,5962)
 (0.988,4976)
 (0.986,4259)
 (0.984,3730)
 (0.982,3333)
 (0.98,3006)
 (0.978,2740)
 (0.976,2514)
 (0.974,2320)
 (0.972,2160)
 (0.97,2021)
 (0.968,1895)
 (0.966,1795)

};

\addplot [
color=blue,
solid,
line width=2.0pt,
mark=square,
mark options={solid}
]
coordinates{
 (0.991,6172)
 (0.99,5542)
 (0.988,4598)
 (0.986,3915)
 (0.984,3420)
 (0.982,3035)
 (0.98,2722)
 (0.978,2465)
 (0.976,2251)
 (0.974,2070)
 (0.972,1916)
 (0.97,1783)
 (0.968,1664)
 (0.966,1562)

};
\addplot [
color=green,
solid,
mark=star,
line width=2.0pt,
]
coordinates{
	(0.966,1728)
	(0.968,1836)
	(0.97,1959)
	(0.972,2100)
	(0.974,2262)
	(0.976,2451)
	(0.978,2675)
	(0.98,2943)
	(0.982,3271)
	(0.984,3681)
	(0.986,4208)
	(0.988,4911)
	(0.99,5894)
	(0.991,6550)

};

\addplot [
color=green,
solid,
line width=2.0pt,
]
coordinates{
 (0.966,1618)
 (0.968,1720)
 (0.97,1835)
 (0.972,1967)
 (0.974,2119)
 (0.976,2296)
 (0.978,2505)
 (0.98,2757)
 (0.982,3064)
 (0.984,3448)
 (0.986,3942)
 (0.988,4600)
 (0.99,5521)
 (0.991,6136)
};

\addplot [
color=purple,
dash pattern=on 5pt off 5pt on 5pt,
line width=2.0pt
]
coordinates{
 (0.991,4526)
 (0.99,4081)
 (0.988,3409)
 (0.986,2927)
 (0.984,2563)
 (0.982,2280)
 (0.98,2057)
 (0.978,1870)
 (0.976,1718)
 (0.974,1590)
 (0.972,1474)
 (0.97,1379)
 (0.968,1295)
 (0.966,1222)

};

\addplot [
color=mycolor3,
mark size=5.0pt,
only marks,
mark=asterisk,
mark options={solid,draw=red}
]
coordinates{
 (0.991,8192)
 (0.982,4096)
 (0.966,2048)

};

\end{axis}

\end{tikzpicture}
\end{NoHyper}
		\label{BEC_k40}
	\end{subfigure}
	\begin{subfigure}[b]{0.49\textwidth}\centering  
%

\begin{tikzpicture}

\definecolor{mycolor1}{rgb}{0.15,0.15,0.15}
\definecolor{mycolor2}{rgb}{0,0.447,0.741}
\definecolor{mycolor3}{rgb}{0.85,0.325,0.098}

\pgfplotsset{every axis grid/.style={style=solid}}

\begin{axis}[
scale only axis,
every outer x axis line/.append style={mycolor1},
every x tick label/.append style={font=\color{mycolor1}},
every outer y axis line/.append style={mycolor1},
every y tick label/.append style={font=\color{mycolor1}},
width=2.5in,
height=1.8in,
xlabel={Erasure probability $\epsilon$},
ylabel={Normalized $n$},
xmin=0.965, xmax=0.995,
ymin=0.7, ymax=1.4,
xmajorgrids,
xminorgrids,
ymajorgrids,
yminorgrids,
legend style={nodes=right},
legend cell align=right,
xtick = {0.965,0.97,0.975,0.98,0.985,0.99,0.995},
xticklabels={.965,.97,.975,.98,.985,.99,.995}
]

\addplot [
color=blue,
solid,
line width=2.0pt,
mark=triangle,
mark options={solid}
]
coordinates{
 (0.991,1.07242)
 (0.99,1.07578)
 (0.988,1.08221)
 (0.986,1.08787)
 (0.984,1.09064)
 (0.982,1.09819)
 (0.98,1.10434)
 (0.978,1.11156)
 (0.976,1.11684)
 (0.974,1.12077)
 (0.972,1.12735)
 (0.97,1.13348)
 (0.968,1.13882)
 (0.966,1.14917)

};

\addplot [
color=blue,
solid,
line width=2.0pt,
mark=square,
mark options={solid}
]
coordinates{
 (0.991,1)
 (0.99,1)
 (0.988,1)
 (0.986,1)
 (0.984,1)
 (0.982,1)
 (0.98,1)
 (0.978,1)
 (0.976,1)
 (0.974,1)
 (0.972,1)
 (0.97,1)
 (0.968,1)
 (0.966,1)

};

\addplot [
color=green,
solid,
line width=2.0pt
]
coordinates{
(0.966,1.03585)
(0.968,1.03365)
(0.97,1.02916)
(0.972,1.02662)
(0.974,1.02367)
(0.976,1.01999)
(0.978,1.01623)
(0.98,1.01286)
(0.982,1.00956)
(0.984,1.00819)
(0.986,1.0069)
(0.988,1.00043)
(0.99,0.996211)
(0.991,0.994167)

};

\addplot [
color=green,
solid,
mark = star,
line width=2.0pt
]
coordinates{
(0.966,1.10627)
(0.968,1.10337)
(0.97,1.09871)
(0.972,1.09603)
(0.974,1.09275)
(0.976,1.08885)
(0.978,1.08519)
(0.98,1.08119)
(0.982,1.07776)
(0.984,1.07632)
(0.986,1.07484)
(0.988,1.06807)
(0.99,1.06351)
(0.991,1.06124)
	
};

\addplot [
color=purple,
dash pattern=on 5pt off 5pt on 5pt,
line width=2.0pt
]
coordinates{
 (0.991,0.733312)
 (0.99,0.736377)
 (0.988,0.741409)
 (0.986,0.747637)
 (0.984,0.749415)
 (0.982,0.751236)
 (0.98,0.755694)
 (0.978,0.758621)
 (0.976,0.763216)
 (0.974,0.768116)
 (0.972,0.769311)
 (0.97,0.773416)
 (0.968,0.778245)
 (0.966,0.78233)

};

\addplot [
color=mycolor3,
mark size=5.0pt,
only marks,
mark=asterisk,
mark options={solid,draw=red}
]
coordinates{
 (0.991,1.33944)
 (0.982,1.3545)
 (0.966,1.29456)

};

\end{axis}

\end{tikzpicture}
		\label{BEC_k40_normalized}
	\end{subfigure}
	\caption{Comparison for low-capacity BECs. The number of information bits is $k=40$ and the target error probability is $p_e = 0.01$. For the right plot, with the same legend entries as the left plot, all the blocklengths $n$ in the left plot are normalized by the value of the lower bound, obtained from Theorem~\ref{Bounds for BEC}. \label{fig:BEC}}
\end{figure}

Note that Theorem~\ref{BEC achiv} \cite[Corollary~42]{polyphd} and Theorem~\ref{BEC conv} \cite[Theoreem~43]{polyphd} (presented in the Appendix) are the raw upper and lower bounds for the optimal blocklength in BEC that both the classical estimation \eqref{poly-formula} (whose precise version for BEC can be found in \cite[Theorem~44]{polyphd}) and Theorem~\ref{Bounds for BEC} are estimating. As expected, the prediction obtained from \cite[Theorem~44]{polyphd} is not precise in the low-capacity regime and it becomes worse as the capacity approaches zero. On the other hand, the estimated bounds from Theorem~\ref{Bounds for BEC} converge to the original raw bounds of Theorem~\ref{BEC achiv} and Theorem~\ref{BEC conv} as the capacity approaches zero. Also, the performance of the polar code is shown in Figure~\ref{fig:BEC}. The polar code is concatenated with cyclic redundancy check (CRC) code of length $6$ and is decoded with the list-SC algorithm \cite{tal2015list} with list size $L=16$. 

Figure~\ref{fig:BSC} considers the scenario of sending $k=40$ bits of information over a low-capacity BSC with target error probability $p_e = 0.01$. 
We have compared in Figure~\ref{fig:BSC}, the predictions from Theorem~\ref{bounds for BSC} and Formula~\eqref{poly-formula} (we used a precise version of Formula~\eqref{poly-formula} for BSC given in \cite[Theorem~41]{polyphd}) together with the raw upper and lower bounds  from Theorem~\ref{BSC achiv} \cite[Corollary~39]{polyphd} and Theorem~\ref{BSC conv} \cite[Theorem~40]{polyphd} that are presented in the Appendix. Note that both Theorem~\ref{bounds for BSC} and \cite[Theorem~41]{polyphd} provide single predictions to estimate the true value of the optimal blocklength that lies between the aforementioned raw bounds.
Therefore, it is not abnormal if neither of the predictions from \cite[Theorem~41]{polyphd} or Theorem~\ref{bounds for BSC} lie between the raw bounds. In this way, Figure~\ref{fig:BSC} shows that, as we expected, the prediction from \cite[Theorem~41]{polyphd} is quite imprecise in the low-capacity regime, particularly in comparison to the prediction from Theorem~\ref{bounds for BSC} which is exact up to $\mathcal{O}(\log \log \kappa)$ terms. The performance of polar codes is also plotted in Figure~\ref{fig:BSC}. An interesting problem is to analyze the finite-length scaling of polar codes in the low-capacity regime \cite{hassani2014finite,mondelli2016unified,goldin2014improved,guruswami2015polar,mahdavifar2020polar}. 

\begin{figure}[h!]
	\centering 
	\begin{subfigure}[b]{0.49\textwidth}
		\centering
%
\begin{NoHyper}
\begin{tikzpicture}

\definecolor{mycolor1}{rgb}{0.15,0.15,0.15}
\definecolor{mycolor2}{rgb}{0,0.447,0.741}
\definecolor{mycolor3}{rgb}{0.85,0.325,0.098}

\pgfplotsset{every axis grid/.style={style=solid}}

\begin{axis}[
scale only axis,
every outer x axis line/.append style={mycolor1},
every x tick label/.append style={font=\color{mycolor1}},
every outer y axis line/.append style={mycolor1},
every y tick label/.append style={font=\color{mycolor1}},
width=2.5in,
height=1.8in,
xlabel={Crossover probability $\delta$},
ylabel={$n$},
xmin=0.34, xmax=0.46,
ymin=0, ymax=9000,
xmajorgrids,
xminorgrids,
ymajorgrids,
yminorgrids,
legend entries={\scriptsize{upper bound from Theorem~\ref{BSC achiv}}, \scriptsize{lower bound from Theorem~\ref{BSC conv}}, \scriptsize{prediction from Theorem~\ref{bounds for BSC}}, \scriptsize{prediction from \cite[Theorem~41]{polyphd}}, \scriptsize{polar-CRC (L=16)}},
legend style={nodes={scale=0.85, transform shape}},
legend style={nodes={right}},
legend cell align=right,
legend pos=north west,
xtick = {0.36,0.38,0.4,0.42,0.44,0.46},
xticklabels={.36,.38,.4,.42,.44,.46}
]

\addplot [
color=green,
solid,
line width=2.0pt,
mark=star,
mark options={solid}
]
coordinates{
(  0.34 ,  868  )
(  0.345 ,  927  )
(  0.35 ,  991  )
(  0.355 ,  1063  )
(  0.36 ,  1141  )
(  0.365 ,  1229  )
(  0.37 ,  1327  )
(  0.375 ,  1437  )
(  0.38 ,  1561  )
(  0.385 ,  1702  )
(  0.39 ,  1862  )
(  0.395 ,  2046  )
(  0.4 ,  2258  )
(  0.405 ,  2504  )
(  0.41 ,  2792  )
(  0.415 ,  3132  )
(  0.42 ,  3539  )
(  0.425 ,  4029  )
(  0.43 ,  4628  )
(  0.435 ,  5370  )
(  0.44 ,  6305  )
(  0.445 ,  7507  )

};

\addplot [
color=green,
solid,
line width=2.0pt,
]
coordinates{
(  0.34 ,  846  )
(  0.345 ,  903  )
(  0.35 ,  965  )
(  0.355 ,  1035  )
(  0.36 ,  1111  )
(  0.365 ,  1196  )
(  0.37 ,  1291  )
(  0.375 ,  1398  )
(  0.38 ,  1518  )
(  0.385 ,  1655  )
(  0.39 ,  1810  )
(  0.395 ,  1987  )
(  0.4 ,  2193  )
(  0.405 ,  2431  )
(  0.41 ,  2710  )
(  0.415 ,  3039  )
(  0.42 ,  3433  )
(  0.425 ,  3907  )
(  0.43 ,  4488  )
(  0.435 ,  5207  )
(  0.44 ,  6112  )
(  0.445 ,  7274  )
	
};

\addplot [
color=blue,
solid,
line width=2.0pt,
mark = square,]
coordinates{
         (0.34 ,      792.05)
        (0.345   ,    847.56)
         (0.35    ,   906.69)
        (0.355   ,    969.92)
         (0.36    ,   1047.9)
        (0.365   ,    1131.3)
        ( 0.37    ,   1220.9)
        (0.375   ,    1327.8)
        (0.38     ,  1443.3)
        (0.385   ,    1568.9)
         (0.39    ,   1716.3)
        (0.395   ,    1888.1)
         ( 0.4      ,   2087)
        (0.405    ,   2316.7)
         (0.41     ,  2582.2)
        (0.415    ,   2909.4)
         (0.42     ,  3286.8)
        (0.425    ,   3745.2)
         (0.43     ,  4299.9)
        (0.435    ,   4991.9)
         (0.44     ,  5871.5)
        (0.445    ,   6993.5)

};

\addplot [
color=purple,
dash pattern=on 5pt off 5pt on 5pt,
line width=2.0pt
]
coordinates{
 (0.34,549.578)
 (0.345,584.578)
 (0.35,623.578)
 (0.355,666.578)
 (0.36,714.578)
 (0.365,768.578)
 (0.37,827.578)
 (0.375,894.578)
 (0.38,970.578)
 (0.385,1061.44)
 (0.39,1159.44)
 (0.395,1271.44)
 (0.4,1399.44)
 (0.405,1549.44)
 (0.41,1725.44)
 (0.415,1938.3)
 (0.42,2185.3)
 (0.425,2484.3)
 (0.43,2848.3)
 (0.435,3306.16)
 (0.44,3876.16)
 (0.445,4614.02)

};

\addplot [
color=mycolor3,
mark size=5.0pt,
only marks,
mark=asterisk,
mark options={solid,draw=red}
]
coordinates{
 (0.34,1024)
 (0.385,2048)
 (0.419,4096)
 (0.4425,8192)

};

\end{axis}

\end{tikzpicture}
\end{NoHyper}
		\label{BSC_k402}
	\end{subfigure}
	\begin{subfigure}[b]{0.49\textwidth}\centering  
%

\begin{tikzpicture}

\definecolor{mycolor1}{rgb}{0.15,0.15,0.15}
\definecolor{mycolor2}{rgb}{0,0.447,0.741}
\definecolor{mycolor3}{rgb}{0.85,0.325,0.098}

\pgfplotsset{every axis grid/.style={style=solid}}

\begin{axis}[
scale only axis,
every outer x axis line/.append style={mycolor1},
every x tick label/.append style={font=\color{mycolor1}},
every outer y axis line/.append style={mycolor1},
every y tick label/.append style={font=\color{mycolor1}},
width=2.5in,
height=1.8in,
xlabel={Crossover probability $\delta$},
ylabel={Normalized $n$},
xmin=0.34, xmax=0.46,
ymin=0.5, ymax=1.4,
xmajorgrids,
xminorgrids,
ymajorgrids,
yminorgrids,
legend style={nodes=right},
legend cell align=right,
legend pos=south east,
xtick = {0.36,0.38,0.4,0.42,0.44,0.46},
xticklabels={.36,.38,.4,.42,.44,.46}]

\addplot [
color=green,
solid,
line width=2.0pt,
mark=star,
mark options={solid}
]
coordinates{
(  0.34 ,  1.026  )
(  0.345 ,  1.0266  )
(  0.35 ,  1.0269  )
(  0.355 ,  1.0271  )
(  0.36 ,  1.027  )
(  0.365 ,  1.0276  )
(  0.37 ,  1.0279  )
(  0.375 ,  1.0279  )
(  0.38 ,  1.0283  )
(  0.385 ,  1.0284  )
(  0.39 ,  1.0287  )
(  0.395 ,  1.0297  )
(  0.4 ,  1.0296  )
(  0.405 ,  1.03  )
(  0.41 ,  1.0303  )
(  0.415 ,  1.0306  )
(  0.42 ,  1.0309  )
(  0.425 ,  1.0312  )
(  0.43 ,  1.0312  )
(  0.435 ,  1.0313  )
(  0.44 ,  1.0316  )
(  0.445 ,  1.032  )
};

\addplot [
color=green,
solid,
line width=2.0pt,
mark options={solid}
]
coordinates{
	(  0.34 ,  1  )
	(  0.345 ,  1  )
	(  0.35 ,  1  )
	(  0.355 ,  1  )
	(  0.36 ,  1  )
	(  0.365 ,  1  )
	(  0.37 ,  1  )
	(  0.375 ,  1  )
	(  0.38 ,  1  )
	(  0.385 ,  1  )
	(  0.39 ,  1  )
	(  0.395 ,  1  )
	(  0.4 ,  1  )
	(  0.405 ,  1  )
	(  0.41 ,  1  )
	(  0.415 ,  1  )
	(  0.42 ,  1  )
	(  0.425 ,  1  )
	(  0.43 ,  1  )
	(  0.435 ,  1  )
	(  0.44 ,  1  )
	(  0.445 ,  1  )
};

\addplot [
color=blue,
solid,
line width=2.0pt,
mark = square,
]
coordinates{
(  0.34 ,  0.93623  )
(  0.345 ,  0.9386  )
(  0.35 ,  0.93958  )
(  0.355 ,  0.93712  )
(  0.36 ,  0.9432  )
(  0.365 ,  0.9459  )
(  0.37 ,  0.9457  )
(  0.375 ,  0.94979  )
(  0.38 ,  0.95079  )
(  0.385 ,  0.94798  )
(  0.39 ,  0.94823  )
(  0.395 ,  0.95023  )
(  0.4 ,  0.95166  )
(  0.405 ,  0.95298  )
(  0.41 ,  0.95284  )
(  0.415 ,  0.95735  )
(  0.42 ,  0.95741  )
(  0.425 ,  0.95859  )
(  0.43 ,  0.95809  )
(  0.435 ,  0.95869  )
(  0.44 ,  0.96065  )
(  0.445 ,  0.96144  )
};

\addplot [
color=purple,
dash pattern=on 5pt off 5pt on 5pt,
line width=2.0pt
]
coordinates{
(  0.34 ,  0.64962  )
(  0.345 ,  0.64737  )
(  0.35 ,  0.64619  )
(  0.355 ,  0.64404  )
(  0.36 ,  0.64318  )
(  0.365 ,  0.64262  )
(  0.37 ,  0.64104  )
(  0.375 ,  0.6399  )
(  0.38 ,  0.63938  )
(  0.385 ,  0.64135  )
(  0.39 ,  0.64057  )
(  0.395 ,  0.63988  )
(  0.4 ,  0.63814  )
(  0.405 ,  0.63737  )
(  0.41 ,  0.63669  )
(  0.415 ,  0.63781  )
(  0.42 ,  0.63656  )
(  0.425 ,  0.63586  )
(  0.43 ,  0.63465  )
(  0.435 ,  0.63495  )
(  0.44 ,  0.63419  )
(  0.445 ,  0.63432  )
};

\addplot [
color=mycolor3,
mark size=5.0pt,
only marks,
mark=asterisk,
mark options={solid,draw=red}
]
coordinates{
  (  0.34 ,  1.2104  )
  (  0.385 ,  1.2375  )
  (  0.419 ,  1.2231  )
  (  0.4425 ,  1.231  )

};

\end{axis}

\end{tikzpicture}
	\end{subfigure}
	\caption{Comparison for low-capacity BSC. The number of information bits is $k=40$ and the target error probability is $p_e = 0.01$. For the right plot, with the same legend entries as the left plot, all the blocklengths $n$ in the left plot are normalized by lower bound given by Theorem~\ref{BSC conv}. \label{fig:BSC}}
\end{figure}


Figure~\ref{fig:saddle} represents the same setting as above, i.e., sending $k=40$ information bits over the BEC($\epsilon$) with the desired probability of error $p_e = 0.01$ for BEC (left) and BSC (right). The vertical axis in both cases is the normalized blocklength and the horizontal axis is the erasure probability $\epsilon$ and the crossover probability $\delta$ for the left and right plots, respectively. In both plots, the corresponding channel parameter varies in an extremely low-capacity region. In such a region, we plotted the RCU achievability (upper) and converse (lower) bounds, our approximation from Theorem~\ref{Bounds for BEC} (for BEC) and Theorem~\ref{bounds for BSC} (for BSC), the saddlepoint approximation \cite{ScarlettPaper,scarlett2014saddlepoint}, and  a simplified version of it known as the exact asymptotics which we simply refer to as the ``simplified saddlepoint approximation'' to avoid confusion. Since the simplified version of the saddlepoint approximation is by construction less precise than the original version, we focus only on the comparison of our results to the original saddlepoint approximation.

In the BEC case, i.e., Figure~\ref{fig:saddle} (left), it can be seen that our approximations are dramatically tight and converge to the true bounds as the capacity goes to zero. However, the saddlepoint approximation does not perform well enough in this regime. This behavior can be described by noticing that despite the fact that the saddlepoint approximation uses an innovative approach to estimate the RCU bounds in all rates, it is still based on Gaussian laws; however, the channel coding bounds for BEC, as this paper shows, are best described by the Poisson laws rather than Gaussian laws in this regime.

In the BSC case where we provide a single prediction of the optimal blocklength that satisfies both RCU achievability and converse bounds, Figures~\ref{fig:saddle} shows that this prediction remains close and slowly converges to the true bounds as the capacity decreases. On the other hand, although the saddlepoint approximation works well when the capacity is not small enough, it starts to break down in the region where the capacity becomes extremely small. Note that here both saddlepoint and our method enjoy the Gaussian laws as the baseline of the analysis, however, the weight pattern of the terms in the true bounds and the saddlepoint approximation starts to deviate as the capacity gets closer to zero and thus the approximation starts to lose precision.

\begin{figure}[h!]
	\centering 
	\begin{subfigure}[b]{0.49\textwidth}
		\centering
%
\begin{NoHyper}
\begin{tikzpicture}

\definecolor{mycolor1}{rgb}{0.15,0.15,0.15}
\definecolor{mycolor2}{rgb}{0,0.447,0.741}
\definecolor{mycolor3}{rgb}{0.85,0.325,0.098}

\pgfplotsset{every axis grid/.style={style=solid}}

\begin{axis}[
scale only axis,
every outer x axis line/.append style={mycolor1},
every x tick label/.append style={font=\color{mycolor1}},
every outer y axis line/.append style={mycolor1},
every y tick label/.append style={font=\color{mycolor1}},
width=2.5in,
height=1.8in,
xlabel={Erasure probability $\epsilon$},
ylabel={Normalized $n$},
xmin=0.966, xmax=.998,
ymin=0.9, ymax=2,
xmajorgrids,
xminorgrids,
ymajorgrids,
yminorgrids,
legend entries={\scriptsize{upper bound from Theorem~\ref{Bounds for BEC}}, \scriptsize{lower bound from Theorem~\ref{Bounds for BEC}}, \scriptsize{upper bound from Theorem~\ref{BEC achiv}},\scriptsize{lower bound from Theorem~\ref{BEC conv}}, \scriptsize{saddlepoint approximation \cite{scarlett2014saddlepoint}},\scriptsize{simplified saddlepoint approx. \cite{scarlett2014saddlepoint}} },
legend style={nodes={scale=0.85, transform shape}},
legend style={nodes={right}},
legend cell align=right,
legend pos=north east,
xtick = {0.966,       0.974,           0.982,       0.99,          0.998},
xticklabels={ .966 ,       .974,              .982,    .99,          .998}
]

\addplot [
color=blue,
solid,
line width=2.0pt,
mark=triangle,
mark options={solid}
]
coordinates{
(0.966000,1.099506)
(0.968000,1.097674)
(0.970000,1.095913)
(0.972000,1.093543)
(0.974000,1.091553)
(0.976000,1.090157)
(0.978000,1.088224)
(0.980000,1.085963)
(0.982000,1.084204)
(0.984000,1.082367)
(0.986000,1.080416)
(0.988000,1.078696)
(0.990000,1.076798)
(0.992000,1.074884)
(0.994000,1.072980)
(0.996000,1.071149)
(0.998000,1.069286)

};

\addplot [
color=blue,
solid,
line width=2.0pt,
mark=square,
mark options={solid}
]
coordinates{
(0.966000,0.947466)
(0.968000,0.950581)
(0.970000,0.953678)
(0.972000,0.956787)
(0.974000,0.959887)
(0.976000,0.963415)
(0.978000,0.966866)
(0.980000,0.969532)
(0.982000,0.972911)
(0.984000,0.975928)
(0.986000,0.978945)
(0.988000,0.982174)
(0.990000,0.985329)
(0.992000,0.988268)
(0.994000,0.991203)
(0.996000,0.994137)
(0.998000,0.997106)

};
\addplot [
color=green,
solid,
mark=star,
line width=2.0pt,
]
coordinates{
(0.966000,1.067985)
(0.968000,1.067442)
(0.970000,1.067575)
(0.972000,1.067616)
(0.974000,1.067485)
(0.976000,1.067509)
(0.978000,1.067864)
(0.980000,1.067465)
(0.982000,1.067559)
(0.984000,1.067575)
(0.986000,1.067478)
(0.988000,1.067609)
(0.990000,1.067560)
(0.992000,1.067497)
(0.994000,1.067441)
(0.996000,1.067458)
(0.998000,1.067477)

};

\addplot [
color=green,
solid,
line width=2.0pt,
]
coordinates{
(0.966000,1.000000)
(0.968000,1.000000)
(0.970000,1.000000)
(0.972000,1.000000)
(0.974000,1.000000)
(0.976000,1.000000)
(0.978000,1.000000)
(0.980000,1.000000)
(0.982000,1.000000)
(0.984000,1.000000)
(0.986000,1.000000)
(0.988000,1.000000)
(0.990000,1.000000)
(0.992000,1.000000)
(0.994000,1.000000)
(0.996000,1.000000)
(0.998000,1.000000)
};

\addplot [
color=black,
dash pattern=on 3pt off 3pt on 3pt,
line width=2.0pt
]
coordinates{
(0.966000,1.235476)
(0.968000,1.284884)
(0.970000,1.280654)
(0.972000,1.211998)
(0.974000,1.305333)
(0.976000,1.271341)
(0.978000,1.293014)
(0.980000,1.207472)
(0.982000,1.291123)
(0.984000,1.290893)
(0.986000,1.198123)
(0.988000,1.259348)
(0.990000,1.263720)
(0.992000,1.272161)
(0.994000,1.161707)
(0.996000,1.149175)
(0.998000,1.198162)

};

\addplot [
color=mycolor3,
dash pattern=on 1pt off 1pt on 1pt,
line width=2.0pt
]
coordinates{
(0.966000,1.437577)
(0.968000,1.436047)
(0.970000,1.434332)
(0.972000,1.442298)
(0.974000,1.441246)
(0.976000,1.439895)
(0.978000,1.418762)
(0.980000,1.384113)
(0.982000,1.450718)
(0.984000,1.450116)
(0.986000,1.395231)
(0.988000,1.358696)
(0.990000,1.326390)
(0.992000,1.318076)
(0.994000,1.330148)
(0.996000,1.357484)
(0.998000,1.272260)
};

\end{axis}

\end{tikzpicture}
\end{NoHyper}
	\end{subfigure}
	\begin{subfigure}[b]{0.49\textwidth}
	\centering  
%
\begin{NoHyper}
\begin{tikzpicture}

\definecolor{mycolor1}{rgb}{0.15,0.15,0.15}
\definecolor{mycolor2}{rgb}{0,0.447,0.741}
\definecolor{mycolor3}{rgb}{0.85,0.325,0.098}

\pgfplotsset{every axis grid/.style={style=solid}}

\begin{axis}[
scale only axis,
every outer x axis line/.append style={mycolor1},
every x tick label/.append style={font=\color{mycolor1}},
every outer y axis line/.append style={mycolor1},
every y tick label/.append style={font=\color{mycolor1}},
width=2.5in,
height=1.8in,
xlabel={Crossover probability $\delta$},
ylabel={Normalized $n$},
xmin=0.445, xmax=0.495,
ymin=0.92, ymax=1.20,
xmajorgrids,
xminorgrids,
ymajorgrids,
yminorgrids,
legend entries={\scriptsize{prediction from Theorem~\ref{bounds for BSC}}, \scriptsize{upper bound from Theorem~\ref{BSC achiv}}, \scriptsize{lower bound from Theorem~\ref{BSC conv}},  \scriptsize{saddlepoint approximation \cite{scarlett2014saddlepoint}},\scriptsize{simplified saddlepoint approx. \cite{scarlett2014saddlepoint}} },
legend style={nodes={scale=0.85, transform shape}},
legend style={nodes={right}},
legend cell align=right,
legend pos=north west,
xtick = {.445 ,          .455    ,        .465  ,        .475 ,            .485 ,         .495},
xticklabels={.445 ,          .455    ,        .465  ,        .475 ,            .485 ,         .495}
]

\addplot [
color=blue,
solid,
line width=2.0pt,
mark=triangle,
mark options={solid}
]
coordinates{
(0.445000,0.960407)
(0.448125,0.960748)
(0.451250,0.961343)
(0.454375,0.961789)
(0.457500,0.962413)
(0.460625,0.962806)
(0.463750,0.963281)
(0.466875,0.963803)
(0.470000,0.964255)
(0.473125,0.964681)
(0.476250,0.965058)
(0.479375,0.965427)
(0.482500,0.965833)
(0.485625,0.966183)
(0.488750,0.966523)
(0.491875,0.966855)
(0.495000,0.967164)
};

\addplot [
color=green,
solid,
mark=star,
line width=2.0pt,
]
coordinates{
(0.445000,1.032032)
(0.448125,1.032159)
(0.451250,1.032286)
(0.454375,1.032441)
(0.457500,1.032663)
(0.460625,1.032756)
(0.463750,1.032897)
(0.466875,1.033056)
(0.470000,1.033250)
(0.473125,1.033415)
(0.476250,1.033532)
(0.479375,1.033684)
(0.482500,1.033847)
(0.485625,1.034005)
(0.488750,1.034151)
(0.491875,1.034307)
(0.495000,1.034459)

};

\addplot [
color=green,
solid,
line width=2.0pt,
]
coordinates{
(0.445000,1.000000)
(0.448125,1.000000)
(0.451250,1.000000)
(0.454375,1.000000)
(0.457500,1.000000)
(0.460625,1.000000)
(0.463750,1.000000)
(0.466875,1.000000)
(0.470000,1.000000)
(0.473125,1.000000)
(0.476250,1.000000)
(0.479375,1.000000)
(0.482500,1.000000)
(0.485625,1.000000)
(0.488750,1.000000)
(0.491875,1.000000)
(0.495000,1.000000)

};

=

\addplot [
color=black,
dash pattern=on 3pt off 3pt on 3pt,
line width=2.0pt
]
coordinates{
(0.445000,1.033682)
(0.448125,1.034116)
(0.451250,1.033366)
(0.454375,1.033671)
(0.457500,1.035618)
(0.460625,1.035926)
(0.463750,1.033912)
(0.466875,1.035898)
(0.470000,1.035418)
(0.473125,1.035582)
(0.476250,1.035840)
(0.479375,1.045228)
(0.482500,1.047032)
(0.485625,1.044031)
(0.488750,1.046561)
(0.491875,1.092023)
(0.495000,1.137523)
};

\addplot [
color=mycolor3,
dash pattern=on 1pt off 1pt on 1pt,
line width=2.0pt
]
coordinates{
(0.445000,1.053478)
(0.448125,1.054170)
(0.451250,1.053558)
(0.454375,1.053911)
(0.457500,1.055560)
(0.460625,1.055368)
(0.463750,1.056780)
(0.466875,1.055243)
(0.470000,1.055458)
(0.473125,1.058887)
(0.476250,1.057195)
(0.479375,1.059537)
(0.482500,1.070354)
(0.485625,1.073854)
(0.488750,1.087699)
(0.491875,1.092023)
(0.495000,1.137523)
};

\end{axis}

\end{tikzpicture}
\end{NoHyper}
	\end{subfigure}
	\caption{Sending $k=40$ information bits over BEC (left) and BSC (right) in the low-capacity regime for the target error probability $p_e=0.01$ to compare our approximations from Theorem~\ref{Bounds for BEC} and  Theorem~\ref{bounds for BSC} or equivalently Corollary~\ref{optimal n for BSC} with saddlepoint approximations \cite{scarlett2014saddlepoint}. \label{fig:saddle}}
\end{figure}

  \vspace*{-3ex}

\section{Conclusion and Future Work} \label{sec:con}
    \vspace*{-2ex}

In this paper, we specified a notion of the low-capacity regime for channel coding and studied channel coding at such regimes from two major perspectives, namely, finite-length fundamental limits and code constructions. More specifically, finite-length analysis specific to the low-capacity regime was carried out for several types of channels including binary erasure channels (BECs), binary symmetric channels (BSCs), and additive white Gaussian noise (AWGN) channels. Furthermore, in the context of code construction, the optimal number of repetitions was characterized for transmission over binary memoryless symmetric (BMS) channels, in terms of the code block length and the underlying channel capacity. It was further shown that capacity-achieving polar codes naturally adopt the aforementioned optimal number of repetitions. 

There are several directions for future work. In terms of fundamental limits, it is interesting to study different classes of discrete memoryless channels beyond the ones considered in this paper to characterize their fundamental non-asymptotic laws of channel coding in the low-capacity regime. In terms of code constructions, it is important to study concatenation schemes with low-complexity decoding algorithms comparable to those of straightforward repetition schemes, which can potentially lead to higher rates. From the practical implementation perspective, there are various other challenges besides channel coding in order to enable reliable communications at very low rate regimes. This includes detection, synchronization, and multi-user communication. For instance, synchronization requires transmission of certain pilot signals which would be difficult to detect due to having a very low power. Alternatively, one can explore non-coherent communications at very low rates which is another interesting research direction. Studying low-rate communications together with multi-user schemes is another important problem. This becomes relevant especially for IoT applications where massive numbers of low-power users are present in the field. To this end, various approaches including grant-free and uncoordinated multiple access \cite{liu2018sparse, polyanskiy2017perspective, amalladinne2020coded} as well as scalable coded non-orthogonal techniques \cite{jamali2021massive} can be explored. 

\singlespace    
\bibliographystyle{IEEEtran}
\bibliography{ref}

\begin{thebibliography}{10}
\providecommand{\url}[1]{#1}
\csname url@samestyle\endcsname
\providecommand{\newblock}{\relax}
\providecommand{\bibinfo}[2]{#2}
\providecommand{\BIBentrySTDinterwordspacing}{\spaceskip=0pt\relax}
\providecommand{\BIBentryALTinterwordstretchfactor}{4}
\providecommand{\BIBentryALTinterwordspacing}{\spaceskip=\fontdimen2\font plus
\BIBentryALTinterwordstretchfactor\fontdimen3\font minus
  \fontdimen4\font\relax}
\providecommand{\BIBforeignlanguage}[2]{{%
\expandafter\ifx\csname l@#1\endcsname\relax
\typeout{** WARNING: IEEEtran.bst: No hyphenation pattern has been}%
\typeout{** loaded for the language `#1'. Using the pattern for}%
\typeout{** the default language instead.}%
\else
\language=\csname l@#1\endcsname
\fi
#2}}
\providecommand{\BIBdecl}{\relax}
\BIBdecl

\bibitem{ConfVersion}
M.~{Fereydounian}, M.~V. {Jamali}, H.~{Hassani}, and H.~{Mahdavifar}, ``Channel
  coding at low capacity,'' in \emph{Proc. 2019 IEEE Inf. Theory Workshop
  (ITW)}, Aug 2019, pp. 1--5.

\bibitem{dobrushin1961mathematical}
R.~Dobrushin, ``Mathematical problems in the {S}hannon theory of optimal coding
  of information,'' in \emph{Proc. 4th Berkeley Symp. Math., Statist., and
  Probability}, vol.~1, 1961, pp. 211--252.

\bibitem{Strassen}
V.~Strassen, ``Asymptotische absch\"atzungen in {Shannon's}
  informationstheorie,'' \emph{{\rm in~}Trans. 3rd Prague Conf. Inf. Theory},
  pp. 689--723, 1962.

\bibitem{polypaper}
Y.~Polyanskiy, H.~V. Poor, and S.~Verd{\'u}, ``Channel coding rate in the
  finite blocklength regime,'' \emph{IEEE Trans. Inf. Theory}, vol.~56, no.~5,
  pp. 2307--2359, 2010.

\bibitem{polyphd}
Y.~Polyanskiy, ``Channel coding: non-asymptotic fundamental limits,'' Ph.D.
  dissertation, Prinston University, 11 2010.

\bibitem{TanAWGN}
V.~Y.~F. Tan and M.~Tomamichel, ``The third-order term in the normal
  approximation for the {AWGN} channel,'' \emph{IEEE Trans. Inf. Theory},
  vol.~61, no.~5, pp. 2430--2438, May 2015.

\bibitem{Erseghe2}
T.~Erseghe, ``Coding in the finite-blocklength regime: Bounds based on laplace
  integrals and their asymptotic approximations,'' \emph{IEEE Trans. Inf.
  Theory}, vol.~62, no.~12, pp. 6854--6883, Dec 2016.

\bibitem{moulin}
P.~Moulin, ``The log-volume of optimal codes for memoryless channels,
  asymptotically within a few nats,'' \emph{IEEE Trans. Inf. Theory}, vol.~63,
  no.~4, pp. 2278--2313, April 2017.

\bibitem{ScarlettPaper}
J.~Scarlett, A.~Martinez, and A.~Guill\'{e}n~i F\`{a}bregas, ``Mismatched
  decoding: Error exponents, second-order rates and saddlepoint
  approximations,'' \emph{IEEE Trans. Inf. Theory}, vol.~60, no.~5, pp.
  2647--2666, 2014.

\bibitem{BetaBeta}
W.~Yang, A.~Collins, G.~Durisi, Y.~Polyanskiy, and H.~V. Poor, ``{Beta–Beta}
  bounds: Finite-blocklength analog of the golden formula,'' \emph{IEEE Trans.
  Inf. Theory}, vol.~64, pp. 6236--6256, 2018.

\bibitem{berrou1996near}
C.~Berrou and A.~Glavieux, ``Near optimum error correcting coding and decoding:
  Turbo-codes,'' \emph{IEEE Transactions on communications}, vol.~44, no.~10,
  pp. 1261--1271, 1996.

\bibitem{GallagerLDPC}
R.~Gallager, ``Low-density parity-check codes,'' \emph{IRE Trans. on Info.
  Theory}, vol.~8, no.~1, pp. 21--28, January 1962.

\bibitem{MacKay1999}
D.~J.~C. MacKay, ``Good error-correcting codes based on very sparse matrices,''
  \emph{IEEE Trans. on Info. Theory}, vol.~45, no.~2, pp. 399--431, March 1999.

\bibitem{ratasuk2016overview}
R.~Ratasuk, N.~Mangalvedhe, Y.~Zhang, M.~Robert, and J.-P. Koskinen, ``Overview
  of narrowband {IoT in LTE Rel-13},'' in \emph{Proc. IEEE Conf. Standard
  Commun. Netw. (CSCN)}, Berlin, Germany, Oct./Nov. 2016, pp. 1--7.

\bibitem{ratasuk2016nb}
R.~Ratasuk, B.~Vejlgaard, N.~Mangalvedhe, and A.~Ghosh, ``{NB-IoT} system for
  {M2M} communication,'' in \emph{Proc. IEEE Wireless Commun. Netw. Conf.},
  2016, pp. 1--5.

\bibitem{andrews2007development}
K.~S. Andrews, D.~Divsalar, S.~Dolinar, J.~Hamkins, C.~R. Jones, and
  F.~Pollara, ``The development of turbo and {LDPC} codes for deep-space
  applications,'' \emph{Proceedings of the IEEE}, vol.~95, no.~11, pp.
  2142--2156, 2007.

\bibitem{liang2020raptor}
H.~Liang, A.~Liu, X.~Tong, and C.~Gong, ``Raptor-like rateless spinal codes
  using outer systematic polar codes for reliable deep space communications,''
  in \emph{IEEE INFOCOM 2020-IEEE Conference on Computer Communications
  Workshops (INFOCOM WKSHPS)}.\hskip 1em plus 0.5em minus 0.4em\relax IEEE,
  2020, pp. 1045--1050.

\bibitem{yang2013block}
W.~Yang, G.~Durisi, T.~Koch, and Y.~Polyanskiy, ``Block-fading channels at
  finite blocklength,'' in \emph{ISWCS 2013; The Tenth International Symposium
  on Wireless Communication Systems}.\hskip 1em plus 0.5em minus 0.4em\relax
  VDE, 2013, pp. 1--4.

\bibitem{yang2015fading}
W.~Yang, \emph{Fading channels: Capacity and channel coding rate in the
  finite-blocklength regime}.\hskip 1em plus 0.5em minus 0.4em\relax Chalmers
  University of Technology, 2015.

\bibitem{yang2014quasi}
W.~Yang, G.~Durisi, T.~Koch, and Y.~Polyanskiy, ``Quasi-static multiple-antenna
  fading channels at finite blocklength,'' \emph{IEEE Transactions on
  Information Theory}, vol.~60, no.~7, pp. 4232--4265, 2014.

\bibitem{huang2012finite}
Y.-W. Huang and P.~Moulin, ``Finite blocklength coding for multiple access
  channels,'' in \emph{2012 IEEE International Symposium on Information Theory
  Proceedings}.\hskip 1em plus 0.5em minus 0.4em\relax IEEE, 2012, pp.
  831--835.

\bibitem{tan2013dispersions}
V.~Y. Tan and O.~Kosut, ``On the dispersions of three network information
  theory problems,'' \emph{IEEE Transactions on Information Theory}, vol.~60,
  no.~2, pp. 881--903, 2013.

\bibitem{kostina2012fixed}
V.~Kostina and S.~Verd{\'u}, ``Fixed-length lossy compression in the finite
  blocklength regime,'' \emph{IEEE Transactions on Information Theory},
  vol.~58, no.~6, pp. 3309--3338, 2012.

\bibitem{nomura2014second}
R.~Nomura and T.~S. Han, ``Second-order slepian-wolf coding theorems for
  non-mixed and mixed sources,'' \emph{IEEE transactions on information
  theory}, vol.~60, no.~9, pp. 5553--5572, 2014.

\bibitem{wang2016fundamental}
L.~Wang, G.~W. Wornell, and L.~Zheng, ``Fundamental limits of communication
  with low probability of detection,'' \emph{IEEE Transactions on Information
  Theory}, vol.~62, no.~6, pp. 3493--3503, 2016.

\bibitem{tahmasbi2018first}
M.~Tahmasbi and M.~R. Bloch, ``First-and second-order asymptotics in covert
  communication,'' \emph{IEEE Transactions on Information Theory}, vol.~65,
  no.~4, pp. 2190--2212, 2018.

\bibitem{watanabe2015nonasymptotic}
S.~Watanabe, S.~Kuzuoka, and V.~Y. Tan, ``Nonasymptotic and second-order
  achievability bounds for coding with side-information,'' \emph{IEEE
  Transactions on Information Theory}, vol.~61, no.~4, pp. 1574--1605, 2015.

\bibitem{martinez2011saddlepoint}
A.~Martinez and A.~G. i~Fabregas, ``Saddlepoint approximation of random-coding
  bounds,'' in \emph{2011 Information Theory and Applications Workshop}.\hskip
  1em plus 0.5em minus 0.4em\relax IEEE, 2011, pp. 1--6.

\bibitem{scarlett2014saddlepoint}
J.~Scarlett, A.~Martinez, and A.~G. i~F{\`a}bregas, ``The saddlepoint
  approximation: Unified random coding asymptotics for fixed and varying
  rates,'' in \emph{2014 IEEE International Symposium on Information
  Theory}.\hskip 1em plus 0.5em minus 0.4em\relax IEEE, 2014, pp. 1892--1896.

\bibitem{font2018saddle}
J.~Font-Segura, G.~Vazquez-Vilar, A.~Martinez, A.~Guillén~i Fàbregas, and
  A.~Lancho, ``Saddlepoint approximations of lower and upper bounds to the
  error probability in channel coding,'' in \emph{2018 52nd Annual Conference
  on Information Sciences and Systems (CISS)}, 2018, pp. 1--6.

\bibitem{altuug2020exact}
Y.~Altu{\u{g}} and A.~B. Wagner, ``On exact asymptotics of the error
  probability in channel coding: symmetric channels,'' \emph{IEEE Transactions
  on Information Theory}, 2020.

\bibitem{honda2018exact}
J.~Honda, ``Exact asymptotics of random coding error probability for general
  memoryless channels,'' in \emph{2018 IEEE International Symposium on
  Information Theory (ISIT)}, 2018, pp. 1844--1848.

\bibitem{reiffen1963note}
B.~Reiffen, ``A note on ``very noisy" channels,'' \emph{Information and
  Control}, vol.~6, no.~2, pp. 126--130, 1963.

\bibitem{gallager1965simple}
R.~Gallager, ``A simple derivation of the coding theorem and some
  applications,'' \emph{IEEE Trans. Inf. Theory}, vol.~11, no.~1, pp. 3--18,
  1965.

\bibitem{majani1988model}
E.~E. Majani, ``A model for the study of very noisy channels, and
  applications,'' Ph.D. dissertation, California Institute of Technology, 1988.

\bibitem{wyner1988capacity}
A.~D. Wyner, ``Capacity and error exponent for the direct detection photon
  channel—{Parts I and II},'' \emph{IEEE Trans. Inf. Theory}, vol.~34, no.~6,
  pp. 1449--1471, 1988.

\bibitem{sakai2020second}
Y.~Sakai, V.~Y. Tan, and M.~Kovai{\v{c}}evi{\v{c}}, ``Second-and third-order
  asymptotics of the continuous-time {Poisson} channel,'' \emph{IEEE Trans.
  Inf. Theory}, vol.~66, no.~8, pp. 4742--4760, 2020.

\bibitem{wagner2020new}
A.~B. Wagner, N.~V. Shende, and Y.~Altu{\u{g}}, ``A new method for employing
  feedback to improve coding performance,'' \emph{IEEE Trans. Inf. Theory},
  vol.~66, no.~11, pp. 6660--6681, 2020.

\bibitem{shende2017very}
N.~V. Shende and A.~B. Wagner, ``On very noisy channels with feedback,'' in
  \emph{Proc. 55th Annu. Allerton Conf. Commun., Control, and Comput.}\hskip
  1em plus 0.5em minus 0.4em\relax IEEE, 2017, pp. 852--859.

\bibitem{arikan2009channel}
E.~Arikan, ``Channel polarization: A method for constructing capacity-achieving
  codes for symmetric binary-input memoryless channels,'' \emph{IEEE Trans.
  Inf. Theory}, vol.~55, no.~7, pp. 3051--3073, 2009.

\bibitem{Talag}
M.~Talagrand, ``The missing factor in {H}oeffding's inequalities,''
  \emph{Annales de l'I.H.P. Probabilités et statistiques}, vol.~31, no.~4, pp.
  689--702, 1995.

\bibitem{Tan}
V.~Y.~F. Tan, ``Asymptotic estimates in information theory with non-vanishing
  error probabilities,'' \emph{arXiv:1504.02608v1 [cs.IT]}, 2015.

\bibitem{Shannon}
C.~E. Shannon, ``A mathematical theory of communication,'' \emph{Bell System
  Tech. Journal}, no.~27, pp. 379--423, 1948.

\bibitem{WolfSC}
J.~Wolfowitz, ``Notes on a general strong converse,'' \emph{Inf. Contr.},
  vol.~12, pp. 1--4, 1968.

\bibitem{hayashi2009information}
M.~Hayashi, ``Information spectrum approach to second-order coding rate in
  channel coding,'' \emph{IEEE Trans. Inf. Theory}, vol.~55, no.~11, pp.
  4947--4966, 2009.

\bibitem{Arikan2}
E.~Ar{\i}kan, ``Source polarization,'' in \emph{Proc. 2010 IEEE Int. Symp. Inf.
  Theory (ISIT)}, 2010, pp. 899--903.

\bibitem{abbe2011polarization}
E.~Abbe, ``Polarization and randomness extraction,'' in \emph{Proc. 2011 IEEE
  Int. Symp. Inf. Theory (ISIT)}, 2011, pp. 184--188.

\bibitem{mondelli2015achieving}
M.~Mondelli, S.~H. Hassani, I.~Sason, and R.~L. Urbanke, ``Achieving {Marton's}
  region for broadcast channels using polar codes,'' \emph{IEEE Trans. Inf.
  Theory}, vol.~61, no.~2, pp. 783--800, 2015.

\bibitem{goela2015polar}
N.~Goela, E.~Abbe, and M.~Gastpar, ``Polar codes for broadcast channels,''
  \emph{IEEE Trans. Inf. Theory}, vol.~61, no.~2, pp. 758--782, 2015.

\bibitem{STY}
E.~\c{S}a\c{s}o\u{g}lu, E.~Telatar, and E.~Yeh, ``Polar codes for the two-user
  binary-input multiple-access channel,'' \emph{IEEE Trans. Inf. Theory},
  vol.~59, no.~10, pp. 6583--6592, 2013.

\bibitem{MELK}
H.~Mahdavifar, M.~El-Khamy, J.~Lee, and I.~Kang, ``Achieving the uniform rate
  region of general multiple access channels by polar coding,'' \emph{IEEE
  Trans. Commun.}, vol.~64, no.~2, pp. 467--478, 2016.

\bibitem{MV}
H.~Mahdavifar and A.~Vardy, ``Achieving the secrecy capacity of wiretap
  channels using polar codes,'' \emph{IEEE Trans. Inf. Theory}, vol.~57,
  no.~10, pp. 6428--6443, 2011.

\bibitem{andersson2010nested}
M.~Andersson, V.~Rathi, R.~Thobaben, J.~Kliewer, and M.~Skoglund, ``Nested
  polar codes for wiretap and relay channels,'' \emph{IEEE Commun. Lett.},
  vol.~14, no.~8, pp. 752--754, 2010.

\bibitem{mahdavifar2015polar}
H.~Mahdavifar, M.~El-Khamy, J.~Lee, and I.~Kang, ``Polar coding for
  bit-interleaved coded modulation,'' \emph{IEEE Trans. Veh. Technol.},
  vol.~65, no.~5, pp. 3115--3127, 2015.

\bibitem{Posner67}
E.~C. Posner, ``Properties of error-correcting codes at low signal-to-noise
  ratios,'' \emph{SIAM J. Appl. Math.}, vol.~15, no.~4, pp. 775--798, 1967.

\bibitem{ChaoLowSNR}
{Chi-Chao Chao}, R.~J. {McEliece}, L.~{Swanson}, and E.~R. {Rodemich},
  ``Performance of binary block codes at low signal-to-noise ratios,''
  \emph{IEEE Trans. Inf. Theory}, vol.~38, no.~6, pp. 1677--1687, Nov 1992.

\bibitem{tal2015list}
I.~Tal and A.~Vardy, ``List decoding of polar codes,'' \emph{IEEE Trans. Inf.
  Theory}, vol.~61, no.~5, pp. 2213--2226, 2015.

\bibitem{alamdar2011simplified}
A.~Alamdar-Yazdi and F.~R. Kschischang, ``A simplified successive-cancellation
  decoder for polar codes,'' \emph{IEEE Commun. Lett.}, vol.~15, no.~12, pp.
  1378--1380, 2011.

\bibitem{el2017relaxed}
M.~El-Khamy, H.~Mahdavifar, G.~Feygin, J.~Lee, and I.~Kang, ``Relaxed polar
  codes,'' \emph{IEEE Trans. Inf. Theory}, vol.~63, no.~4, pp. 1986--2000,
  2017.

\bibitem{dumer2017polar}
I.~Dumer, ``Polar codes with a stepped boundary,'' in \emph{Proc. 2017 IEEE
  Int. Symp. Inf. Theory (ISIT)}.\hskip 1em plus 0.5em minus 0.4em\relax IEEE,
  2017, pp. 2613--2617.

\bibitem{abbasi2022polar}
F.~Abbasi, H.~Mahdavifar, and E.~Viterbo, ``Polar coded repetition,''
  \emph{IEEE Transactions on Communications}, vol.~70, no.~10, pp. 6399--6409,
  2022.

\bibitem{dumer2021codes}
I.~Dumer and N.~Gharavi, ``Codes approaching the shannon limit with polynomial
  complexity per information bit,'' in \emph{2021 IEEE International Symposium
  on Information Theory (ISIT)}.\hskip 1em plus 0.5em minus 0.4em\relax IEEE,
  2021, pp. 238--243.

\bibitem{abbasi2022hybrid}
F.~Abbasi, H.~Mahdavifar, and E.~Viterbo, ``Hybrid non-binary repeated polar
  codes,'' \emph{IEEE Transactions on Wireless Communications}, vol.~21, no.~9,
  pp. 7582--7594, 2022.

\bibitem{hassani2014finite}
S.~H. Hassani, K.~Alishahi, and R.~L. Urbanke, ``Finite-length scaling for
  polar codes,'' \emph{IEEE Trans. Inf. Theory}, vol.~60, no.~10, pp.
  5875--5898, 2014.

\bibitem{mondelli2016unified}
M.~Mondelli, S.~H. Hassani, and R.~L. Urbanke, ``Unified scaling of polar
  codes: Error exponent, scaling exponent, moderate deviations, and error
  floors,'' \emph{IEEE Trans. Inf. Theory}, vol.~62, no.~12, pp. 6698--6712,
  2016.

\bibitem{goldin2014improved}
D.~Goldin and D.~Burshtein, ``Improved bounds on the finite length scaling of
  polar codes,'' \emph{IEEE Trans. Inf. Theory}, vol.~60, no.~11, pp.
  6966--6978, 2014.

\bibitem{guruswami2015polar}
V.~Guruswami and P.~Xia, ``Polar codes: Speed of polarization and polynomial
  gap to capacity,'' \emph{IEEE Trans. Inf. Theory}, vol.~61, no.~1, pp. 3--16,
  2015.

\bibitem{mahdavifar2020polar}
H.~Mahdavifar, ``Polar coding for non-stationary channels,'' \emph{IEEE
  Transactions on Information Theory}, vol.~66, no.~11, pp. 6920--6938, 2020.

\bibitem{liu2018sparse}
L.~Liu, E.~G. Larsson, W.~Yu, P.~Popovski, C.~Stefanovic, and E.~De~Carvalho,
  ``Sparse signal processing for grant-free massive connectivity: A future
  paradigm for random access protocols in the internet of things,'' \emph{IEEE
  Signal Processing Magazine}, vol.~35, no.~5, pp. 88--99, 2018.

\bibitem{polyanskiy2017perspective}
Y.~Polyanskiy, ``A perspective on massive random-access,'' in \emph{2017 IEEE
  International Symposium on Information Theory (ISIT)}.\hskip 1em plus 0.5em
  minus 0.4em\relax IEEE, 2017, pp. 2523--2527.

\bibitem{amalladinne2020coded}
V.~K. Amalladinne, J.-F. Chamberland, and K.~R. Narayanan, ``A coded compressed
  sensing scheme for unsourced multiple access,'' \emph{IEEE Trans. Inf.
  Theory}, vol.~66, no.~10, pp. 6509--6533, 2020.

\bibitem{jamali2021massive}
M.~V. Jamali and H.~Mahdavifar, ``Massive coded-{NOMA} for low-capacity
  channels: A low-complexity recursive approach,'' \emph{IEEE Trans. Commun.},
  vol.~69, no.~6, pp. 3664--3681, 2021.

\bibitem{Wolfowitz}
J.~Wolfowitz, \emph{Coding Theorems of Information Theory}, 3rd~ed.\hskip 1em
  plus 0.5em minus 0.4em\relax New York: Springer-Verlag, 1978.

\bibitem{Gal}
R.~G. Gallager, \emph{Information Theory and Reliable Communication}.\hskip 1em
  plus 0.5em minus 0.4em\relax New York: Wiley, 1968.

\bibitem{Poor}
H.~V. Poor, \emph{An Introduction to Signal Detection and Estimation}.\hskip
  1em plus 0.5em minus 0.4em\relax New York: Springer-Verlag, 1994.

\bibitem{Elias}
P.~Elias, ``Coding for two noisy channels,'' in \emph{Proc. 3d London Symp.
  Inf. Theory. Washington, DC}, Sep. 1955, pp. 61--76.

\bibitem{Zach}
V.~Zacharovas and H.-K. Hwang, ``A {C}harlier–{P}arseval approach to
  {P}oisson approximation and its applications,'' \emph{Lithuanian Math. J.},
  vol.~50, no.~1, pp. 88--119, 2010.

\bibitem{stein}
C.~Stein, \emph{Approximate Computation of Expectations}.\hskip 1em plus 0.5em
  minus 0.4em\relax Inst. of Mathematical Statistic, 1987.

\bibitem{hirschhorn2014refinement}
M.~D. Hirschhorn and M.~B. Villarino, ``The missing factor in {H}oeffding's
  inequalities,'' \emph{The Ramanujan J.}, no.~34, pp. 73--81, 2014.

\bibitem{durrett}
R.~Durrett, \emph{Probability: Theory and Examples}.\hskip 1em plus 0.5em minus
  0.4em\relax Cambridge University Press, 2010.

\bibitem{richardson2008modern}
T.~Richardson and R.~Urbanke, \emph{Modern Coding Theory}.\hskip 1em plus 0.5em
  minus 0.4em\relax Cambridge University Press, 2008.

\bibitem{korada2009polar}
S.~B. Korada, ``Polar codes for channel and source coding,'' Ph.D.
  dissertation, Ecole Polytechnique Federale de Lausanne, 7 2009.

\bibitem{hassani2013polarization}
S.~H. Hassani, ``Polarization and spatial coupling: Two techniques to boost
  performance,'' Ph.D. dissertation, Ecole Polytechnique Federale de Lausanne,
  9 2013.

\end{thebibliography}

\newpage
\appendix
\section{Proofs for BEC}\label{BEC proofs}
In this section, we prove the converse and achievability bounds of Theorem~\ref{Bounds for BEC}. We first state Theorems~\ref{BE for Poisson}--\ref{BEC achiv} which are necessary for these proofs. For the results in coding theory, we generally refer to \cite{polyphd} as it has well collected and presented the corresponding proofs. See also \cite{Tan}, \cite{Wolfowitz}, \cite{Gal}, \cite{Poor}, and \cite{Elias}.

\begin{theorem}[Strong Poisson Convergence]\label{BE for Poisson}
	For $1\le i \le n$, let $X_i$'s be independent random variables with $P(X_i =1)=1-P(X_i=0)=p_i$. Define $S_n = \sum_{i=1}^n X_i$ . Also define $\lambda_k=\sum_{i=1}^np_i^k$ and let $\lambda=\lambda_1$ and $\theta=\lambda_2/\lambda_1$. Then, 
	\begin{align}\label{feq}
	\bigg|\text{Pr} \{S_n\leq m\}-\text{Pr} \{X\leq m\}\bigg|\leq \frac{\sqrt{2}\,(\sqrt{e}-1)\,\theta}{(1-\theta)^{3/2}}\sqrt{\psi(m)},
	\end{align}
	and
	\begin{equation}
	\bigg|\text{Pr} \{S_n=m\}-\text{Pr} \{X=m\}\bigg|\leq \frac{\sqrt{6}\,(\sqrt{e}-1)\,\theta}{(1-\theta)^2\sqrt{\lambda}}\sqrt{\psi(m)},
	\end{equation}
	where $X\sim \text{Poisson}(\lambda)$ and the function $\psi(m)$ is defined as
	\begin{align}\label{defpsi}
	\psi(m):=\min\bigg\{\text{Pr} \{X\leq m\},\text{Pr} \{X> m\}\bigg\}.
	\end{align}
\end{theorem}
\begin{proof}
	See \cite[Theorem~3.4, Lemma~3.7]{Zach}.
\end{proof}

\begin{theorem}[Poisson Convergence]
	For $1\le i \le n$, let $X_i$'s be independent random variables with $P(X_i =1)=1-P(X_i=0)=p_i$. Define $S_n = \sum_{i=1}^n X_i$ and $\lambda_k=\sum_{i=1}^np_i^k$. Let $\mu_n$ be the distribution of $S_n$ and $\nu_n$ be the Poisson distribution with mean $\lambda_1$. Then the following holds.
	\begin{equation}
	    \sup_A \bigg|\mu_n(A)-\nu_n(A)\bigg| \leq \frac{\lambda_2}{\lambda_1}
	\end{equation}
\end{theorem}
\begin{proof}
	See \cite[page 89]{stein}.
\end{proof}

\begin{theorem}[Converse Bound for BEC]\label{BEC conv}
	For any $(M,p_e)$-code over the BEC$^n$($\epsilon$), we have
	\begin{equation}
	\label{BEC converse}
	p_e \ge \sum_{r\,\, < \log_2 M}\binom{n}{r}\epsilon^{n-r}(1-\epsilon)^{r}\bigg(1-\frac{2^{r}}{M}\bigg).
	\end{equation}
\end{theorem}
\begin{proof}
	See \cite[Theorem~43]{polyphd}.    
\end{proof}

\begin{theorem}[RCU Achievability Bound for BEC]\label{BEC achiv} 
	There exists an $(M,p_e)$-code over BEC$^n$($\epsilon$) such that
	\begin{equation}
	\label{BEC ach.}
	p_e \le \sum_{r=0}^n\binom{n}{r}\epsilon^{n-r}(1-\epsilon)^{r}2^{-[r-\log_2(M-1)]^+}.
	\end{equation}
\end{theorem}
\begin{proof}
	See \cite[Corollary~42]{polyphd}.
\end{proof}

\begin{proof}[{\bf Achievability Bound of Theorem~\ref{Bounds for BEC}}]
	Consider $n$ transmissions over the BEC($\epsilon$) which are indexed by $i=1,\cdots,n$. For the $i$-th transmission, we let $X_i$ be a Bernoulli random variable which is $0$ if the output of the $i^\text{th}$ channel is an erasure and is $1$ otherwise, i.e., $\text{Pr}\{X_i=1\}=1-\epsilon$. Suppose $S_n = \sum_{i=1}^n X_i$ and denote $\kappa=n(1-\epsilon)$. We will use the result of Theorem~\ref{BEC achiv} and show that if a number $M_1$ satisfies (\ref{M1}), then it will dissatisfy  the inequality in \eqref{BEC ach.}. As a result, we obtain $M_1 \leq M^*(n,p_e)$.  Now, by considering \eqref{BEC ach.}, we define
	\begin{align}
	I_1 &= \sum_{r\, < \,\log_2 M_1}\binom{n}{r}\epsilon^{n-r}(1-\epsilon)^r,\\
	I_2&=\sum_{\log_2 M_1\,\leq\,r\,\leq \,n}\binom{n}{r}\epsilon^{n-r}(1-\epsilon)^r\frac{M_1}{2^{r}}.
	\end{align}    
	Note that $I_1 = \text{Pr}\left\{S_n < \log_2 M_1\right\}$. Let $X$ be an arbitrary Poisson random variable with mean $\kappa$, that is $X\sim \text{Poisson}(\kappa)$. Having this, we can write the following 
	\begin{align}
	\label{ff0}I_1 &=\text{Pr}\left\{S_n < \log_2 M_1\right\} \leq \text{Pr}\set{X < \log_2 M_1} + \bigg|  \text{Pr}\set{S_n < \log_2 M_1} -  \text{Pr}\set{X < \log_2 M_1} \bigg|,\\
	\nonumber I_2 &= \sum_{\log_2 M_1\,\leq\,r\,\leq \,n} \text{Pr}\set{S_n = r} \frac{M_1}{2^r} \\
	\label{fff} &\leq \sum_{\log_2 M_1\,\leq\,r\,\leq \,n} \text{Pr}\set{X =  r} \frac{M_1}{2^r} +  \sum_{\log_2 M_1\,\leq\,r\,\leq \,n}\bigg| \text{Pr}\set{X =  r} - \text{Pr}\set{S_n = r} \bigg| \frac{M_1}{2^r}\\
	\label{ff1}&\leq \sum_{\log_2 M_1\,\leq\,r\,<\, \infty} \text{Pr}\set{X =  r} \frac{M_1}{2^r} +  \sum_{\log_2 M_1\,\leq\,r\,< \,\infty}\bigg| \text{Pr}\set{X =  r} - \text{Pr}\set{S_n = r} \bigg| \frac{M_1}{2^r}.
	\end{align}
	Note that in both \eqref{ff0} and \eqref{fff}, we used a simple fact since that for any two real numbers $a$ and $b$, we have $a \leq b+|a-b|$. 
	Putting \eqref{ff0} and \eqref{ff1} together, we obtain
	\begin{align}
	\label{I1I2} I_1 + I_2  \leq \mathbb{E}\bigg[\mathbf{1}\left(X < \log_{2}M_1\right)+\mathbf{1}\left(X \geq \log_{2}M_1\right)\frac{M_1}{2^X}\bigg]  +J_1+J_2,
	\end{align}
	where
	\begin{align}
	J_1&=\bigg|  \text{Pr}\set{S_n < \log_2 M_1} -  \text{Pr}\set{X < \log_2 M_1} \bigg|, \\
	J_2&= \sum_{\log_2 M_1\,\leq\,r\,< \,\infty}\bigg| \text{Pr}\set{X =  r} - \text{Pr}\set{S_n = r} \bigg| \frac{M_1}{2^r}.
	\end{align}
	Using Theorem~\ref{BE for Poisson}, we have
	\begin{align}
	\label{J1}J_1=& \bigg|  \text{Pr}\set{S_n < \log_2 M_1} -  \text{Pr}\set{X < \log_2 M_1} \bigg| \leq \alpha_1\sqrt{ \text{Pr}\set{X < \log_2 M_1}},
	\end{align}
	and
	\begin{align}
	\label{J2p}&  \bigg|\text{Pr}\set{X =  r} - \text{Pr}\set{S_n = r} \bigg|  \leq  \alpha_2 \sqrt{ \text{Pr}\set{X \leq r}},
	\end{align}
	where
	\begin{align}
	\alpha_1 = \frac{\sqrt{2}}{\epsilon^{3/2}}\left(\sqrt{e}-1\right)(1-\epsilon),\quad \alpha_2 = \frac{\sqrt{6}}{\epsilon^2\sqrt{\kappa}}\left(\sqrt{e}-1\right)(1-\epsilon).
	\end{align}
	Before continuing, let us clearly explain how applying Theorem~\ref{BE for Poisson} leads to \eqref{J1}. Deriving \eqref{J2p} from Theorem~\ref{BE for Poisson} is similar. Only for the sake of this explanation and only in the following paragraph, we borrow the local notation of Theorem~\ref{BE for Poisson}, i.e., $\theta,\lambda,\lambda_1,\lambda_2$, and $\psi(\cdot)$. Here, we consider the previously defined $S_n=\sum_{i=1}^{n} X_i$ where $p_i = \operatorname{Pr} (X_i = 1) = 1-\epsilon$ which leads to $\lambda=\lambda_1=\sum_{i=1}^{n} p_i = n(1-\epsilon) = \kappa$ in Theorem~\ref{BE for Poisson}. Hence, our previously defined $X\sim\operatorname{Poisson}(\kappa)$ matches with the $X\sim \operatorname{Poisson}(\lambda)$ in Theorem~\ref{BE for Poisson}. Moreover, $\lambda_2 = n(1-\epsilon)^2$ and thus $\theta = \lambda_2/\lambda_1 = 1-\epsilon$. Hence, \eqref{feq} becomes 
	\begin{align}\label{o1}
	\nonumber\bigg|\text{Pr} \{S_n\leq m\}-\text{Pr} \{X\leq m\}\bigg|&\leq \frac{\sqrt{2}\,(\sqrt{e}-1)\,(1-\epsilon)}{\epsilon^{3/2}}\sqrt{\psi(m)} \\
    &= \alpha_1\sqrt{\psi(m)} \leq \alpha_1\sqrt{\operatorname{Pr}(X\leq m)}.
	\end{align}
	The leftmost inequality is due to the fact that, by the definition of $\psi(\cdot)$ in \eqref{defpsi}, we have $\psi(m) \leq \operatorname{Pr}(X\leq m)$. Finally, note that \eqref{o1} holds for every real number $m$ and put $m=\lfloor \log_2 M_1 \rfloor +1/2$. This makes the inequalities strict, i.e., $\operatorname{Pr}(S_n \leq m) = \operatorname{Pr}(S_n < \log_2 M_1) $ and  $\operatorname{Pr}(X \leq m) = \operatorname{Pr}(X < \log_2 M_1) $. Replacing these into \eqref{o1} gives \eqref{J1}.
	
	We now upper-bound $J_2$ using \eqref{J2p}, as follows:
	\begin{align}
	\nonumber J_2&= 2\sum_{\log_2 M_1\,\leq\,r\,< \,\infty}\frac{M_1}{2^{r+1}}\bigg| \text{Pr}\set{X =  r} - \text{Pr}\set{S_n = r} \bigg|\\
	\nonumber & \leq 2\alpha_2 \sum_{\log_2 M_1\,\leq\,r\,< \,\infty}\frac{M_1}{2^{r+1}}\sqrt{ \text{Pr}\set{X \leq r}}\\
	\label{last}&\leq 2\alpha_2 \sqrt{\sum_{\log_2 M_1\,\leq\,r\,< \,\infty} \frac{M_1}{2^{r+1}}\text{Pr}\set{X \leq r}}.
	\end{align}
	For obtaining the last part, note that $\sqrt{x}$ is a concave function and
	\begin{equation}
	    \sum_{\log_2 M_1\,\leq\,r\,< \,\infty} \frac{M_1}{2^{r+1}}\leq 1.
	\end{equation}
	Thus, (\ref{last}) follows from Jensen inequality for $\sqrt{x}$, that is, $\mathbb{E}\left[\sqrt{Z}\right]\leq \sqrt{\mathbb{E}\left[Z\right]}$.
	Also, consider
	\begin{align}
	\nonumber \sum_{\log_2 M_1\,\leq\,r\,< \,\infty} \frac{M_1}{2^{r+1}}\text{Pr}\set{X \leq r}&=\sum_{\log_2 M_1\,\leq\,r\,< \,\infty} \frac{M_1}{2^{r+1}}\sum_{i=0}^re^{-\kappa} \frac{\kappa^i}{i!}\\
	\nonumber&=\sum_{i=0}^{\infty}e^{-\kappa} \frac{\kappa^i}{i!}\left(\mathbf{1}\left(i< \log_2 M_1\right)+\frac{\mathbf{1}\left(i\geq\log_2M_1\right)}{2^{i-\log_2 M_1}}\right)\\
	\nonumber&=\sum_{i<\log_2 M_1}e^{-\kappa} \frac{\kappa^i}{i!}+M_1\sum_{i\geq\log_2 M_1}e^{-\kappa} \frac{\kappa^i}{i!}\frac{1}{2^i}\\
	\nonumber&=\mathcal{P}_{\kappa}(\log_2 M_1)+M_1e^{-\kappa/2}\sum_{i\geq\log_2 M_1}e^{-\kappa/2} \frac{(\kappa/2)^i}{i!}\\
	\label{tt}&=\mathcal{P}_{\kappa}(\log_2 M_1)+M_1e^{-\kappa/2}\left(1-\mathcal{P}_{\kappa/2}(\log_2M_1)\right).
	\end{align}
	From \eqref{last} and \eqref{tt}, we get
	\begin{align}
	\label{J2}J_2&\leq 2\alpha_2 \sqrt{\mathcal{P}_{\kappa}(\log_2 M_1)+M_1e^{-\kappa/2}\left(1-\mathcal{P}_{\kappa/2}(\log_2M_1)\right)}.
	\end{align}
	Also, considering the notation in \eqref{Poisson cdf}, we can write
	\begin{align}
	\nonumber&\mathbb{E}\bigg[\mathbf{1}\left(X < \log_{2}M_1\right)+\mathbf{1}\left(X \geq \log_{2}M_1\right)\frac{M_1}{2^X}\bigg] \\
	\nonumber&\qquad\qquad\qquad=\mathbb{E}\bigg[\mathbf{1}\left(X < \log_{2}M_1\right)\bigg] +\mathbb{E}\bigg[\frac{M_1}{2^X}\bigg]-\mathbb{E}\bigg[\left(X < \log_{2}M_1\right)\frac{M_1}{2^X}\bigg] \\
	\nonumber&\qquad\qquad\qquad= \mathcal{P}_{\kappa}( \log_{2}M_1)+\sum_re^{-\kappa} \frac{\kappa^r}{r!}\cdot\frac{M_1}{2^r}-\sum_{r\,<\,\log_{2}M_1}e^{-\kappa} \frac{\kappa^r}{r!}\cdot\frac{M_1}{2^r}\\
	\nonumber&\qquad\qquad\qquad= \mathcal{P}_{\kappa}( \log_{2}M_1)+M_1e^{-\kappa/2}\sum_re^{-\kappa/2} \frac{(\kappa/2)^r}{r!}-M_1e^{-\kappa/2}\sum_{r\,<\,\log_{2}M_1}e^{-\kappa/2} \frac{(\kappa/2)^r}{r!}\\
	&\qquad\qquad\qquad= \mathcal{P}_{\kappa}( \log_{2}M_1)+M_1e^{-\kappa/2}-M_1e^{-\kappa/2}\mathcal{P}_{\kappa/2}(\log_2 M_1). \label{we4}
	\end{align}
	Now, \eqref{I1I2}, \eqref{J1}, \eqref{J2}, and \eqref{we4} together result in
	\begin{align}
	\nonumber I_1 + I_2 &\leq \mathcal{P}_{\kappa}( \log_{2}M_1)+M_1e^{-\kappa/2}\left(1-\mathcal{P}_{\kappa/2}(\log_2 M_1) \right)\\
	\nonumber &+\alpha_1\sqrt{\mathcal{P}_{\kappa}(\log_2M_1)}+2\alpha_2 \sqrt{\mathcal{P}_{\kappa}(\log_2 M_1)+M_1e^{-\kappa/2}\left(1-\mathcal{P}_{\kappa/2}(\log_2M_1)\right)}\\
	\nonumber&= \mathfrak{P}_1(M_1)+2\alpha_2 \sqrt{ \mathfrak{P}_1(M_1)} +\alpha_1\sqrt{\mathcal{P}_{\kappa}(\log_2M_1)}\\
	\label{ff}&\leq p_e,
	\end{align}
	where
	\begin{align}
	\mathfrak{P}_1(M_1)&=\mathcal{P}_{\kappa}( \log_{2}M_1)+M_1e^{-\kappa/2}\left(1-\mathcal{P}_{\kappa/2}(\log_2 M_1) \right).
	\end{align} 
	Note that \eqref{ff} holds by the definition of $M_1$ in \eqref{M1}. Hence, we showed
	\begin{align}
	\sum_{r=0}^n\binom{n}{r}\epsilon^{r}(1-\epsilon)^{n-r}2^{-[r-\log_2(M_1-1)]^+} \leq  I_1+I_2 \leq p_e,
	\end{align}
	which means $M_1$ dissatisfies the inequality in \eqref{BEC ach.}. Hence, $M_1\leq M^*(n,p_e) $.
\end{proof}

\begin{proof}[{\bf Converse Bound of Theorem~\ref{Bounds for BEC} }]
	Consider $n$ transmissions over the BEC($\epsilon$) which are indexed by $i=1,\cdots,n$. For the $i$-th transmission, we let $X_i$ be a Bernoulli random variable which is $0$ if the output of the $i^\text{th}$ channel is an erasure and is $1$ otherwise, i.e., $\text{Pr}\{X_i=1\}=1-\epsilon$. Suppose $S_n = \sum_{i=1}^n X_i$ and denote $\kappa=n(1-\epsilon)$. We will use the result of Theorem~\ref{BEC conv} and show that if a number $M_2$ satisfies \eqref{M2}, then it will dissatisfy  the inequality in \eqref{BEC converse}. As a result, we obtain $M^*(n,p_e)\leq M_2$. Now, by considering \eqref{BEC converse}, we define 
	\begin{align}
	I_1 &= \sum_{r\, < \,\log_2 M_2}\binom{n}{r}\epsilon^{n-r}(1-\epsilon)^r,\\
	I_2&=\sum_{r\, < \,\log_2 M_2}\binom{n}{r}\epsilon^{n-r}(1-\epsilon)^r\frac{2^{r}}{M_2}.
	\end{align}    
	Note that $I_1 = \text{Pr}\left\{S_n < \log_2 M_2\right\}$. Let $X$ be an arbitrary Poisson random variable with mean $\kappa$, that is $X\sim \text{Poisson}(\kappa)$. Having this, we can write the following 
	\begin{align}
	\label{ff0s}I_1 &\geq \text{Pr}\set{X < \log_2 M_2} - \bigg|  \text{Pr}\set{S_n < \log_2 M_2} -  \text{Pr}\set{X < \log_2 M_2} \bigg|,\\
	\label{fffs} I_2 &= \sum_{r\,<\,\log_2 M_2} \text{Pr}\set{S_n = r} \frac{2^r}{M_2} \\
	\label{ff1s}&\leq \sum_{r\,<\,\log_2 M_2} \text{Pr}\set{X =  r} \frac{2^r}{M_2}  +  \sum_{r\,<\,\log_2 M_2} \bigg| \text{Pr}\set{X =  r} - \text{Pr}\set{S_n = r} \bigg| \frac{2^r}{M_2}. 
	\end{align}
	Note that in both \eqref{ff0s} and \eqref{fffs}, we used a simple fact since that for any two real numbers $a$ and $b$, we have $a \geq b-|a-b|$ and $a\leq b+|a-b|$. Putting \eqref{ff0s} and \eqref{ff1s} together, we obtain
	\begin{align}\label{I1I2s}
	I_1 - I_2  \geq \mathbb{E}\bigg[\mathbf{1}\left(X < \log_{2}M_2\right)\left(1-\frac{2^X}{M_2}\right)\bigg]-J_1-J_2,
	\end{align}
	where
	\begin{align}
	J_1 & =   \bigg|  \text{Pr}\set{S_n < \log_2 M_2} -  \text{Pr}\set{X < \log_2 M_2} \bigg|,\\
	J_2&= \sum_{r\,<\,\log M_2} \bigg| \text{Pr}\set{X =  r} - \text{Pr}\set{S_n = r} \bigg| \frac{2^r}{M_2} . 
	\end{align}
	Using Theorem~\ref{BE for Poisson}, we have
	\begin{align}
	\label{J1s}J_1=& \bigg|  \text{Pr}\set{S_n < \log_2 M_2} -  \text{Pr}\set{X < \log_2 M_2} \bigg| \leq\alpha_1\sqrt{ \text{Pr}\set{X < \log_2 M_2}},\\
	\label{J2sp}&  \bigg|\text{Pr}\set{X =  r} - \text{Pr}\set{S_n = r} \bigg|  \leq \alpha_2 \sqrt{ \text{Pr}\set{X \leq r}},
	\end{align}
	where
	\begin{align}
	\alpha_1 = \frac{\sqrt{2}}{\epsilon^{3/2}}\left(\sqrt{e}-1\right)(1-\epsilon),\quad \alpha_2 = \frac{\sqrt{6}}{\epsilon^2\sqrt{\kappa}}\left(\sqrt{e}-1\right)(1-\epsilon).
	\end{align}
	Before continuing, let us clearly explain how applying Theorem~\ref{BE for Poisson} leads to \eqref{J1s}\footnote{The application of Theorem~\ref{BE for Poisson} is totally similar to the one explained in the achievability proof of BEC but it is repeated here to make each proof self-contained.}. Deriving \eqref{J2sp} from Theorem~\ref{BE for Poisson} is similar. Only for the sake of this explanation and only in the following paragraph, we borrow the local notation of Theorem~\ref{BE for Poisson}, i.e., $\theta,\lambda,\lambda_1,\lambda_2$, and $\psi(\cdot)$. Here, we consider the previously defined $S_n=\sum_{i=1}^{n} X_i$ where $p_i = \operatorname{Pr} (X_i = 1) = 1-\epsilon$ which leads to $\lambda=\lambda_1=\sum_{i=1}^{n} p_i = n(1-\epsilon) = \kappa$ in Theorem~\ref{BE for Poisson}. Hence, our previously defined $X\sim\operatorname{Poisson}(\kappa)$ matches with the $X\sim \operatorname{Poisson}(\lambda)$ in Theorem~\ref{BE for Poisson}. Moreover, $\lambda_2 = n(1-\epsilon)^2$ and thus $\theta = \lambda_2/\lambda_1 = 1-\epsilon$. Hence, \eqref{feq} becomes 
	\begin{align}\label{o1s}
	\bigg|\text{Pr} \{S_n\leq m\}-\text{Pr} \{X\leq m\}\bigg|\leq \frac{\sqrt{2}\,(\sqrt{e}-1)\,(1-\epsilon)}{\epsilon^{3/2}}\sqrt{\psi(m)} = \alpha_1\sqrt{\psi(m)} \leq \alpha_1\sqrt{\operatorname{Pr}(X\leq m)}.
	\end{align}
	The leftmost inequality is due to the fact that, by the definition of $\psi(\cdot)$ in \eqref{defpsi}, we have $\psi(m) \leq \operatorname{Pr}(X\leq m)$. Finally, note that \eqref{o1} holds for every real number $m$ and put $m=\lfloor \log_2 M_2 \rfloor +1/2$. This makes the inequalities strict, i.e., $\operatorname{Pr}(S_n \leq m) = \operatorname{Pr}(S_n < \log_2 M_2) $ and  $\operatorname{Pr}(X \leq m) = \operatorname{Pr}(X < \log_2 M_2) $. Replacing these into \eqref{o1s} gives \eqref{J1s}.
	
	We now upper-bound $J_2$ using \eqref{J2sp}, as follows:
	\begin{align}
	\nonumber J_2&= \sum_{r\,<\,\log M_2}\frac{2^r}{M_2}\bigg| \text{Pr}\set{X =  r} - \text{Pr}\set{S_n = r} \bigg|\\
	\nonumber & \leq \alpha_2 \sum_{r\,<\,\log M_2}\frac{2^r}{M_2}\sqrt{ \text{Pr}\set{X \leq r}}\\
	\label{lasts}&\leq \alpha_2 \sqrt{\sum_{r\,<\,\log M_2}\frac{2^r}{M_2}\text{Pr}\set{X \leq r}}.
	\end{align}
	For obtaining the last part, note that $\sqrt{x}$ is a concave function and
	\begin{equation}
	    \sum_{r\,<\,\log M_2} \frac{2^r}{M_1}\leq1.
	\end{equation}
	Thus, (\ref{lasts}) follows from Jensen inequality for $\sqrt{x}$, that is, $\mathbb{E}\left[\sqrt{Z}\right]\leq \sqrt{\mathbb{E}\left[Z\right]}$.
	Also, consider
	\begin{align}
	\nonumber &\sum_{r\,<\,\log M_2}\frac{2^r}{M_2}\text{Pr}\set{X \leq r}=\sum_{r\,<\,\log M_2}\frac{2^r}{M_2}\left(1-\text{Pr}\set{X > r} \right)\\
	\nonumber&= 1-\frac{1}{M_2}-\sum_{r\,<\,\log M_2}\frac{2^r}{M_2}\sum_{i=r+1}^{\infty}e^{-\kappa} \frac{\kappa^i}{i!}\\
	\nonumber&=1-\frac{1}{M_2} -\sum_{i=0}^{\infty}e^{-\kappa} \frac{\kappa^i}{i!}\left(\mathbf{1}\left(i\geq\log_2 M_2\right)+\mathbf{1}\left(i<\log_2M_2\right)\frac{2^i-1}{M_2}\right)\\
	\nonumber&=1-\frac{1}{M_2}- \left(1-\mathcal{P}_{\kappa}(\log_2 M_2)+\frac{e^{\kappa}}{M_2}\sum_{i\,<\,\log M_2}e^{-2\kappa} \frac{(2\kappa)^i}{i!}-\frac{\mathcal{P}_{\kappa}(\log_2 M_2)}{M_2}\right)\\
	\nonumber&=\mathcal{P}_{\kappa}( \log_{2}M_2)-\frac{e^{\kappa}}{M_2}\,\mathcal{P}_{2\kappa}\left(\log_2M_2\right)-\frac{1-\mathcal{P}_{\kappa}(\log_2 M_2)}{M_2}\\
	\label{tts}&\leq \mathcal{P}_{\kappa}( \log_{2}M_2)-\frac{e^{\kappa}}{M_2}\,\mathcal{P}_{2\kappa}\left(\log_2M_2\right).
	\end{align}
	From \eqref{lasts} and \eqref{tts}, we obtain
	\begin{align}
	\label{J2s}J_2&\leq \alpha_2 \sqrt{\mathcal{P}_{\kappa}( \log_{2}M_2)-\frac{e^{\kappa}}{M_2}\,\mathcal{P}_{2\kappa}\left(\log_2M_2\right)}.
	\end{align}
	Also, considering the notation in \eqref{Poisson cdf}, we can write
	\begin{align}
	\nonumber\mathbb{E}\bigg[\mathbf{1}\left(X < \log_{2}M_2\right)\left(1-\frac{2^X}{M_2}\right)\bigg]  &= \mathbb{E}\bigg[\mathbf{1}\left(X < \log_{2}M_2\right)\bigg] -\mathbb{E}\bigg[\left(X < \log_{2}M_2\right)\frac{2^X}{M_2}\bigg] \\
	\nonumber&= \mathcal{P}_{\kappa}( \log_{2}M_2)-\sum_{r\,<\,\log_{2}M_2}e^{-\kappa} \frac{\kappa^r}{r!}\cdot\frac{2^r}{M_2}\\
	\nonumber&=\mathcal{P}_{\kappa}( \log_{2}M_2)-\frac{e^{\kappa}}{M_2}\sum_{r\,<\,\log_{2}M_2}e^{-2\kappa} \frac{(2\kappa)^r}{r!}\\
	&=\mathcal{P}_{\kappa}( \log_{2}M_2)-\frac{e^{\kappa}}{M_2}\,\mathcal{P}_{2\kappa}\left(\log_2M_2\right). \label{we2}
	\end{align}
	Now, \eqref{I1I2s}, \eqref{J1s}, \eqref{J2s}, and \eqref{we2} together result in
	\begin{align}
	\nonumber I_1 - I_2 &\geq  \mathcal{P}_{\kappa}( \log_{2}M_2)-\frac{e^{\kappa}}{M_2}\,\mathcal{P}_{2\kappa}\left(\log_2M_2\right)-\alpha_1\sqrt{ \mathcal{P}_{\kappa}(\log_2M_2)}\\
 \nonumber&\qquad\qquad\qquad\qquad\qquad-\alpha_2\sqrt{\mathcal{P}_{\kappa}( \log_{2}M_2)-\frac{e^{\kappa}}{M_2}\,\mathcal{P}_{2\kappa}\left(\log_2M_2\right)}\\
	\nonumber    &\geq \mathfrak{P}_2(M_2)-\alpha_2\sqrt{\mathfrak{P}_2(M_2)}-\alpha_1\sqrt{ \mathcal{P}_{\kappa}(\log_2M_2)}\\
	\label{ee}&\geq p_e,
	\end{align}
	where
	\begin{align}
	\mathfrak{P}_2(M_2)&= \mathcal{P}_{\kappa}( \log_{2}M_2)-\frac{e^{\kappa}}{M_2}\,\mathcal{P}_{2\kappa}\left(\log_2M_2\right).
	\end{align}
	Note that \eqref{ee} holds by the definition of $M_2$ in \eqref{M2}. Hence, we showed
	\begin{align}
	\sum_{r\,< \,\log_2 M_2}\binom{n}{r}\epsilon^{n-r}(1-\epsilon)^{r}\bigg(1-\frac{2^{r}}{M_2}\bigg)=I_1-I_2\geq p_e,
	\end{align}
	which means $M_2$ dissatisfies the inequality in \eqref{BEC converse}. Hence, $M^*(n,p_e) \leq M_2$.
\end{proof}

\subsection{How to Compute the Bounds in Theorem~\ref{Bounds for BEC}} \label{ramanujan}
The problem essentially boils down to accurate computation of the probabilities $\rm{Pr}(X = s)$ when $X$ is a Poisson random variable with average $\kappa$. We have 
${\rm{Pr}}(X = s) = e^{-\kappa} \kappa^s/s!$. The value of $s!$ can be approximated using the refined Ramanujan's formula \cite{hirschhorn2014refinement}: 
\begin{equation}
    s!  = \sqrt{\pi} \left( \frac{s}{e} \right)^s \left(8s^3 + 4s^2 + s+ \frac{\theta(s)}{30} \right)^{\frac{1}{6}} ,  
\end{equation}
where 
\begin{equation}
    \theta_1(s):=1 - \frac{11}{8s}+\frac{79}{112s^2} \leq \theta(s) \leq \theta_2(s):= 1 -\frac{11}{8s}+\frac{79}{112s^2} + \frac{20}{33s^3}.
\end{equation}
By plugging-in the lower bound for $s!$ from the above formula  (for $s \leq 10$ we can use the exact value of $s!$) we obtain
\begin{equation}
    {\rm Pr}(X = s) = \frac{e^{s-\kappa + s \ln(\kappa/s)}}{\sqrt{\pi} \left(8s^3 + 4s^2 + s+ \frac{\theta_1(s)}{30}\right)^{\frac{1}{6}}}.
\end{equation}
Note that this formula is exact up to a multiplicative factor of $1+s^{-6}$ which for $s \geq 10$ gives us a $(1+ 10^{-6})$-approximation. Moreover, one can use the lower and upper bound approximations for $s!$ to bound the Poisson probability from above and below and hence obtain bounds for $M_1, M_2$ in Theorem~\ref{Bounds for BEC}.  

Once the Poisson probabilities are approximated (or bounded) suitably, we can compute values of $M_1$ and $M_2$ up to any precision by truncating the expectation and using the bisection method to solve the equations. 

\section{Proofs for BSC}\label{BSC proofs}
In this section, we prove the converse and achievability bounds of Theorem~\ref{bounds for BSC}. In the proofs, we will be using Theorems~\ref{BE}--\ref{BSC conv} as well as Lemmas~\ref{est}--\ref{beta} which are stated at the beginning of this section. For the results in coding theory, we generally refer to \cite{polyphd} as it has well collected and presented the corresponding proofs. See also \cite{Tan}, \cite{Wolfowitz}, \cite{Gal}, \cite{Poor}, and \cite{Elias}.

\begin{theorem}[Berry-Esseen]
	\label{BE}
	For $1\le i \le n$, let $X_i$'s be i.i.d. random variables with $\mu=\mathbb{E}[X_i]$, $\sigma^2 = \mathbb{V}ar[X_i]$, and $\rho=\mathbb{E}\left[\abs{X_i-\mu}^3\right] < \infty$. Then, 
	\begin{equation}
	\abs{\Pr\bigg\{\sum_{i=1}^n \frac{X_i-\mu}{\sigma\sqrt{n}} > x\bigg\}-Q(x)} \le \frac{3\rho}{\sigma^3\sqrt{n}}.
	\end{equation}
\end{theorem}
\begin{proof}
	See \cite[Theorem~3.4.9]{durrett}.
\end{proof}

\begin{theorem}[A Sharp Tail Inequality]
	\label{Tal}
	For $1\le i \le n$, let $X_i$ be independent centered random variables such that $\abs{X_i} \le b$. Define $\sigma^2 = \mathbb{V}ar\bigg[\sum_{i =1}^n X_i\bigg]$ and let $\gamma>0$ be an arbitrary constant. Then, for $0\le t\le \frac{\sigma}{\gamma b}$ we have
	\begin{equation}
	\Pr\bigg\{\sum_{i=1}^n X_i \ge t\bigg\} \le \bigg(\frac{1}{\sqrt{2\pi}}\frac{\sigma}{t}+\gamma\frac{b}{\sigma}\bigg)e^{-nH(\frac{\sigma^2}{n},b,\frac{t}{n})},
	\end{equation}
	where 
	\begin{equation}
	H(\nu,b,x)=\bigg(1+\frac{bx}{\nu}\bigg)\frac{\nu}{b^2+\nu}\ln\bigg(1+\frac{bx}{\nu^2}\bigg)+\bigg(1-\frac{x}{b}\bigg)\frac{b^2}{b^2+\nu}\ln\bigg(1-\frac{x}{b}\bigg).
	\end{equation}
\end{theorem}
\begin{proof}
	See \cite[Theorem~1.1]{Talag}.
\end{proof}

\begin{theorem}[RCU Achievability Bound for BSC]\label{BSC achiv} 
	There exists an $(M,p_e)$-code over BSC$^n$($\delta$) such that
	\begin{equation}
	\label{BSC ach.}
	p_e \leq \sum_{r=0}^n\binom{n}{r}\delta^{r}(1-\delta)^{n-r}\min\bigg\{1,(M-1) S_n^r\bigg\},
	\end{equation}
	where
	\begin{equation}
	    S_n^r=\sum_{s=0}^r\binom{n}{s}2^{-n}.
	\end{equation}
\end{theorem}
\begin{proof}
	See \cite[Corollary~39]{polyphd}.
\end{proof}

\begin{theorem}[Converse Bound for BSC]\label{BSC conv}
	For any $(M,p_e)$-code over the BSC$^n$($\delta$), we have
	\begin{equation}
	\label{converse}
	M \le \frac{1}{\beta_{1-p_e}^n}.
	\end{equation}
	where $\beta_{\alpha}^n$ for a real $\alpha\in[0,1]$ is defined below based on values of $\beta_{\ell}$ where $\ell$ is an integer:
	\begin{align}
	\beta_{\alpha}^n&=(1-\lambda)\beta_L+\lambda\beta_{L+1},\\
	\beta_{\ell}&=\sum_{r=0}^{\ell}\binom{n}{r}2^{-n},
	\end{align}
	such that $\lambda\in [0,1)$ and integer $L$ satisfy the following:
	\begin{align}
	\alpha&=(1-\lambda)\alpha_L+\lambda\alpha_{L+1},\\
	\alpha_{\ell}&=\sum_{r=0}^{\ell-1}\binom{n}{r}\delta^r(1-\delta)^{n-r}.
	\end{align}
\end{theorem}
\begin{proof}
	See \cite[Theorem~40]{polyphd}.
\end{proof}

\begin{lemma}\label{est}
	Consider transmission over BSC($\delta$) in low-capacity regime and let $\kappa=n(1-h_2(\delta))$. Then the following hold:
	\begin{enumerate} [i)]
		\item $\sqrt{n}\left(\frac{1}{2}-\delta\right)  = \sqrt{\frac{\ln 2}{2}\,\kappa}+\mathcal{O}\left(\frac{\kappa\sqrt{\kappa}}{n}\right).$
		\item $\sqrt{n}\,\log_2 \frac{1-\delta}{\delta}  =2\sqrt{\frac{2}{\ln 2}\,\kappa}+\mathcal{O}\left(\frac{\kappa\sqrt{\kappa}}{n}\right).$
	\end{enumerate}
\end{lemma}
\begin{proof}
	\begin{enumerate}[i)]
		\item Taylor expansion of $h_2(\delta)$ around $\frac{1}{2}$ is given by
		\begin{equation}
		h_2(\delta) = 1-\frac{1}{2\ln 2}\sum_{n=1}^\infty \frac{(1-2\delta)^{2n}}{n(2n-1)}.
		\end{equation}
		Thus, the estimation of $h_2(\delta)$ up to the third order will be the following:
		\begin{equation}
		    h_2(\delta)=1-\frac{1}{2\ln 2}(1-2\delta)^2-\frac{1}{2\ln 2}\cdot\frac16(1-2\delta)^4+\mathcal{O}\left((1-2\delta)^6\right).
		\end{equation}
		Therefore, 
		\begin{equation}
		    C=1-h_2(\delta) = \frac{1}{2\ln 2}(1-2\delta)^2+\frac{1}{2\ln 2}\cdot\frac16(1-2\delta)^4+\mathcal{O}\left((1-2\delta)^6\right).
		\end{equation}
		Now assuming $x=(1-2\delta)^2$ leads to the following equation:
		\begin{equation}
		\frac{x^2}{6}+x-2C\ln 2=0.
		\end{equation}
		Solving this equation gives
		\begin{equation}
		x=3\left(-1+\sqrt{1+\frac43C\ln 2}\right)= 2C\ln 2+\mathcal{O}\left(C^2\right).
		\end{equation}
		Therefore,
		\begin{equation}\label{eq1}
		\frac 12-\delta=\frac{\sqrt{x}}{2}= \sqrt{\frac{\ln 2}{2}C}+\mathcal{O}\left(C^{\frac 32}\right).
		\end{equation}
		We want to obtain $\sqrt{n}(1/2-\delta)$. Therefore, we multiply both sides of \eqref{eq1} by $\sqrt{n}$ to obtain
		\begin{equation}
		\sqrt{n}\left(\frac{1}{2}-\delta\right) = \sqrt{\frac{\ln 2}{2}nC}+\mathcal{O}\left(nC^{\frac 32}\right).
		\end{equation}
		Note that $\kappa=nC$. As a result,
		\begin{equation}\label{q}
		\sqrt{n}\left(\frac{1}{2}-\delta\right)= \sqrt{\frac{\ln 2}{2}\kappa}+\mathcal{O}\left(\frac{\kappa\sqrt{\kappa}}{n}\right).
		\end{equation}
		\item The first order estimation of the function $\log_2 ({1-\delta})/{\delta}$ around $1/2$ is
		\begin{equation}
		    \sqrt{n}\,\log_2 \frac{1-\delta}{\delta} = \frac{4}{\ln 2} \left(\frac 12 - \delta\right)+\mathcal{O}\left(\left(\frac 12-\delta\right)^2\right).
		\end{equation}
		Hence, using \eqref{q}, we arrive at
		\begin{equation}
		    \sqrt{n}\,\log_2 \frac{1-\delta}{\delta} = 2\sqrt{\frac{2}{\ln 2}\kappa}+\mathcal{O}\left(\frac{\kappa\sqrt{\kappa}}{n}\right).
		\end{equation}
	\end{enumerate}
\end{proof}

\begin{lemma}
	\label{Q}
	Suppose $0<p<Q(1)\approx0.16$.\footnote{Note that similar results can be obtained when $0<p<1$ but taking the assumption $0<p<Q(1)\approx0.16$ into account leads to a very simple format for bounds and will not cause any conflict because we are interested in small bit and block error probabilities. More specifically, typical values of $p$ would be around $10^{-3}-10^{-6}$.} Then the following hold:
	\begin{enumerate}[i)]
		\item[(i)] $\sqrt{2\pi}\,p\, Q^{-1}(p)<e^{-\frac{Q^{-1}(p)^2}{2}}<\sqrt{2\pi}\,p\, \left(Q^{-1}(p)+1\right).$
		\item[(ii)] $\sqrt{8\pi+2-2\ln p}-2\sqrt{2\pi}<Q^{-1}(p) <\sqrt{-\ln2\pi-2\ln p}.$
	\end{enumerate}
\end{lemma}
\begin{proof}
	\begin{enumerate}[i)]
		\item For $x>0$, it is well known that 
		\begin{equation}
		\label{jj}
		\frac{\phi(x)}{x+\frac{1}{x}}=\left(\frac{x}{1+x^2}\right)\phi(x)<Q(x)<\frac{\phi(x)}{x},
		\end{equation}
		where $\phi(x)$ is the probability distribution function of the Gaussian distribution, i.e., $\phi(x):= e^{-x^2/2}/\sqrt{2\pi}$.  Note that for $x>1$, \eqref{jj} becomes
		\begin{equation}\label{Qj}
		\frac{\phi(x)}{x+1}<Q(x)<\frac{\phi(x)}{x}.
		\end{equation}
		Now, define $p=Q(x)$. Thus, $x=Q^{-1}(p)>1$. Therefore,
		\begin{equation}
		\frac{e^{-\frac{Q^{-1}(p)^2}{2}}}{\sqrt{2\pi}\,\left(Q^{-1}(p)+1\right)}<p<\frac{e^{-\frac{Q^{-1}(p)^2}{2}}}{\sqrt{2\pi}\,Q^{-1}(p)}.
		\end{equation}
		Simplify to get
		\begin{equation}
		\sqrt{2\pi}\,p\, Q^{-1}(p)<e^{-\frac{Q^{-1}(p)^2}{2}}<\sqrt{2\pi}\,p\, \left(Q^{-1}(p)+1\right).
		\end{equation}
		\item From \eqref{Qj}, for $x>1$ we have 
		\begin{equation}
		    Q(x)<\frac{   e^{-\frac{x^2}{2}}   }{   \sqrt{2\pi} }.
		\end{equation}
		Therefore,
		\begin{equation}
		    x<\sqrt{    -2\ln  \left( \sqrt{2\pi}Q(x)   \right) }.
		\end{equation}
		Put $x=Q^{-1}(p)$ to get
		\begin{equation}
		    Q^{-1}(p) < \sqrt{2\ln\frac{1}{\sqrt{2\pi}\,p}}=\sqrt{-\ln2\pi-2\ln p}.
		\end{equation}
		For the other side, note that from \eqref{Qj}, for $x>1$, we also have
		\begin{equation}
		    Q(x)>\frac{\phi(x)}{2x}=\frac{e^{-\frac{x^2}{2}}}{2x\sqrt{2\pi}}.
		\end{equation}
		Assume $p = Q(x)$ and take the $\ln (\cdot)$ from both sides to get 
		\begin{equation}
		  \frac{x^2}{2}+\ln (2\sqrt{2\pi}\,x)+\ln p >0.
		\end{equation}
		Now using the inequality $\ln (2\sqrt{2\pi}\,x)\le 2\sqrt{2\pi}\,x-1$ leads to 
		\begin{equation}
		    \frac{x^2}{2}+2\sqrt{2\pi}\,x+\ln p-1 >0.
		\end{equation}
		Hence, 
		\begin{equation}
		    Q^{-1}(p)=x>-2\sqrt{2\pi}+\sqrt{8\pi+2-2\ln p}.
		\end{equation}
	\end{enumerate}
\end{proof}

\begin{lemma}
	\label{comb}
	Suppose $0<\delta<1$ such that $n\delta$ is an integer. Then, for all $0\leq r\leq n$, we have
	\begin{equation}
	    \binom{n}{r}\delta^{r}(1-\delta)^{n-r}  \leq \frac{\theta}{\sqrt{n}}.
	\end{equation}
	where 
	\begin{equation}
	    \theta = \frac{e}{2\pi \sqrt{\delta(1-\delta)}}.
	\end{equation}
\end{lemma}
\begin{proof}
	For $0\leq r\leq n$ define
	\begin{equation}
	    A(r)=\binom{n}{r}\delta^{r}(1-\delta)^{n-r},
	\end{equation}
	and 
	\begin{equation}
	    r^*=\arg \max_{0 \le r \le n} A(r).
	\end{equation}
	First of all, note that from the Mode of Binomial distribution, we know that $r^*= \lfloor(n+1)\delta\rfloor=n\delta$. Also from Stirling formula, for any integer $n$, we have
	\begin{equation}
	    \sqrt{2\pi n}\left(\frac{n}{e}\right)^n\leq n! \leq e\sqrt{ n}\left(\frac{n}{e}\right)^n.
	\end{equation}
	Therefore,
	\begin{align}
	\nonumber A(r^*) &=\binom{n}{n\delta}\delta^{n\delta}(1-\delta)^{n-n\delta}= \frac{   n!\,\delta^{n\delta}(1-\delta)^{n-n\delta}    }{   n\delta!\,(n-n\delta)!   }\\
	\nonumber &\leq \frac {e\sqrt{ n}\left(\frac{n}{e}\right)^n  \delta^{n\delta}(1-\delta)^{n-n\delta}  }{   \sqrt{2\pi n\delta}\left(\frac{n\delta}{e}\right)^{n\delta}     \sqrt{2\pi n(1-\delta)}\left(\frac{n(1-\delta)}{e}\right)^{n(1-\delta)}   }\\
	&=\frac{e}{2\pi \sqrt{\delta(1-\delta)}}\cdot\frac{1}{\sqrt{n}}.
	\end{align}	
\end{proof}

\begin{definition}
	Consider the following binary hypothesis test:
	\begin{align}
	&H_0: X \sim P,\\
	&H_1: X \sim Q,
	\end{align}
	where $P$ and $Q$ are two probability distributions on the same space $\mathcal{X}$. Suppose  a continuous decision rule $\zeta$ as a mapping from observation space $\mathcal{X}$ to $[0,1]$ with $\zeta(x) \approx 0$ corresponding to reject $H_1$ and $\zeta(x) \approx 1$ corresponding to reject $H_0$.
	Now, define the smallest type-II error in this binary hypothesis test given that type-I error is not greater than $p_e$ as the following:
	\begin{equation}\label{beta def}
	\beta_{1-p_e}(P,Q)=\inf_{\zeta: \mathcal{X}\rightarrow[0,1]}\bigg\{\mathbb{E}_{X \sim Q}[1-\zeta(X)]:\mathbb{E}_{X\sim P}[\zeta(X)]\le p_e\bigg\}.
	\end{equation}
\end{definition}

\begin{lemma}
	\label{beta}
	Consider two discrete probability distributions $P$ and $Q$ on a measure space $\mathcal{X}$. Define the product distributions $P^n$ and $Q^n$ as
	\begin{equation}
	    P^n(\mathbf{x})=\prod_{i=1}^{n}P(x_i),\qquad Q^n(\mathbf{x})=\prod_{i=1}^{n}Q(x_i),
	\end{equation}
	where $\mathbf{x}=(x_1,\dots,x_n)\in \mathcal{X}^n$. Then, for $p_e \in (0,1)$ and any $\gamma$, we have
	\begin{equation}
	\log_2\beta_{1-p_e}(P^n,Q^n)\ge -nD+\sqrt{nV}Q^{-1}\left(p_e+\gamma+\frac{B}{\sqrt{n}}\right)+\log_2\gamma,
	\end{equation}
	where 
	\begin{align}
	D&=D(P\|Q),\\
	V&=\int_{\mathcal{X}}\left(\log_2\frac{dP(x)}{dQ(x)}-D\right)^2dP(x),\\
	T&=\int_{\mathcal{X}}\abs{\log_2\frac{dP(x)}{dQ(x)}-D}^3dP(x),\\
	B&=\frac{T}{V^{\frac{3}{2}}}.
	\end{align}
\end{lemma}
\begin{proof}
	In the proof of \cite[Lemma~14]{polyphd}, put $\alpha=1-p_e, \Delta=\gamma\sqrt{n},P_i=P, Q_i=Q$ and consider the logarithm in base $2$.
\end{proof}

\subsection{Proof of Theorem~\ref{bounds for BSC}} \label{BSC proofs1}
\begin{proof}[{\bf Achievability Bound of Theorem~\ref{bounds for BSC}}]
	Considering $S_n^r:=\sum_{s=0}^r\binom{n}{s}2^{-n}$, define $T$  and $S$ as follows:
	\begin{align}
	T &:= n\delta+\sqrt{n\delta(1-\delta)}Q^{-1}\bigg(p_e-\frac{p_e}{\sqrt{\kappa}}-\frac{c_1}{\sqrt{n}}\bigg) ,    \label{a0}\\
	S&:=  \sum_{r\,\le\, T}\binom{n}{r}\delta^r(1-\delta)^{n-r}S_n^r \label{S},
	\end{align}
	where $c_1=\frac{3(2\delta^2-2\delta+1)}{\sqrt{\delta(1-\delta)}}$ is a constant. Suppose that we choose some $M$ that satisfies
	\begin{equation}
	\label{a1}
	M \le \frac{p_e}{\sqrt{\kappa}\,S}.
	\end{equation}
	Now define
	\begin{align}
	I_1 &= \sum_{r\,\le\, T}\binom{n}{r}\delta^r(1-\delta)^{n-r}MS_n^r,\\
	I_2&=\sum_{T<\,r\leq\, n}\binom{n}{r}\delta^r(1-\delta)^{n-r}.
	\end{align}
	From \eqref{S} and \eqref{a1}, we have
	\begin{equation}
	\label{I1}
	I_1=MS\leq \frac{p_e}{\sqrt{\kappa}}.
	\end{equation}
	For $1\leq i\leq n$, define $V_i$ to be a Bernoulli random variable with $\text{Pr}\{V_i=1\}=1-\text{Pr}\{V_i=0\}=\delta$. Then define 
	\begin{equation}
	\mu=\mathbb{E}\left[V_i\right]=\delta,\quad \sigma^2 = \mathbb{V}ar[V_i]=\delta(1-\delta),\quad \rho=\mathbb{E}\left[\abs{V_i-\mu}^3\right]=\delta(1-\delta)\left(2\delta^2-2\delta+1\right).
	\end{equation}
	Now, using Theorem~\ref{BE}, we can write
	\begin{equation}
	\label{a3}
	I_2=\text{Pr}\left\{\sum_{i=1}^nV_i>T\right\} \le Q\bigg(       \frac{T-n\mu}{   \sigma\sqrt{n} }   \bigg)+\frac{3\rho}{\sigma^3\,\sqrt{n}} = Q\bigg(\frac{T-n\delta}{\sqrt{n\delta(1-\delta)}}\bigg)+\frac{c_1}{\sqrt{n}} \overset{(\ref{a0})}{=} p_e-\frac{p_e}{\sqrt{\kappa}}.
	\end{equation} 
	Therefore, \eqref{I1} and \eqref{a3} together give
	\begin{align}
	I_1+I_2 \leq p_e-\frac{p_e}{\sqrt{\kappa}}+\frac{p_e}{\sqrt{\kappa}}=p_e.
	\end{align}
	Now, we can write
	\begin{align}
	\sum_{r=0}^n\binom{n}{r}\delta^{r}(1-\delta)^{n-r}\min\bigg\{1,(M-1) S_n^r\bigg\}\
	\nonumber&\le \sum_{r=0}^n\binom{n}{r}\delta^{r}(1-\delta)^{n-r}\min\bigg\{1,M S_n^r\bigg\}\\
	&\le I_1+I_2
	\le p_e. 
	\end{align}
	Thus, $M$ dissatisfies the inequality in \eqref{BSC ach.} and as a result, $M^*(n,p_e)\ge M$. Note that $M$ was arbitrarily chosen to satisfy \eqref{a1}, This means for any $M$ satisfying \eqref{a1}, we have $M^*(n,p_e)\ge M$. Hence,
	\begin{equation}
	\label{lb}
	M^*(n,p_e)\ge\sup\set{M: M \le \frac{p_e}{\sqrt{\kappa}\,S}}=\frac{p_e}{\sqrt{\kappa}\,S}.
	\end{equation}
	Now, in order to find a lower bound for $M^*(n,p_e)$, it suffices to find an upper bound for $S$. This is our main goal for the rest of the proof. Note that due to Lemma \ref{comb}, there exists a constant $\theta$ such that 
	\begin{equation}\label{theta}
	\binom{n}{r}\delta^r(1-\delta)^{n-r}\leq \frac{\theta}{\sqrt{n}}.
	\end{equation}
	Moreover, note that $S_n^r$ is increasing with respect to $r$. Define 
	\begin{equation}
	\beta=\frac 14 \sqrt{\frac{\ln 2}{2}}\cdot \frac{\log_2 \kappa}{\sqrt{\kappa}}.
	\end{equation}
	With this choice of $\beta$, we continue as follows:
	\begin{align}
	\nonumber S &= \sum_{r\,\leq \,T - \beta \sqrt{n} } \binom{n}{r}\delta^{r}(1-\delta)^{n-r}S_{n}^r +\sum_{T - \beta \sqrt{n} \,<\, r \,\leq \,T} \binom{n}{r}\delta^{r}(1-\delta)^{n-r}S_{n}^r   \\
	\nonumber &\leq S_{n}^{T - \beta \sqrt{n} }\sum_{r\,\leq \,T - \beta \sqrt{n} } \binom{n}{r}\delta^{r}(1-\delta)^{n-r} +S_{n}^T\sum_{T - \beta \sqrt{n} \,<\, r \,\leq \,T} \binom{n}{r}\delta^{r}(1-\delta)^{n-r}\\
	&\overset{(\ref{theta})}{\leq} S_{n}^{T - \beta \sqrt{n} }+S_{n}^T\beta\theta. \label{a2}
	\end{align} 
	In order to find an upper bound for $S$, it then suffices to find an upper bound for the right-hand side in \eqref{a2}. For $1 \le i \le n$, let $Y_i \sim$ Bernoulli($\frac{1}{2}$) and $X_i = \frac{1}{2}-Y_i$. Note that $\mathbb{E}[X_i]=0$, $\mathbb{V}ar[X_i]=\frac{1}{4}$ and $\abs{X_i} \le \frac{1}{2}$. Then Theorem~\ref{Tal} implies
	\begin{align}
	S_n^r &= \Pr\left\{\sum_{i=1}^nY_i \le r\right\}=\Pr\left\{\sum_{i=1}^nX_i \ge \frac{n}{2}-r\right\}\le \bigg(\frac{1}{2\sqrt{2\pi}}\cdot\frac{\sqrt{n}}{\frac{n}{2}-r}+\frac{\gamma}{\sqrt{n}}\bigg)e^{-nH(\frac{1}{4},\frac{1}{2},\frac{1}{2}-\frac{r}{n})}.
	\end{align}
	A simple calculation shows that $H(\frac{1}{4},\frac{1}{2},\frac{1}{2}-\frac{r}{n})=\left(1-h_2\left(\frac{r}{n}\right)\right) \ln2$. Thus,
	\begin{equation}
	\label{snt}
	S_n^r \le \bigg(\frac{1}{2\sqrt{2\pi}}\cdot\frac{1}{\frac{\sqrt{n}}{2}-\frac{r}{\sqrt{n}}}+\frac{\gamma}{\sqrt{n}}\bigg)2^{-n\left(1-h_2\left(\frac{r}{n}\right)\right)}.
	\end{equation}
	Due to part (i) of Lemma~\ref{est} and part (ii) of Lemma~\ref{Q}, the first term in the right-hand side of \eqref{snt}, when $r=T$, can be written as
	\begin{align}
	\nonumber  \left(   \frac{\sqrt{n}}{2}-\frac{T}{\sqrt{n}}    \right)^{-1}  &=  \left(   \sqrt{n}\left(\frac{1}{2}-\delta\right)-\sqrt{\delta(1-\delta)}Q^{-1}\left(p_e-\frac{p_e}{\sqrt{\kappa}}-\frac{c_1}{\sqrt{n}}\right)  \right)^{-1}\\ 
	\nonumber  &=  \left(   \sqrt{\frac{\ln 2}{2}\,\kappa}+\mathcal{O}\left(\frac{\kappa\sqrt{\kappa}}{n}\right)-\sqrt{\delta(1-\delta)}Q^{-1}\left(p_e-\frac{p_e}{\sqrt{\kappa}}-\frac{c_1}{\sqrt{n}}\right)   \right)^{-1} \\ 
	\nonumber &=\left(  \sqrt{\frac{\ln 2}{2}\,\kappa}\right)^{-1}+\mathcal{O}\left(\frac{\sqrt{\kappa}}{n}\right)+\mathcal{O}\left(\frac{\sqrt{-\log p_e}}{\kappa}\right)  \\
	\label{a4} &= \left(  \sqrt{\frac{\ln 2}{2}\,\kappa}\right)^{-1}+\mathcal{O}\left(\frac{\sqrt{-\log p_e}}{\kappa}\right). 
	\end{align}
	Note that to use Lemma~\ref{est}, we need to have $\kappa\sqrt{\kappa}=o(n)$. This is given due to the criterion of the low-capacity regime for Theorem~\ref{bounds for BSC} which states that $\kappa=f(n)=(o(n))^{2/3}$. This condition is also used for further simplification of order terms. Similarly, when $r=T-\beta\,\sqrt{n}$, we have
	\begin{align}
	\label{a5}
	\left(   \frac{\sqrt{n}}{2}-\frac{T-\beta\sqrt{n}}{\sqrt{n}}    \right)^{-1}  = \left(   \frac{\sqrt{n}}{2}-\frac{T}{\sqrt{n}} +\beta    \right)^{-1} = \left(  \sqrt{\frac{\ln 2}{2}\,\kappa}\right)^{-1}+\mathcal{O}\left(\frac{\sqrt{-\log p_e}}{\kappa}\right).  
	\end{align}
	Now, the goal is to estimate the term $n\left(1-h_2(\frac rn)\right)$ in the right-hand side of \eqref{snt} for $r=T$ and $r=T-\beta\sqrt{n}$. Using the third order estimation of $h_2(\cdot)$ gives
	\begin{align} 
	\nonumber h_2\left(\frac{T}{n}\right)= h_2(\delta)+\left(\log_2\frac{1-\delta}{\delta}\right)\sqrt{\frac{\delta(1-\delta)}{n}}Q^{-1}&\left(p_e-\frac{p_e}{\sqrt{\kappa}}-\frac{c_1}{\sqrt{n}}\right)\\
 &-\frac{1}{n\ln 2}Q^{-1}\left(p_e\right)^2+\mathcal{O}\left(\frac{1}{n\sqrt{\kappa}}\right).
	\end{align}
	Therefore, by the definition of $\kappa$ and part (ii) of Lemma~\ref{est} we obtain
	\begin{align}\label{ey3}
	\nonumber n\left(1-h_2\left(\frac{T}{n}\right)\right)= \kappa-2\sqrt{\frac{2\delta(1-\delta)}{\ln 2}}\cdot \sqrt{\kappa}\,Q^{-1}&\left(p_e-\frac{p_e}{\sqrt{\kappa}}-\frac{c_1}{\sqrt{n}}\right)\\
 &+\frac{1}{\ln 2}Q^{-1}\left(p_e\right)^2+\mathcal{O}\left(\frac{1}{\sqrt{\kappa}}\right).
	\end{align}
	Note that 
	\begin{equation}\label{ey}
	    \frac{d}{dx}Q^{-1}(x)=-\sqrt{2\pi}\,e^{\frac{Q^{-1}(x)^2}{2}}.
	\end{equation}
	Thus, by replacing the Taylor expansion of $Q^{-1}(\cdot)$ in \eqref{ey3} and using the parts (i) and (ii) of Lemma \ref{Q} together with \eqref{ey}, we arrive at 
	\begin{align}\label{h}
	n\left(1-h_2\left(\frac{T}{n}\right)\right)= \kappa-2\sqrt{\frac{2\delta(1-\delta)}{\ln 2}}\cdot \sqrt{\kappa}\,Q^{-1}\left(p_e\right)+\mathcal{E}_1,
	\end{align}
	where
	\begin{equation}\label{error term}
	\mathcal{E}_1=\frac{1}{\ln 2}Q^{-1}\left(p_e\right)^2+\mathcal{O}\left(\frac{1}{\sqrt{-\log p_e}}\right)+\mathcal{O}\left(\frac{1}{\sqrt{\kappa}}\right)=\frac{1}{\ln 2}Q^{-1}\left(p_e\right)^2+\mathcal{O}\left(\frac{1}{\sqrt{-\log p_e}}\right).
	\end{equation}
	Exploiting the same analogy leads to the following result for $r=T-\beta\sqrt{n}$ :
	\begin{align}\label{h2}
	\nonumber n\left(1-h_2\left(\frac{T-\beta\sqrt{n}}{n}\right)\right)&= \kappa-2\sqrt{\frac{2\delta(1-\delta)}{\ln 2}}\cdot \sqrt{\kappa}\,Q^{-1}\left(p_e\right)+2\sqrt{\frac{2}{\ln 2}}\cdot\beta \sqrt{\kappa}+\mathcal{E}_1\\
	&= \kappa-2\sqrt{\frac{2\delta(1-\delta)}{\ln 2}}\cdot \sqrt{\kappa}\,Q^{-1}\left(p_e\right)+\frac 12\log_2 \kappa+\mathcal{E}_1.
	\end{align}
	Now, \eqref{snt}, \eqref{a4}, and \eqref{h} together imply
	\begin{align}\label{f1}
	S_n^T \leq \frac{1}{2\sqrt{\pi\kappa\ln 2}}2^{-\left(\kappa-2\sqrt{\frac{2\delta(1-\delta)}{\ln 2}}\cdot \sqrt{\kappa}\,Q^{-1}\left(p_e\right)+\mathcal{E}_1\right)}+\mathcal{O}\left(\frac{\sqrt{-\log p_e}}{\kappa\, 2^{\kappa}}\right).  
	\end{align}
	Similarly, \eqref{snt}, \eqref{a5}, and \eqref{h2} together imply 
	\begin{align}\label{f2}
	S_n^{  T-\beta\sqrt{n} } \leq \frac{1}{2\sqrt{\pi\kappa\ln 2}}2^{-\left(\kappa-2\sqrt{\frac{2\delta(1-\delta)}{\ln 2}}\cdot \sqrt{\kappa}\,Q^{-1}\left(p_e\right)+\frac 12\log_2 \kappa+\mathcal{E}_1\right)}+\mathcal{O}\left(\frac{\sqrt{-\log p_e}}{\kappa\, 2^{\kappa}}\right).  
	\end{align}
	As a result, from \eqref{f1}, \eqref{f2}, and \eqref{a2}, we have
	\begin{align}
	\nonumber S&\leq S_{n}^{T - \beta \sqrt{n} }+S_{n}^T\beta\theta \\
	&\leq \frac{1}{2\sqrt{\pi\ln 2}}\cdot \frac{1}{\kappa}\,\left(1+\theta\log_2\kappa\right)2^{-\left(\kappa-2\sqrt{\frac{2\delta(1-\delta)}{\ln 2}}\cdot \sqrt{\kappa}\,Q^{-1}\left(p_e\right)+\mathcal{E}_1\right)}+\mathcal{O}\left(\frac{\sqrt{-\log p_e}}{\kappa\, 2^{\kappa}}\right).
	\end{align}
	Now, taking the logarithm of both sides gives
	\begin{align}
	\nonumber -\log_2 S \geq \kappa-2\sqrt{\frac{2\kappa\delta(1-\delta)}{\ln 2}}Q^{-1}\left(p_e\right)&+\mathcal{E}_1+\log_2 \kappa\\
 &-\log_2 \frac{1+\theta\log_2\kappa}{2\sqrt{\pi\ln 2}}+\mathcal{O}\left(\frac{\sqrt{-\log p_e}}{\log \kappa}\right).
	\end{align}
	Therefore, by replacing $\mathcal{E}_1$ from \eqref{error term}, we find 
	\begin{align}
	\nonumber -\log_2 S \geq \kappa &-2\sqrt{\frac{2\delta(1-\delta)}{\ln 2}}\cdot \sqrt{\kappa}\,Q^{-1}\left(p_e\right)+\log_2 \kappa+\frac{1}{\ln 2}Q^{-1}\left(p_e\right)^2\\
	&+\mathcal{O}\left(\log\log\kappa\right)+\mathcal{O}\left(\frac{1}{\sqrt{-\log p_e}}\right)+\mathcal{O}\left(\frac{\sqrt{-\log p_e}}{\log \kappa}\right).
	\end{align}
	Further simplifications and taking into account the dominant terms then result in
	\begin{align}\label{s}
	-\log_2 S \geq \kappa-2\sqrt{\frac{2\delta(1-\delta)}{\ln 2}}\cdot \sqrt{\kappa}\,Q^{-1}\left(p_e\right)+\log_2 \kappa+\frac{1}{\ln 2}Q^{-1}\left(p_e\right)^2+\mathcal{O}\left(\log\log\kappa\right).
	\end{align}
	Finally, from \eqref{lb} and \eqref{s}, we conclude that
	\begin{align}
	\nonumber &\log_2M^*(n,p_e)\ge\log_2p_e-\frac 12\log_2\kappa -\log_2S\\
	& \geq \kappa-2\sqrt{\frac{2\kappa\delta(1-\delta)}{\ln 2}}\,Q^{-1}\left(p_e\right)+\frac 12\log_2 \kappa+\log_2p_e+\frac{1}{\ln 2}Q^{-1}\left(p_e\right)^2+\mathcal{O}\left(\log\log\kappa\right).
	\end{align}
	Note that by part (ii) of Lemma~\ref{Q}, we have
	\begin{equation}
	    \frac{1}{\ln 2}Q^{-1}\left(p_e\right)^2=-2\log_2p_e+\mathcal{O}(1).
	\end{equation}
	Hence,
	\begin{align}
	\log_2M^*(n,p_e)
	& \geq \kappa-2\sqrt{\frac{2\kappa\delta(1-\delta)}{\ln 2}}\,Q^{-1}\left(p_e\right)+\frac 12\log_2 \kappa-\log_2p_e+\mathcal{O}\left(\log\log\kappa\right).
	\end{align}
\end{proof}    

\begin{proof}[{\bf Converse Bound of Theorem~\ref{bounds for BSC}}]
	Let $X \sim $ Bernoulli(${1}/{2}$) and $Y$ be the input and output of the BSC($\delta$), respectively. Also suppose $P_X$, $P_Y$, and $P_{XY}$ are distributions of $X$, $Y$, and the joint distribution of $(X,Y)$, respectively. Define $P=P_{XY}$ and $Q=P_XP_Y$, and then define $P^n$ and $Q^n$ in terms of $P$ and $Q$ as they are in Lemma~\ref{beta}. Also consider $\beta_{1-p_e}(P^n,Q^n)$ as in \eqref{beta def}. Under these choices of $P$ and $Q$, it can be verified that $\beta_{1-p_e}^n$ defined in \eqref{converse}, is a piecewise linear approximation of $\beta_{1-p_e}(P^n,Q^n)$ based on discrete values of error probabilities. Therefore, from \eqref{converse}, we can write
	\begin{align}
	\log_2M^*(n,p_e) \le -\log_2\beta_{1-p_e}(P^n,Q^n).
	\end{align}
	Now, using Lemma~\ref{beta} for any $\gamma>0$, we have
	\begin{align}
	\label{asbound}
	\log_2M^*(n,p_e) \le nD-\sqrt{nV}Q^{-1}\left(p_e+\gamma+\frac{B}{\sqrt{n}}\right)-\log_2\gamma.
	\end{align}
	Under the specific values of $P$ and $Q$ given above, the quantities $D$,$V$,$T$, and $B$ in Lemma~\ref{beta} can be computed as follows:
	\begin{align}
	D&=1-h_2(\delta),\\
	V&=\delta(1-\delta)\left(\log_2\frac{1-\delta}{\delta}\right)^2,\\
	T&=\delta(1-\delta)\left(\log_2\frac{1-\delta}{\delta}\right)^3\left(2\delta^2-2\delta+1\right),\\
	B&=\frac{2\delta^2-2\delta+1}{\sqrt{\delta(1-\delta)}}.
	\end{align}
	Replacing these quantities and using Lemma~\ref{est}, we can rewrite \eqref{asbound} as the following:
	\begin{align}\label{eyy}
	\log_2M^*(n,p_e) \le \kappa-2\sqrt{\frac{2\kappa\delta(1-\delta)}{\ln 2}}Q^{-1}\left(p_e+\gamma+\frac{B}{\sqrt{n}}\right)-\log_2\gamma+\mathcal{O}\left(\frac{\kappa\sqrt{-\kappa\log_2p_e}}{n}\right).
	\end{align}
	Note that to use Lemma~\ref{est}, we need to have $\kappa\sqrt{\kappa}=o(n)$. This is given due to the criterion for being in the low-capacity regime for Theorem~\ref{bounds for BSC} which states that $\kappa=f(n)=(o(n))^{2/3}$. This condition is also used for further simplification of order terms. 
	
	At this step, we choose $\gamma={p_e}/{\sqrt{\kappa}}$ and replace the Taylor expansion of $Q^{-1}(\cdot)$ in \eqref{eyy}. By using parts (i) and (ii) of Lemma~\ref{Q} together with \eqref{ey} in that Taylor expansion, we can conclude that 
	\begin{align}
	\nonumber \log_2M^*(n,p_e) \le \kappa &-2\sqrt{\frac{2\kappa\delta(1-\delta)}{\ln 2}}Q^{-1}\left(p_e\right)-\log_2\frac{p_e}{\sqrt{\kappa}}\\
 &+\mathcal{O}\left(\frac{\kappa\sqrt{-\kappa\log_2p_e}}{n}\right)+\mathcal{O}\left(\frac{1}{\sqrt{-\log p_e}}\right).
	\end{align}
	Further simplifications and considering the dominant terms then give
	\begin{align}
	\log_2M^*(n,p_e) &\le \kappa-2\sqrt{\frac{2\kappa\delta(1-\delta)}{\ln 2}}Q^{-1}\left(p_e\right)-\log_2\frac{p_e}{\sqrt{\kappa}}+\mathcal{O}\left(\frac{1}{\sqrt{-\log p_e}}\right).
	\end{align}
\end{proof}

\subsection{Proof of Corollary~\ref{optimal n for BSC}}\label{BSC proofs2}
\begin{proof}
	In order to obtain the optimal blocklength $n^*$ for transmission of $k$ information bits over a low-capacity BSC($\delta$), it suffices to replace $M^*(n^*,p_e)=k$. Then $n^*$ can be computed by solving (\ref{BSC b}). Replace $M^*(n^*,p_e)=k$ and $\kappa=n^*C$ in \eqref{BSC b}, where $C=1-h_2(\delta)$ to obtain
	\begin{align}
	k = n^*C-2\sqrt{\frac{2\delta(1-\delta)}{\ln 2}}\cdot\sqrt{n^*C}\,Q^{-1}\left(p_e\right)-\log_2 p_e+\mathcal{O}(\log \kappa).
	\end{align}
	Define $x=\sqrt{n^*C}$, $a=\sqrt{\frac{2\delta(1-\delta)}{\ln 2}}Q^{-1}\left(p_e\right)$ and $b=k+\log_2 p_e+\mathcal{O}(\log \kappa)$. Thus, $x^2-2ax-b=0.$
	Note that the answer will be $x= a+\sqrt{a^2+b}$. Therefore, 
	$\sqrt{n^*C}= a+\sqrt{a^2+b}$. 
	More simplifications are as follows:
	\begin{align}
	\nonumber n^*&= \frac1C \left(a+\sqrt{a^2+b}\right)^2\\
	\nonumber &=\frac 1C \left(      2a^2+b+   2a   \left(     \sqrt{k}+    \mathcal{O}\left(    \frac{-\log p_e}{k}  \right)    +    \mathcal{O}\left(   \frac{\log \kappa}{k}  \right)   \right)\right)\\
	\label{nn}&=\frac 1C \left(k+2a\sqrt{k}+2a^2+\log_2 p_e+\mathcal{O}(\log \kappa)\right).
	\end{align}
	Note that in the low-capacity regime, under optimal blocklength, $\log\kappa \approx \log k$. Therefore, by substituting the values of $C$ and $a$ in \eqref{nn}, we obtain
	\begin{equation}
	    n^* =\frac 1{1-h_2(\delta)} \left(k+2\sqrt{\frac{2\kappa\delta(1-\delta)}{\ln 2}}Q^{-1}( p_e )+\frac{4\delta(1-\delta)}{\ln 2}Q^{-1}(p_e)^2+\log_2 p_e+\mathcal{O}(\log k)\right).
	\end{equation}
\end{proof}

\section{Proofs for AWGN Channel}
In this section, we prove the converse and achievability bounds of Theorem~\ref{bounds for AWGN}. In the proofs, we will be using Theorems~\ref{AWGN conv.}--\ref{AWGN ach.} as well as Lemma~\ref{eta} which are stated at the beginning of this section.
For the results in coding theory, we generally refer to \cite{polyphd} as it has well collected and presented the corresponding proofs. See also \cite{Tan}, \cite{Wolfowitz}, \cite{Gal}, \cite{Poor}, and \cite{Elias}.

\begin{theorem}[General Converse Bound]\label{AWGN conv.}
	Consider $n$ independent uses of a channel with input alphabet $\mathcal{X}$ and output alphabet $\mathcal{Y}$. Let $P_{X}$, $P_{Y}$, and $P_{XY}$ be distributions on $\mathcal{X}$, $\mathcal{Y}$, and $\mathcal{X}\times \mathcal{Y}$, respectively, and suppose $P_{X}^n$, $P_{Y}^n$, and $P_{XY}^n$ are their product distributions over $n$ independent trials. Then the following holds:
	\begin{equation}
	M^*(n, p_e) \leq \sup_{P_{X}} \inf_{P_{Y}} \frac{1}{\beta_{1-p_e}\left(P_{XY}^n,P_{X}^nP_{Y}^n\right)}.
	\end{equation}
	where $\beta_{1-p_e}$ is defined as \eqref{beta def}.	
\end{theorem}
\begin{proof}
	See \cite[Theorem~29]{polyphd}.
\end{proof}

\begin{theorem}[Achievability Bound for AWGN]\label{AWGN ach.}
	There exists an $(M,p_e,\eta)$-code over AWGN$^n$($\eta$) such that
	\begin{equation}
	\log_2 M^*(n, p_e,\eta ) \ge nC -\sqrt{nV}\,Q^{-1}( p_e )+\mathcal{O}(1), 
	\end{equation}
	where
	\begin{align}
	C&= \frac 12\log_2\left(1+\eta\right),\\
	V&=\frac{     \eta(\eta+2) }{   2 (\eta+1)^2\ln^2 2  }.
	\end{align}
\end{theorem}
\begin{proof}
	See the achievability bound in \cite[Theorem~73]{polyphd}.
\end{proof}

\begin{lemma}\label{eta}
	Consider transmission over AWGN($\eta$) in low-capacity regime where by definition $\kappa=\frac n2\log_2 (1+\eta)$. Then we have
	\begin{equation}
	    n\eta= 2\,\kappa\ln 2+\mathcal{O}\left(\frac{\kappa^2}{n}\right).
	\end{equation}
\end{lemma}
\begin{proof}
	Consider $C=\frac 12\log(1+\eta)$. Thus, by using Taylor expansion of $\log(1+x)$ up to the second order, we arrive at
	\begin{equation}
	    C= \frac {1}{2\ln 2}\ln(1+\eta)= \frac {1}{2\ln 2}\left(\eta-\frac{\eta^2}{2}\right)+\mathcal{O}\left(\eta^3\right).
	\end{equation}
	which leads to solving $\eta^2-2\eta+4C\ln 2=0$. As a result,
	\begin{align}
	\eta=1-\sqrt{1-4C\ln 2}= 2\,C \ln 2+\mathcal{O}\left(C^2\right).
	\end{align}
	Now, considering $\kappa=nC$, results in
	\begin{equation}
	    n\eta=2\,nC \ln 2+\mathcal{O}\left(nC^2\right)=2\,\kappa\ln 2+\mathcal{O}\left(\frac{\kappa^2}{n}\right).
	\end{equation}
\end{proof}

\subsection{Proof of Theorem~\ref{bounds for AWGN}} \label{AWGN proofs1}
\begin{proof}[{\bf Converse Bound of Theorem~\ref{bounds for AWGN}}]
	Let $X$ and $Y$ be a uniform input and the corresponding output of an AWGN($\eta$) channel. Under the notation of Theorem~\ref{AWGN conv.}, define $P=P_{XY}$ and $Q=P_{X}P_{Y}$. Therefore, we have 
	\begin{equation}\label{mm}
	M^*(n, p_e) \leq  \frac{1}{\beta_{1-p_e}\left(P^n,Q^n\right)}.
	\end{equation}
	Also from Lemma~\ref{beta}, we obtain
	\begin{equation}\label{bb}
	\log_2\beta_{1-p_e}(P^n,Q^n)\ge -nD+\sqrt{nV}Q^{-1}\left(p_e+\gamma+\frac{B}{\sqrt{n}}\right)+\log_2\gamma,
	\end{equation}
	where $\gamma>0$ and in order to compute the quantities $D$, $V$ and $B$, consider the random variable $H_n = \log_2\frac{dP^n(z)}{dQ^n(z)}$. It can be verified that $H_n=\sum_{i=1}^nh_i$, where
	\begin{align}
	h_i=\frac 12 \log_2\left(1+\eta\right)+\frac{\eta}{2(1+\eta)\ln 2}-\frac{1}{2(1+\eta)\ln 2}\left(\eta Z_i^2-2\sqrt{\eta}\,Z_i\right),
	\end{align}
	assuming that $Z_i$'s are independent standard normal random variables.
	Thus, a simple calculation shows that for all $i\in\set{1,\dots,n}$, we have
	\begin{align}
	D&=\mathbb{E}[h_i]=\frac 12 \log_2\left(1+\eta\right),\\
	V&=\mathbb{V}ar[h_i]=\frac{\eta(\eta+2)}{2(\eta+1)^2\ln^22},\\
	T&=\mathbb{E}\left[\abs{h_i-D}^3\right],\\
	B&=\frac{T}{V^{\frac 32}}.
	\end{align}
	As we will see in the rest of the proof, computing $T$ and $B$ are not necessary as they will appear in terms that vanish compared to other terms. Now, by replacing the values computed above, \eqref{mm} together with \eqref{bb} yields
	\begin{align}
	\log_2 M^*(n,p_e) \leq \frac n2 \log_2\left(1+\eta\right)-\frac{    \sqrt{   n\eta(\eta+2)  }        }{      \sqrt{2}(\eta+1)\ln 2     }Q^{-1}\left(p_e+\gamma+\frac{B}{\sqrt{n}}\right)-\log_2\gamma.
	\end{align}
	Now, by setting $\gamma={p_e}/{\sqrt{\kappa}}$ and plugging in Lemma~\ref{eta}, and then using part (ii) of Lemma~\ref{Q} together with the definition of $\kappa$, we conclude that
	\begin{align}
	\log_2 M^*(n,p_e) \leq \kappa-\frac{    \sqrt{   \kappa(\eta+2)  }        }{      (\eta+1)\sqrt{\ln 2}     }Q^{-1}(p_e+\frac{p_e}{\sqrt{\kappa}}+\frac{B}{\sqrt{n}})-\log_2\frac{p_e}{\sqrt{\kappa}}+\mathcal{O}(\frac{\kappa^2\sqrt{-\log p_e}}{n}).
	\end{align}
	Note that 
	\begin{equation}
	  \frac{d}{dx}Q^{-1}(x)=-\sqrt{2\pi}\,e^{\frac{Q^{-1}(x)^2}{2}}.  
	\end{equation}
	Therefore, by applying the Taylor expansion of $Q^{-1}(\cdot)$, and using parts (i) and (ii) of Lemma~\ref{Q}, we achieve
	\begin{align}
	\nonumber \log_2M^*(n,p_e) \le \kappa &-\frac{    \sqrt{   \eta+2  }  \sqrt{\kappa}      }{      (\eta+1)\sqrt{\ln 2}   }Q^{-1}\left(p_e\right)-\log_2\frac{p_e}{\sqrt{\kappa}}\\
 &\qquad\qquad+\mathcal{O}\left(\frac{\kappa^2\sqrt{-\log_2p_e}}{n}\right)+\mathcal{O}\left(\frac{1}{\sqrt{-\log p_e}}\right).
	\end{align}
	Considering the dominant terms then results in
	\begin{align}
	\log_2M^*(n,p_e) &\le \kappa-\frac{    \sqrt{   \eta+2  }   \sqrt{\kappa}     }{      (\eta+1)\sqrt{\ln 2} } Q^{-1}\left(p_e\right)+\frac{1}{2}\log_2\kappa-\log_2p_e+\mathcal{O}\left(\frac{1}{\sqrt{-\log p_e}}\right).
	\end{align}
\end{proof}

\begin{proof}[{\bf Achievability Bound of Theorem~\ref{bounds for AWGN}}]
	The proof of Theorem~\ref{AWGN ach.} is still valid in the low-capacity regime. As a result, using Theorem~\ref{AWGN ach.} together with Lemma~\ref{eta} and part (ii) of Lemma~\ref{Q} results in
	\begin{align}
	\nonumber \log_2 M^*(n, p_e,\eta ) &\ge \kappa -\frac{\sqrt{\eta+2}}{(\eta+1)\sqrt{\ln2}}\cdot \sqrt{\kappa}\,Q^{-1}( p_e )+\mathcal{O}(1)+\mathcal{O}\left(\frac{\kappa\sqrt{-\kappa\log p_e}}{n}\right)\\
	&=\kappa -\frac{\sqrt{\eta+2}}{(\eta+1)\sqrt{\ln2}}\cdot \sqrt{\kappa}\,Q^{-1}( p_e )+\mathcal{O}(1).
	\end{align}
\end{proof}

\subsection{Proof of Corollary~\ref{optimal n for AWGN}}\label{AWGN proofs2}
\begin{proof}
	In order to obtain the optimal blocklength $n^*$ for transmission of $k$ information bits over a low capacity AWGN($\eta$), it suffices to replace $M^*(n^*,p_e,\eta)=k$. Then $n^*$ can be computed by solving \eqref{AWGN b}. Substitute $M^*(n^*,p_e)=k$ and $\kappa=n^*C$ in \eqref{AWGN b}, where $C=\frac 12\log_2\left(1+\eta\right)$, to obtain
	\begin{align}
	k = n^*C-\frac{\sqrt{\eta+2}}{(\eta+1)\sqrt{\ln2}}\cdot\sqrt{n^*C}\,Q^{-1}\left(p_e\right)+\mathcal{E}.
	\end{align}
	Define $x=\sqrt{n^*C}$, $a=\frac{\sqrt{\eta+2}}{2(\eta+1)\sqrt{\ln2}}Q^{-1}\left(p_e\right)$ and $b=k-\mathcal{E}$. Thus, we have $x^2-2ax-b=0$. This results in $x= a+\sqrt{a^2+b}$. Therefore, $\sqrt{n^*C}= a+\sqrt{a^2+b}$.
	More simplifications are as follows:
	\begin{align}
	\nonumber n^*&= \frac1C \left(a+\sqrt{a^2+b}\right)^2\\
	\nonumber &=\frac 1C \left(      2a^2+b+   2a   \left(     \sqrt{k}+    \mathcal{O}\left(    \frac{-\log p_e}{k}  \right)    +    \mathcal{O}\left(   \frac{\log \kappa}{k}  \right)   \right)\right)\\
	\label{nnd}&=\frac 1C \left(k+2a\sqrt{k}+2a^2-\mathcal{E}+\mathcal{O}\left(    \frac{(-\log p_e)^{3/2}}{k}  \right) \right).
	\end{align}
	Note that in the low-capacity regime, under optimal blocklength, $\log\kappa \approx \log k$. Therefore, by substituting the values of $C$ and $a$ in \eqref{nnd} and comparing the orders, we obtain
	\begin{equation}
	    n^* =\frac 2{\log_2(1+\eta)} \left(k+\frac{\sqrt{\eta+2}\sqrt{k}}{(\eta+1)\sqrt{\ln2}}Q^{-1}( p_e )+\frac{\eta+2}{2(\eta+1)^2\ln2}Q^{-1}(p_e)^2+\mathcal{O}\left(\log_2 \frac{1}{p_e}\right)\right).
	\end{equation}
	Applying part (ii) of Lemma~\ref{Q} then results in
	\begin{equation}
	    n^* =\frac 2{\log_2(1+\eta)} \left(k+\frac{\sqrt{\eta+2}}{(\eta+1)\sqrt{\ln2}}Q^{-1}( p_e )\cdot \sqrt{k}+\mathcal{O}\left(\log_2 \frac{1}{p_e}\right)\right).
	\end{equation}
\end{proof}

\section{Proofs for the Coding Part}\label{rpp}
Below we state a lemma that will be used in the proof of Theorem~\ref{r}.
\begin{lemma}
	\label{inequality}
	\begin{enumerate}[i)]
		\item If \,$z\geq 0$,\quad then \,$1-z \leq e^{-z} \le 1-z+\frac{z^2}{2}$.
		\item If \,$0<z<1$,\quad then \,$\ln(1-z) \le -z-\frac{z^2}{2}$.
	\end{enumerate}
\end{lemma}
\begin{proof}
	\begin{enumerate}[i)]
		\item Define $f(z)= e^{-z} -1+z$.  Note that $f'(z)=-e^{-z}+1 \ge 0$ for $z\ge0$. This means $f$ is increasing over $z\ge0$. Thus, $f(z)\ge f(0)=0$ which proves the lower bound.\\
		Define $g(z)= e^{-z} -1+z-\frac{z^2}{2}$.  Note that $g'(z)=-e^{-z}+1-z\le 0$ for $z\ge0$ due to the lower bound proved above. This means $g$ is decreasing over $z\ge0$. Thus, $g(z)\le g(0)=0$ which proves the upper bound.
		\item $\ln(1-z) = -\sum_{i=1}^\infty\frac{z^i}{i}\le -z-\frac{z^2}{2}$ for $0<z<1$.
	\end{enumerate}
\end{proof}

\subsection{Proof of Theorem~\ref{r}}\label{rpp1}
\begin{proof}
	Let $m$ be the number of repeated blocks of size $r$, i.e., $m=\frac nr$ and consequently, $m_{\beta}=\frac n{r_{\beta}}$. Note that the maximum achievable rate in this setting is $\frac{1-\epsilon^{r}}{r}=\frac mn\left(1-\epsilon^{\frac{n}{m}}\right)$. Thus, $m_\beta$ is the solution of the following problem:
	\begin{equation}
	\begin{aligned}
	& \text{minimize}
	& & m \\
	& \text{subject to}
	& & m\left(1-\epsilon^{\frac{n}{m}}\right) \geq \beta n(1-\epsilon).
	\end{aligned}
	\end{equation}
	It can be easily verified that the answer to this problem $m_\beta$ indeed satisfies 
	\begin{equation}
	m_\beta\left(1-\epsilon^{\frac{n}{m_\beta}}\right) = \beta n(1-\epsilon).
	\end{equation}
	Now putting $\kappa=n(1-\epsilon)$ gives
	\begin{equation}
	m_\beta\left(1-\epsilon^{\frac{\kappa}{(1-\epsilon)m_\beta}}\right) = \beta \kappa.
	\end{equation}
	Define $x = \frac{\kappa}{m_\beta}$ to get
	\begin{equation}
	1-\epsilon^{\frac{x}{1-\epsilon}} = \beta x.
	\end{equation}
	Therefore,
	\begin{equation}
	1-  \beta x =\epsilon^{\frac{x}{1-\epsilon}}= e^{\frac{\ln \epsilon}{1-\epsilon}\,x}=e^{-x\ell},
	\end{equation}
	where $\ell=-\frac{\ln \epsilon}{1-\epsilon}$. Now, let $z=x\ell$. Thus,
	\begin{equation}
	1- \frac{ \beta}{\ell} z = e^{-z}.
	\end{equation}
	Finally, define $\gamma = \frac{ \beta}{\ell}$ to get
	\begin{equation}
	\label{equa}
	1- \gamma z = e^{-z}.
	\end{equation}
	Note that $z=x\ell \ge 0$. Thus we can use Lemma~\ref{inequality}. Due to part (ii) of Lemma~\ref{inequality}, we can write
	\begin{align}
	1- \gamma z = e^{-z} \le 1-z+\frac{z^2}{2} \quad \Longrightarrow \quad 2(1- \gamma ) \le z \label{e1}.
	\end{align}
	Taking logarithm of both sides of \eqref{equa}, gives
	\begin{equation}
	-z=\ln(1- \gamma z) .
	\end{equation}
	Now use part (iii) of Lemma~\ref{inequality}  to obtain
	\begin{equation}
	-z=\ln(1- \gamma z) \le -\gamma z-\frac{(\gamma z)^2}{2}.
	\end{equation}
	Therefore,
	\begin{equation}
	\label{e2}
	z \le 2(1-\gamma)\frac{1}{\gamma^2}.
	\end{equation}
	Note that \eqref{e1} together with \eqref{e2} yields
	\begin{equation}
	2(1-\gamma)\le z \le 2(1-\gamma)\frac{1}{\gamma^2}.
	\end{equation}
	Remember $z = x\ell = {\kappa\ell}/{(m_\beta)}$ . Hence,
	\begin{equation}
	\frac{\kappa\ell}{2(1-\gamma)}\cdot\gamma^2\le m_\beta \le \frac{\kappa\ell}{2(1-\gamma)}.
	\end{equation}
	Now, replacing $\gamma = {\beta}/{\ell}$, $\kappa=n(1-\epsilon)$, and $m_\beta={n}/{r_\beta}$ result in
	\begin{equation}
	\frac{n(1-\epsilon)\ell}{2\left(1-\frac{\beta}{\ell}\right)}\cdot\left(\frac{\beta}{\ell}\right)^2\le \frac{n}{r_\beta}\le \frac{n(1-\epsilon)\ell}{2\left(1-\frac{\beta}{\ell}\right)}.
	\end{equation}
\end{proof}        

\subsection{Proof of Theorem~\ref{r_bms}}\label{rpp2}
\begin{proof}
	We use extremes of information combining. Consider two BMS channels $W_1, W_2$ with capacity $C_1, C_2$ respectively. Note that $\text{BEC}(1-C_1)$ has capacity $C_1$ and $\text{BEC}(1-C_2)$ has capacity $C_2$. Then we know from extremes of information combining \cite[Chapter~4]{richardson2008modern} that 
	\begin{equation} \label{extremes}
	C(W_1 \circledast W_2) \leq  C\left(\text{BEC}(1-C_1) \circledast \text{BEC}(1-C_2)\right),
	\end{equation}
	where $W_1 \circledast W_2$ is the BMS channel whose output is formed as the union of the output of $W_1$ and $W_2$, i.e., we send the input bit once through $W_1$ and once through $W_2$ and the resulting outcomes together will be the outcome of $W_1 \circledast W_2$. 
	
	Now, it is clear that for any BMS channel $W$ we have  
	\begin{equation}
	    W^r = \underbrace{W \circledast \cdots \circledast W}_{r \text{ times}}.
	\end{equation}  
	Thus, assuming $C(W) = C$, then by using \eqref{extremes} $r$ times we obtain $C(W^r) \leq C( \text{BEC}(1-C)^r)$. 
\end{proof}

\subsection{Proof of Theorem~\ref{polar_repetition}}\label{proofpolar}
\begin{proof}
Let $i \in [n]$ be an arbitrary index. Recall that in order to construct the $i$-th sub-channel for a polar code of length $n=2^m$ on channel $W$, which is denoted by $W_n^{(i)}$, 
we proceed as follows: (i) Consider the binary expansion $i = b_1b_2\cdots b_m$. (ii) Start with $W_0 = W$. (iii) For $j\in\set{1,\dots,m}$, let $W_{j} = W_{j-1} \circledast W_{j-1}  $ if $b_j = 1$, and otherwise,  let $W_{j} = W_{j-1} \boxast W_{j-1}$. (iv) The channel $W_m$ is the sub-channel corresponding to the $i$-th index.  Also recall that for any BMS channel $W$, we have (see \cite[Lemma~3.16]{korada2009polar}, \cite{hassani2013polarization}) 
\begin{equation}
    Z(W \circledast W) = Z(W)^2 \,\,\,\text{ and }\,\,\, Z(W\boxast W) \geq \sqrt{1-(1-(1-Z(W))^2)^2},
\end{equation}
which by simple manipulations will be simplified to
\begin{equation} \label{z_bounds}
1 - Z(W \circledast W) \leq 2 (1 - Z(W)) \,\,\,\text{ and }\,\,\,  1 - Z(W\boxast W) \leq  4*(1-Z(W))^2.
\end{equation}
Now, for an integer $t \leq m$ let $i_t $ be such that in the binary expansion $i_t = b_1b_2\cdots b_m$, all the bits $b_j$ are equal to $1$ except for $b_t$ which is $0$ (i.e., $i_t = 2^m - 2^{m-t}-1$).   Using the bounds in \eqref{z_bounds}, we can write
\begin{equation}
    1 - Z(W_n^{(i_t)})  \leq 2^{m-t+2}\left( 2^{t-1}(1-Z(W)) \right)^2.
\end{equation}
Thus, if  $i_t$ is a good sub-channel, then we must have $Z(W_m^{(i_t)}) \leq \frac 12$ which by using the above inequality gives 
\begin{equation}
    2^{m-t+2}\left( 2^{t-1}(1-Z(W)) \right)^2 \geq \frac 12 \quad \Longrightarrow \quad 2^t \geq \frac{1}{2^{m+1} (1-Z(W))^2}. 
\end{equation}
Note also that for any BMS channel, we have $C(W) \geq 1- Z(W) $, and thus the above inequality implies
\begin{equation} \label{t_bound}
2^t \geq \frac{1}{2^{m+1}C(W)^2}\quad \Longrightarrow\quad 2^{m-t+1} \leq 4(2^{m} C(W))^2 = 4(n(c(W)))^2.
\end{equation}
Now, recalling the fact that the binary expansion of index $i_t$ has only one position with zero value (the $t$-th position), we can conclude that for any other good index $i$, we must have the following: The first $t-1$ bits of the binary expansion of $i-1$ should be $1$, and $t$ is lower bounded from \eqref{t_bound}. This means that the polar code corresponding to $W$ will have at least $2^{t-1}$ repetitions in the beginning. 
\end{proof}

\end{document}